\documentclass[11pt]{amsart}
\usepackage[foot]{amsaddr}
\usepackage{ifxetex}
\ifxetex
\usepackage[no-math]{fontspec}
\else
\fi
\usepackage{amsmath}
\usepackage{amsfonts}

\usepackage{amssymb}
\usepackage{bm}
\usepackage{amsthm}
\usepackage{fullpage}
\usepackage[lining,semibold,type1]{libertine} 
\usepackage[T1]{fontenc} 
\usepackage{textcomp} 
\usepackage[varqu,varl]{inconsolata}
\usepackage[libertine,vvarbb]{newtxmath}
\usepackage[scr=rsfso]{mathalfa}
\usepackage{bm}
\usepackage{dsfont}

\usepackage{listings}
\usepackage{hyperref, color}
\hypersetup{colorlinks=true,citecolor=blue, linkcolor=blue, urlcolor=blue}
\usepackage[linesnumbered,boxed,ruled,vlined]{algorithm2e}
\usepackage{bm}
\usepackage{bbm}
\usepackage[numbers]{natbib}
\usepackage{xcolor}
\usepackage{enumerate} 
\usepackage{enumitem}
\usepackage{tabularx}
\usepackage{array}
\usepackage{cleveref}
\usepackage{mathrsfs}

\newcolumntype{L}[1]{>{\raggedright\arraybackslash}p{#1}}
\newcolumntype{C}[1]{>{\centering\arraybackslash}m{#1}}
\newcolumntype{R}[1]{>{\raggedleft\arraybackslash}p{#1}}
  
\usepackage{makecell}

\usepackage{footnote}
\makesavenoteenv{tabular}

\newcommand{\TODO}[1]{\typeout{TODO: \the\inputlineno: #1}\textbf{{\color{red}[[[ #1 ]]]}}}

\newcommand{\DTV}[2]{d_{\mathrm{TV}}\left({#1},{#2}\right)}

\newcommand{\e}{\mathrm{e}}

\renewcommand{\epsilon}{\varepsilon}

\newtheorem{theorem}{Theorem}[section]
\newtheorem{observation}[theorem]{Observation}
\newtheorem{claim}[theorem]{Claim}
\newtheorem*{claim*}{Claim}
\newtheorem{condition}[theorem]{Condition}

\newtheorem{fact}[theorem]{Fact}
\newtheorem{lemma}[theorem]{Lemma}
\newtheorem{proposition}[theorem]{Proposition}
\newtheorem{corollary}[theorem]{Corollary}
\theoremstyle{definition}

\newtheorem{definition}[theorem]{Definition}
\newtheorem{remark}[theorem]{Remark}
\newtheorem*{remark*}{Remark}

\renewcommand{\emptyset}{\varnothing}

\newcommand{\tuple}[1]{\left(#1\right)} \newcommand{\eps}{\varepsilon}
\newcommand{\inner}[2]{\left\langle #1,#2\right\rangle}
\newcommand{\tp}{\tuple}

\newcommand{\abs}[1]{\left\vert#1\right\vert}
\newcommand{\ctp}[1]{\left\lceil#1\right\rceil}

\newcommand{\0}{-1}
\newcommand{\1}{+1}

\def\*#1{\boldsymbol{#1}} 
\def\+#1{\mathcal{#1}} 
\def\-#1{\mathrm{#1}} 
\def\^#1{\mathscr{#1}} 

\usepackage{todonotes}

\usepackage{xifthen}

\DeclareMathOperator*{\oPr}{\mathbf{Pr}}
\renewcommand{\Pr}[2][]{ \ifthenelse{\isempty{#1}}
  {\oPr\left[#2\right]}
  {\oPr_{#1}\left[#2\right]} } 

\DeclareMathOperator*{\oE}{\mathds{E}}
\newcommand{\E}[2][]{ \ifthenelse{\isempty{#1}}
  {\oE\left[#2\right]}
  {\oE_{#1}\left[#2\right]} }

\DeclareMathOperator*{\oVar}{\mathbf{Var}}
\newcommand{\Var}[2][]{ \ifthenelse{\isempty{#1}}
  {\oVar\left[#2\right]}
  {\oVar_{#1}\left[#2\right]} }

\def\oEnt{\mathbf{Ent}}
\newcommand{\Ent}[2][]{ \ifthenelse{\isempty{#1}}
  {\oEnt\left[#2\right]}
  {\oEnt_{#1}\left[#2\right]} }

\newcommand{\EE}[2][]{ \ifthenelse{\isempty{#1}}
  {\mathbf{E}\left[#2\right]}
  {\mathbf{E}_{#1}\left[#2\right]} }
  
\newcommand{\Rd}{\mathsf{Trans}}
\newcommand{\relaxT}[2][]{
  \ifthenelse{\isempty{#2}}
  {t_{\mathrm{rel}}^{\mathrm{#1}}}
  {t_{\mathrm{rel}}^{\mathrm{#1}}(#2)}
}

\newcommand{\KL}[2]{D_{\mathrm{KL}}\tp{{#1}\,\Vert\, {#2}}}

\title{Optimal mixing for two-state anti-ferromagnetic spin systems}
\date{}
\author{Xiaoyu Chen}
\author{Weiming Feng}
\author{Yitong Yin}
\author{Xinyuan Zhang}

\address[Xiaoyu Chen, Yitong Yin, Xinyuan Zhang]{State Key Laboratory for Novel Software Technology, Nanjing University, 163 Xianlin Avenue, Nanjing, Jiangsu Province, China. \textnormal{E-mails: \url{chenxiaoyu233@smail.nju.edu.cn}, \url{yinyt@nju.edu.cn}, \url{zhangxy@smail.nju.edu.cn}}}
\address[Weiming Feng]{School of Informatics, University of Edinburgh, Informatics Forum, Edinburgh, EH8 9AB, United Kingdom. \textnormal{E-mail: \url{wfeng@ed.ac.uk}}}

\begin{document}
\begin{abstract}
We prove an optimal $\Omega\tp{n^{-1}}$ lower bound for modified log-Sobolev (MLS) constant of the Glauber dynamics for anti-ferromagnetic two-spin systems with $n$ vertices in the tree uniqueness regime.
Specifically, this optimal MLS bound holds for the following classes of two-spin systems in the tree uniqueness regime:
\begin{itemize}
  \item all \emph{strictly} anti-ferromagnetic two-spin systems (where both edge parameters $\beta,\gamma<1$), 
  which cover the hardcore models and the anti-ferromagnetic Ising models;
  \item general anti-ferromagnetic two-spin systems on regular graphs.
\end{itemize}
Consequently, an optimal $O(n\log n)$ mixing time  holds for these anti-ferromagnetic two-spin systems when  the uniqueness condition is satisfied.
These MLS and mixing time bounds hold for any bounded or unbounded maximum degree, and the constant factors in the bounds depend only on the gap to the uniqueness threshold.
We prove this by showing a boosting theorem for MLS constant for distributions satisfying certain spectral independence and marginal stability properties.
\end{abstract}
\maketitle

\setcounter{tocdepth}{1}
\tableofcontents

\clearpage

\section{Introduction}
Two-state spin systems, or \emph{two-spin systems}, are canonical graphical models arising from pairwise constrained Boolean variables.
A two-spin system is specified on an undirected graph $G=(V,E)$ by three parameters $\beta, \gamma, \lambda \geq 0$, where the two edge parameters $\beta$ and $\gamma$ specify the \emph{edge activities}, and the vertex parameter $\lambda$ specifies the \emph{external field}. 
A \emph{configuration} $\sigma \in \{\0, \1\}^V$ assigns each vertex $v \in V$ a $\pm 1$-\emph{spin}.
This defines a \emph{Gibbs distribution} $\mu$ over all the configurations $\sigma \in \{\0, \1\}^V$ by:
\begin{align*}
\forall \sigma\in \{\0, \1\}^V,\qquad  \mu(\sigma) &\triangleq \frac{1}{Z} \beta^{m_+(\sigma)} \gamma^{m_-(\sigma)} \lambda^{n_+(\sigma)},
\end{align*}
where $m_{\pm}(\sigma) \triangleq \abs{\{\{u, v\} \in E \mid \sigma_u = \sigma_v = \pm 1\}}$ denotes the number of $\pm 1$-monochromatic edges in $\sigma$, 
$n_{+}(\sigma) \triangleq \abs{\{v \in V \mid \sigma_v = \1\}}$ denotes the number of  vertices assigned with $\1$-spin in $\sigma$, 
and the normalizing factor, known as the \emph{partition function}, is given by:
\begin{align*}
  Z &\triangleq \sum_{\sigma \in \{\0, \1\}^V} \beta^{m_+(\sigma)} \gamma^{m_-(\sigma)} \lambda^{n_+(\sigma)}.
\end{align*}

The hardcore models and the Ising models are two classes of extensively studied two-spin systems. 
\begin{itemize}
\item \emph{Hardcore model with fugacity $\lambda$}: 
a two-spin system with $\beta = 0$ and $\gamma = 1$; 
\item \emph{Ising model with temperature $\beta$ and external field $\lambda$}: 
a two-spin system with $\beta = \gamma$. 
\end{itemize}

A two-spin system is called \emph{ferromagnetic} if $\beta\gamma > 1$ and \emph{anti-ferromagnetic} if $\beta\gamma < 1$.
The hardcore models are anti-ferromagnetic. 
An Ising model is ferromagnetic if $\beta>1$ and anti-ferromagnetic if $\beta<1$.

The \emph{Glauber dynamics} (a.k.a \emph{heat bath, Gibbs sampling}) is a canonical Markov chain for sampling from the Gibbs distribution $\mu$.
Let $\Omega(\mu)$ denote the support of $\mu$.
The chain is defined on space $\Omega(\mu)$ as:
\begin{itemize}
\item to move from the current state $\sigma\in\Omega(\mu)$, pick a vertex $v \in V$ uniformly at random;
\item and replace the spin $\sigma_v$ with a random spin according to the marginal distribution $\mu_v^{\sigma_{V \setminus \{v\}}}$.
\end{itemize}
This chain is reversible and stationary at $\mu$~\cite{levin2017markov}. 
The \emph{mixing time} of a chain $(X_t)_{t \geq 0}$ is defined by:
\begin{align*}
  \forall 0 < \epsilon < 1, \quad T_{\-{mix}}(\epsilon) &\triangleq \max_{X_0 \in \Omega(\mu)} \min \{ t \mid \DTV{X_t}{\mu} \leq \epsilon\},
\end{align*}
where $\DTV{X_t}{\mu}$ denotes the total variation distance between the distribution of $X_t$ and $\mu$.

The \emph{modified log-Sobolev (MLS) constant}~\cite{bobkov2006modified}  plays an important role in tight analysis of mixing times. 
Let $P :\Omega(\mu) \times \Omega(\mu) \to \mathds{R}_{\geq 0}$ denote the transition matrix of the Glauber dynamics on $\mu$.
For any function $f: \Omega(\mu)\to \mathds{R}_{\geq 0}$, the \emph{Dirichlet form} is defined by:
\begin{align*}
\+E_{P}(f, \log f) \triangleq \inner{f}{(I-P)\log f}_{\mu}, 
\end{align*}
where the inner product $\inner{f}{g}_{\mu}\triangleq \sum_{\sigma \in \Omega(\mu)}f(\sigma)g(\sigma)\mu(\sigma)$.
And define  the entropy:
\begin{align*}
\Ent[\mu]{f} \triangleq \EE[\mu]{f \log f} - \EE[\mu]{f} \log \EE[\mu]{f},	
\end{align*}
where $\EE[\mu]{f} \triangleq \sum_{\sigma \in \Omega(\mu)}\mu(\sigma)f(\sigma)$.
In above definitions, we assume $0\log 0 = 0$.

The modified log-Sobolev constant for the Glauber dynamics on $\mu$ is given by:
\begin{align} \label{eq:mlsc} 
  \rho^{\mathrm{GD}}(\mu) &\triangleq \inf \left\{ \left.\frac{\+E_{P}(f, \log f)}{\Ent[\mu]{f}} \;\right\vert\; f:\Omega(\mu) \to \mathds{R}_{\geq 0},\; \Ent[\mu]{f} \not= 0 \right\}.
\end{align}
It bounds the mixing time of Glauber dynamics as follows:
Denote $\mu_{\min} \triangleq \min_{\sigma \in \Omega(\mu)} \mu(\sigma)$, and
\begin{align*}
  T_{\mathrm{mix}}(\epsilon) \leq \frac{1}{\rho^{\mathrm{GD}}(\mu)} \tp{\log \log \frac{1}{\mu_{\min}} + \log \frac{1}{2\epsilon^2}}.
\end{align*}
Proving mixing time upper bound is reduced to establishing the {modified log-Sobolev inequality (MLSI)} that lower bounds the MLS constant. 
However, this task used to be notoriously difficult, especially when 
the maximum degree of the model is unbounded
and no marginal probability lower bound is assumed.

\subsection{Results for  two-spin  systems}
We prove an $\mathrm{e}^{-O(1/\delta)}n^{-1}$ lower bound for the MLS constant for Glauber dynamics  
for the anti-ferromagnetic two-spin systems with $n$ vertices in the tree uniqueness regime with a slack $\delta \in (0, 1)$.
This MLS bound is asymptotically optimal in $n$ and implies an optimal $O(n \log n)$ mixing time for the Glauber dynamics when the uniqueness condition is satisfied with a constant gap $\delta$.

Consider two-spin systems on graph $G = (V, E)$ with parameters $(\beta,\gamma,\lambda)$.
By symmetry, we can assume:
\begin{align}
0\leq \beta \leq \gamma, \quad \gamma > 0 \quad\text{and}\quad\lambda > 0. \label{eq:2-spin-parameters-WLOG}
\end{align} 
A tuple $(\beta,\gamma,\lambda)$ is called {anti-ferromagnetic} if it further satisfies $\beta\gamma < 1$ in addition to this.

The following uniqueness condition for anti-ferromagnetic two-spin system was characterized in \cite{LLY13}. 
\begin{definition}
\label{definition-d-uniqueness}
Let $d \geq 1$ be an integer. An anti-ferromagnetic $(\beta,\gamma,\lambda)$ is  \emph{$d$-unique with gap $\delta\in(0,1)$} if 
\begin{align}
		\abs{F'_d(\hat{x}_d)} = \frac{d(1-\beta\gamma)\hat{x}_d}{(\beta \hat{x}_d + 1)(\hat{x}_d+\gamma)}\leq 1-\delta, \text{ where $\hat{x}_d$ is the unique fixed point of } F_d(x) = \lambda \tp{\frac{\beta x + 1 }{x + \gamma}}^d. \label{eq:tree-recursion-decay}
	\end{align}	
\end{definition}
The property of being $d$-unique corresponds to the uniqueness of Gibbs measure on $(d+1)$-regular tree.
It was well known that sampling in anti-ferromagnetic two-spin systems on $\Delta$-regular graphs is intractable if $(\beta, \gamma, \lambda)$ is not $(\Delta-1)$-unique \cite{sly2012computational,galanis2015inapproximability}.
We consider the following criterion for two-spin systems.

\begin{condition}[uniqueness criterion]\label{condition-canonical-two-spin}
Let $\delta\in(0,1)$. 
The  anti-ferromagnetic  two-spin system specified by $(\beta,\gamma,\lambda)$ on graph $G=(V,E)$ 
with maximum degree $\Delta\ge 3$
satisfies:
\begin{itemize}
\item when $\gamma \leq 1$:  $(\beta, \gamma, \lambda)$ is $(\Delta-1)$-unique with gap $\delta$; 
\item when $\gamma > 1$:  $(\beta, \gamma, \lambda)$ is $(\Delta-1)$-unique with gap $\delta$ and $G$ is $\Delta$-regular. 
\end{itemize}
\end{condition}

For the classes of anti-ferromagnetic two-spin systems satisfying such uniqueness criterion,
we show the following optimal bounds on the MLS constant and the mixing time of Glauber dynamics.

\begin{theorem}[main theorem: two-spin systems]\label{thm:2spin-theorem}
Let $\delta \in (0,1)$. 
There exists a $C(\delta) = \exp(O(1/\delta))$ such that 
for every anti-ferromagnetic two-spin system with $n$ vertices
that satisfies \Cref{condition-canonical-two-spin} with gap $\delta$, 
the modified log-Sobolev constant $\rho^{\mathrm{GD}}$ of the Glauber dynamics satisfies
  \begin{align*}
    \rho^{\mathrm{GD}} \ge \frac{1}{C(\delta) n}.
  \end{align*}
  Consequently, the mixing time of the Glauber dynamics  is bounded as
  \begin{align*}
    T_{\mathrm{mix}}(\epsilon) \le C(\delta) n \tp{2\log n + \log \log \tp{\alpha} + \log \log \tp{\lambda + \lambda^{-1}} + \log \frac{1}{2\epsilon^2}},
  \end{align*}
  where 
  $\alpha=\begin{cases}
  \gamma+\gamma^{-1}+2 & \text{if }\beta = 0\\
  \beta^{-1} +2 & \text{if }\beta > 0
  \end{cases}$.
\end{theorem}
Due to the hardness results in \cite{sly2012computational,galanis2015inapproximability}, 
\Cref{thm:2spin-theorem} gives sharp computational phase transitions,
since sampling in  not-$(\Delta-1)$-unique $\Delta$-regular anti-ferromagnetic two-spin systems  is intractable.

\begin{remark}[comparison to the up-to-$\Delta$-uniqueness]
The uniqueness condition (\Cref{condition-canonical-two-spin}) assumed by \Cref{thm:2spin-theorem} slightly deviates from the \emph{up-to-$\Delta$-uniqueness} (i.e.~$d$-unique for all $1\le d\le \Delta-1$) assumed in e.g.~\cite{LLY13,chen2020rapid,chen2020optimal,chen2021rapid} for spin systems with $\Delta$-bounded maximum degree . 

It is known that  $|F'_d(\hat{x}_d)|$ in \eqref{eq:tree-recursion-decay} is monotonically increasing in $d$ if and only if $\gamma\le1$  (\Cref{prop:equiv-LLY}). 
Therefore, when $\gamma \leq 1$, being $(\Delta-1)$-unique immediately implies the up-to-$\Delta$-uniqueness;
and in contrast when  $\gamma > 1$, the property of being $d$-unique may no longer be monotone in $d$.
And hence:
\begin{itemize}
\item
Case $(\gamma \leq 1)$: the uniqueness condition assumed by \Cref{thm:2spin-theorem}  is the same as the up-to-$\Delta$-uniqueness on instances with $\Delta$-bounded max-degree, as in~\cite{LLY13,chen2020rapid,chen2020optimal,chen2021rapid};
\item
Case $(\gamma > 1)$: 
the uniqueness condition assumed by \Cref{thm:2spin-theorem}  is  restricted to the regular graphs, 
but it can give strictly broader regime than the up-to-$\Delta$-uniqueness.
\end{itemize}
To the best of our knowledge, this is the first time that a strictly  stronger algorithmic result is obtained on regular graphs than general graphs,
for anti-ferromagnetic two-spin systems.
\end{remark}

%

Both the hardcore and anti-ferromagnetic Ising models fall into the \emph{strictly} anti-ferromagnetic case where $\gamma\le1$.
Hence the following corollaries hold, whose formal proofs are given in \Cref{sec:append}.

\begin{corollary}[hardcore model]\label{cor:hardcore}
  Let $\delta \in (0,1)$. There exists a $C(\delta) = \exp(O(1/\delta))$ such that for every hardcore model on $n$-vertex graph $G=(V,E)$ with maximum degree $\Delta \ge 3$ and fugacity $\lambda \le (1-\delta) \lambda_c(\Delta) = (1-\delta)\frac{(\Delta-1)^{\Delta-1}}{(\Delta-2)^\Delta}$, 
the mixing time of the Glauber dynamics is bounded as
  \begin{align*}
   T_{\mathrm{mix}}(\epsilon) \le C(\delta)n\tp{2\log n + \log \frac{1}{2\epsilon^2}}.
  \end{align*}
\end{corollary}

  \begin{corollary}[anti-ferromagnetic Ising model]\label{cor:Ising}
   Let $\delta \in (0,1)$. There exists a $C(\delta) = \exp(O(1/\delta))$ such that for every 
  anti-ferromagnetic Ising model with temperature $\beta\in(0,1)$ and external field $\lambda>0$ on $n$-vertex graph $G=(V,E)$ with maximum degree $\Delta \ge 3$ that satisfies either one of the followings: 
  \begin{itemize}
  \item $\beta \geq \frac{\Delta - 2 + \delta}{\Delta - \delta}$; 
  \item $\beta < \frac{\Delta - 2 + \delta}{\Delta - \delta}$ and $\lambda \in (0,\lambda_c] \cup [\bar{\lambda}_c,\infty)$, where  $\lambda_c = \lambda_c(\delta, \beta)$ and $\bar{\lambda}_c = \bar{\lambda}_c(\delta,\beta)$ that satisfy  $\lambda_c\le \bar{\lambda}_c$ and $\lambda_c\bar{\lambda}_c=1$, are the critical thresholds for $\lambda$ in anti-ferromagnetic Ising model \cite{LLY13,sinclair2014approximation};
\end{itemize}
the mixing time of Glauber dynamics is bounded as
\begin{align*}
    T_{\mathrm{mix}}(\epsilon) \le C(\delta) n \tp{2\log n + \log \log \tp{\beta^{-1} + 3} + \log \log \tp{\lambda + \lambda^{-1}} + \log \frac{1}{2\epsilon^2}}.
  \end{align*}
  \end{corollary}


Note that the Ising uniqueness regime in \Cref{cor:Ising} is much broader than the  regime $\beta \in[\frac{\Delta - 2 + \delta}{\Delta - \delta},1)$  assumed in \cite{chen2020rapid, chen2020optimal, chen2021optimalIsing, anari2021entropicII}  for the anti-ferromagnetic case, 
which corresponds to the uniqueness regime for all external fields $\lambda$.
In fact, before this work, proving optimal mixing times for $\lambda$-dependent uniqueness regimes was a major challenge to the current techniques \cite{anari2021entropicII}.

The modified log-Sobolev inequalities (MLSI) are very powerful.
For example, 
by the Herbst argument (e.g.~\cite[Lemma 15]{CGM19}), 
the MLSI in \Cref{thm:2spin-theorem} also implies the following concentration bound.
\begin{corollary} \label{cor:2-spin-concentration}
Let $\delta \in (0,1)$. 
There exists a $C(\delta) = \exp(O(1/\delta))$ such that 
for every anti-ferromagnetic two-spin system with $n$ vertices, 
if \Cref{condition-canonical-two-spin} is satisfied with gap $\delta$, 
then it holds for the Gibbs distribution $\mu$ and 
 for any observable function $f: \Omega(\mu) \to \mathds{R}$ and any $\alpha \geq 0$ that
  \begin{align*}
    \Pr[x \sim \mu]{\abs{f(x) - \E[\mu]{f}} \geq \alpha} \leq 2\exp\tp{-\frac{\alpha^2 C(\delta)}{2n\nu(f)}},
  \end{align*}
  where $\nu(f)$ is the maximum of one-step variances,
  \begin{align*}
    \nu(f) \triangleq \max_{x \in \Omega(\mu)} \left\{\sum_{y \in \Omega(\mu)} P(x, y)(f(x) - f(y))^2\right\},
  \end{align*}
  where $P$ denotes the transition matrix of the Glauber dynamics over $\mu$.
\end{corollary}

\subsection{Results for general distributions}
Let $\mu$ be a distribution over $\{\0,\1\}^{n}$ and  let $\Omega(\mu)$ be its support.
Given $\Lambda \subseteq [n]$, we use $\mu_\Lambda$ to denote the marginal distribution on $\Lambda$ projected from $\mu$, and we write $\mu_i=\mu_{\{i\}}$ for $i \in [n]$.
Given any partial configuration $\sigma \in \Omega(\mu_\Lambda)$ where $\Lambda \subseteq [n]$, 
we use $\mu^\sigma$  to denote the conditional distribution over $\{\0,\1\}^{n}$ induced by $\mu$ conditional on $\sigma$, 
and we use $\mu^{\sigma \land i \gets x}$ to denote the conditional distribution obtained from $\mu^\sigma$ by further conditioning on the spin of $i\in[n]$ being fixed as $x\in \Omega(\mu^\sigma_i)$. 
 

The notion of \emph{spectral independence} was introduced by Anari, Liu and Oveis Gharan in~\cite{anari2020spectral}.
We use the absolute version of the spectral independence considered in~\cite{feng2021rapid,chen2021rapid}.
%
%
\begin{definition}[{spectral independence (absolute version)}]\label{definition-SI}
Let $\mu$ be a distribution over $\{-1,+1\}^{n}$.
For any $\Lambda \subseteq [n]$, $\sigma \in \Omega(\mu_\Lambda)$, 
the \emph{absolute influence matrix} $\Psi_{\mu^\sigma} \in \mathds{R}_{\geq 0}^{n\times n}$ is defined as  
\begin{align*}
\forall i,j \in [n],\quad \Psi_{\mu^\sigma}(i,j) \triangleq \max_{x,y \in \Omega(\mu^\sigma_i)}\DTV{\mu_j^{\sigma \land i \gets x}}{\mu_j^{\sigma \land i \gets y}},
\end{align*}	
where $\DTV{\cdot}{\cdot}$ denotes the total variation distance.
Let $\eta > 0$.
The distribution $\mu$ is said to be \emph{$\eta$-spectrally independent (SI)} if for any $\Lambda \subseteq V$, any $\sigma \in \Omega(\mu_\Lambda)$, the spectral radius of the influence matrix $\Psi_{\mu^\sigma}$ satisfies
\begin{align*}
\rho\tp{\Psi_{\mu^\sigma}} \leq \eta.
\end{align*}
%
%
%
\end{definition}

It was known that assuming constant marginal lower bound, 
the spectral independence can guarantee the optimal mixing of Glauber dynamics \cite{chen2020optimal, blanca2021mixing}.
In fact, 
MLSIs have been proved assuming the same marginal lower bound \cite{Katalin19, SS19}.
However, such strong condition on marginal bounds does not hold in general for spin systems with unbounded maximum degree, 
and it is a  major open problem to prove MLSI and optimal mixing time for such models.

We introduce the following notion that weakens the marginal lower bound condition.
\begin{definition}[marginal stability]\label{def:bounded-distribution}
  Let $\zeta > 0$.
A distribution $\mu$ over $\{\0,\1\}^{n}$ is said to be \emph{$\zeta$-marginally stable} 
if for any $i\in[n]$, any $S\subseteq\Lambda\subseteq [n]\setminus\{i\}$,  and any $\sigma\in\Omega(\mu_\Lambda)$,
\begin{align*}
R_i^{\sigma}\le \zeta \quad\text{ and }\quad R_i^{\sigma}\le \zeta \cdot R_i^{\sigma_S}, 
\end{align*}
where $R_i^{\sigma}\triangleq \frac{\mu_i^{\sigma}(+1)}{\mu_i^{\sigma}(-1)}$ denotes the marginal ratio, and $R_i^{\sigma_S}$ is accordingly defined for $\sigma_S$.
%
\end{definition}


%

The marginal lower bound assumption imposes a lower bound $b$ for the marginal probability $\mu^{\sigma}_i(x) > b$ for all possible spins $x$.
The marginal stability weakens this to the following properties combined:
\begin{enumerate}
\item 
a one-sided marginal lower bound, to ensure that $\mu_i^{\sigma}(-1)$ is not be too small;
\item a one-sided decay of correlation, to ensure that pinning does not bigly increase the marginal ratio.
\end{enumerate}
Such condition 
ingeniously captures the subcritical two-spin systems. 
On one hand, it is {strong enough}, together with the spectral independence property to guarantee the optimal mixing of Glauber dynamics. 
On the other hand, 
it is also {weak enough} to be satisfied by the subcritical two-spin systems. 


In order to deduce optimal mixing times from spectral independence and marginal stability, we need these properties to hold   for all subcritical external fields. 
Given a distribution $\mu$ over $\{-1,+1\}^{n}$ and a vector $\*\lambda = (\phi_v)_{v \in [n]} \in \mathds{R}^{n}_{>0}$ that specifies the local fields,
we use $(\*\lambda*\mu)$ to denote the distribution obtained from ``magnetizing'' $\mu$ with the local fields in $\*\lambda$. 
Formally:
\begin{align}
\forall \sigma \in \{\0,\1\}^{n},\quad (\*\lambda*\mu)(\sigma) \propto \mu(\sigma) \prod_{i \in [n]: \sigma_i = \1}\lambda_i.	\label{eq:definition-local-fields}
\end{align}
In particular,  
if $\lambda_i = \theta \in \mathds{R}_{>0}$ for all $i \in [n]$ for some scalar $\theta\in\mathds{R}_{>0}$, we simply write $(\theta*\mu)$ for $(\*\lambda*\mu)$.

We formalize the following sufficient condition for a  MLSI for Glauber dynamics.
\begin{condition}\label{condition:CSI-MS}
Let $\eta>1$, $\epsilon > 0$, $\zeta> 1$  be parameters.  The $\mu$ is a distribution over $\{\0,\1\}^{n}$ that satisfies:
\begin{enumerate}
\item \label{condition:CSI} 
$(\*\lambda* \mu)$ is $\eta$-spectrally independent for all $\*\lambda \in (0,1+\epsilon]^{n}$;
\item  \label{condition:MS}
$(\*\lambda * \mu)$ is $\zeta$-marginally stable for all  $\*\lambda \in (0,1]^{n}$.
\end{enumerate}
\end{condition}
An $(\epsilon=0)$ variant of \Cref{condition:CSI-MS} (\ref{condition:CSI}) was used in \cite{chen2021rapid} and called \emph{complete spectral independence}.

Recall that $\rho^{\mathrm{GD}}(\mu)$ denotes the modified log-Sobolev constant of Glauber dynamics on $\mu$ and $\mu^\sigma$ is the conditional distribution over $\{\0,\1\}^{n}$ induced by $\mu$ conditional on $\sigma$.
We further denote by $\rho^{\mathrm{GD}}_{\min}(\mu)$
the minimum modified log-Sobolev constant for $\mu^{\sigma}$ over all possible $\sigma$:
\begin{align*}
\rho^{\mathrm{GD}}_{\min}(\mu) \triangleq \min_{\Lambda \subseteq [n]} \min_{\sigma \in \Omega(\mu_\Lambda)}\rho^{\mathrm{GD}}(\mu^\sigma).
\end{align*}

\begin{theorem}[main theorem: general]\label{theorem-general}
For any distribution $\mu$ over $\{\0,\1\}^{n}$, if $\mu$ satisfies \Cref{condition:CSI-MS} with $\eta>1$, $\epsilon > 0$, $\zeta> 1$,
then the following holds for the modified log-Sobolev constants for Glauber dynamics: 
\begin{align*}
\forall\theta\in(0,1),\quad \rho^{\mathrm{GD}}(\mu) \geq \tp{\frac{\theta}{\mathrm{e}}}^{30\eta + \frac{\log (4\zeta)}{\log(1+\epsilon)}}\rho^{\mathrm{GD}}_{\min}(\theta * \mu).
\end{align*}
\end{theorem}

\Cref{theorem-general} is a boosting theorem for modified log-Sobolev inequality (MLSI).
%
By choosing a suitable constant gap $\theta$, 
the MLSI for the original near-critical distribution $\mu$ is reduced, by losing a constant factor, to the MLSI for the magnetized distribution $\theta * \mu$ that falls into a subcritical regime, where the  minimum MLS constant $\rho^{\mathrm{GD}}_{\min}(\theta * \mu)$ is easier to analyze. 
%
%
A similar boosting theorem for the Poincar\'e constant (spectral gap) was established in \cite{chen2021rapid},  
essentially by assuming  the spectral independence part of \Cref{condition:CSI-MS}.
Here we prove a similar boosting for the  MLSI  by further assuming the marginal stability.

\begin{remark}[applications to spin systems]
When applying \Cref{theorem-general} to anti-ferromagnetic two-spin systems, 
%
as in  \cite{chen2021rapid}, we can first preprocess the distribution $\mu$ by properly flipping the roles of spins in $\{+1,-1\}$ for each vertex, so that after the preprocessing, the distribution will only get ``less critical'' by decreasing the local field at every vertex.
We then formally verify the spectral independence and the marginal stability properties in \Cref{condition:CSI-MS} for the flipped distribution assuming the uniqueness.

\Cref{theorem-general} can then be applied to boosting the MLSI in the subcritical regime up to the uniqueness threshold, 
where the MLSI in the subcritical regime can be obtained from, for example, the result on the Ricci curvature in~\cite{EHMT17}.
This proves \Cref{thm:2spin-theorem}. The detailed analysis is given in \Cref{section-app}.
%
%
\end{remark}

\subsection{Background and related work}
The computational phase transition for sampling and counting in two-spin systems has drawn considerable studies~\cite{jerrum1993polynomial,  GJP03}.  
The $d$-uniqueness (\Cref{definition-d-uniqueness}) represents the uniqueness of infinite-volume Gibbs measure on $(d+1)$-regular tree~\cite{weitz2005combinatorial, LLY13}.
Initiated in a seminal work of Weitz~\cite{weitz2006counting}, 
correlation-decay based deterministic counting algorithms 
were given for anti-ferromagnetic two-spin systems with $\Delta$-bounded maximum degree that satisfy the up-to-$\Delta$-uniqueness (being $d$-unique for all $d<\Delta$) 
\cite{weitz2006counting, LLY12, sinclair2014approximation, LLY13}.
These algorithms  run in time $n^{O(\log \Delta)}$. 
Together with the hardness results in the non-uniqueness regime~\cite{sly2010computational, sly2012computational, galanis2015inapproximability}, this gives a computational phase transition for spin systems with constant maximum degree.

Due to a general lower bound \cite{hayes2007general}, the optimal mixing time of Glauber dynamics is $O(n \log n)$.
It is also widely believed that such optimal mixing time should hold for the two-spin systems in the uniqueness regime.
%
%
Proving such conjectures is extremely challenging.
A substantial body of research works have dedicated to this.
Using coupling based techniques, optimal $O(n \log n)$ mixing times were proved 
 assuming girth lower bound~\cite{hayes2006coupling, efthymiou2019convergence} 
or for Ising models with constant maximum degree~\cite{MS13}.

The spectrum based techniques tries to lower bound the spectral gap of Glauber dynamics.
In a seminal work~\cite{anari2020spectral}, Anari, Liu, and Oveis Gharan introduced the concept of spectral independence and
applied the tools from high-dimensional expander walks developed in~\cite{ALOV19, alev2020improved} 
to relate the spectral gap to the decay of correlation properties.
For anti-ferromagnetic two-spin systems satisfying the uniqueness condition with gap $\delta$,
the spectral gap was improved from $n^{-O(1/\delta)}$~\cite{anari2020spectral, chen2020rapid}, to $\Delta^{-O(1/\delta)} n^{-1}$~\cite{chen2020optimal, jain2021spectral}, and finally to $(1/\delta)^{-O(1/\delta)} n^{-1}$~\cite{chen2021rapid} which was optimal in $n$ for arbitrary maximum degree $\Delta$.
However, as spectral gaps, they are not sufficient for optimal $O(n \log n)$ mixing time.

Entropy based techniques that could prove modified log-Sobolev inequalities (MLSI) were considered \cite{CGM19, ALOV20}. 
Although modified log-Sobolev (MLS) constants can give tight bounds on mixing times, they  are notoriously difficult to analyze.
In many previous works~\cite{CMT15, FM16, EHMT17, Katalin19, SS19, conforti2020probabilistic}, the optimal $\Omega(n^{-1})$ MLS bounds were proved only in the regimes where more standard techniques such as coupling   could also work.
Perhaps the first breakthrough to this was the one achieved by Chen, Liu and Vigoda \cite{chen2020optimal}: 
there and in a follow-up work \cite{blanca2021mixing}, a $(b/\Delta)^{O(1/(\delta b))} n^{-1}$ MLS bound was proved 
for anti-ferromagnetic two-spin systems satisfying the uniqueness condition with gap $\delta$,  
assuming marginal lower bound $b$.
This MLS bound beats the coupling in regimes and is optimal in $n$ for constant $\Delta$.
However, the reliance on margin bound results in a bad dependence on the max-degree $\Delta$.

Recently, Anari, Jain, Koehler, Pham, and Vuong~\cite{anari2021entropic} proposed the notion of \emph{entropic independence}, 
which was crucial for removing the reliance on marginal lower bound in CLV's argument \cite{chen2020optimal}, by assuming spectral independence for all fields.
%
This was followed by~\cite{anari2021entropicII} and \cite{chen2021optimalIsing}, 
where both works used the field dynamics invented in \cite{chen2021rapid} to connect the entropic independence to the MLS constant of Glauber dynamics.
Both succeeded in proving optimal mixing for Ising models in a uniqueness regime that holds for all external fields.
A major obstacle for the current techniques is to prove \emph{optimal} mixing for spin systems with \emph{unbounded} maximum degree in \emph{field-dependent} uniqueness regimes, which is typical for computational phase transitions for anti-ferromagnetic two-spin systems.

\bigskip
\paragraph{\textbf{Concurrent work}}
When preparing the current paper, we were informed by Yuansi Chen and Ronen Eldan about their concurrent work \cite{chen2022localization},
in which they independently prove the same optimal mixing bound for the hardcore model, through a more abstract framework called ``localization schemes''.

\section{Outline of Proofs}
In this section, we outline our proof of \Cref{theorem-general}.

\subsection{Product domination and block factorization}\label{sec:outline-PD-UBF}
The spectral independence and marginal stability in \Cref{condition:CSI-MS} 
together ensure a  property called \emph{product domination}, which plays a key role in the proof.

The probability generating function $g_\mu$ for a distribution $\mu$ over $\{\0,\1\}^{n}$ is defined by
\begin{align*}
g_{\mu}(z_1,z_2,\ldots,z_n) \triangleq \sum_{\sigma \in \{\0,\1\}^{n}}\mu(S)\prod_{i \in [n]: \sigma_i = \1}z_i.	
\end{align*}
\begin{definition}[product domination]\label{def:product-dominated}
Let $\alpha \in (0,1)$ be real. 
A distribution $\mu$ over $\{\0,\1\}^{n}$ is said to be \emph{$(1/\alpha)$-product dominated} 
if for all $(z_1, \cdots, z_n) \in  \mathds{R}^{n}_{> 0}$,
\begin{align*}   
 g_\mu(z_1^\alpha,z_2^\alpha,\ldots,z_n^\alpha)^{\frac{1}{\alpha}} \le \prod_{i=1}^n \tp{\mu_i(\1) z_i + \mu_i(-1)}.
 \end{align*}
 Furthermore, $\mu$ is $(1/\alpha)$-product dominated on $D \subseteq \mathds{R}^{n}_{> 0}$ if the above holds for all $(z_1,\cdots,z_n) \in D$. 
\end{definition}

This property asserts that the ``$\alpha$-fractional'' form $g_\mu(z_1^\alpha,z_2^\alpha,\ldots,z_n^\alpha)^{{1}/{\alpha}}$ of the generating function $g_\mu$ is dominated by the generating function of a product distribution,
 in which the $i$-th variable takes the value $\1$ independently with probability $\mu_i(\1)$.
The same  $\alpha$-fractional form appeared in the notion of fractionally log-concave distributions~\cite{AASV21, anari2021entropic}.
More significantly,
product domination gives an equivalent characterization of the entropic independence introduced in \cite{anari2021entropic}.
More precisely,
$\mu$ is $(1/\alpha)$-product dominated if and only if its homogenization $\mu^{\mathrm{hom}}$ is $(1/\alpha)$-entropically independent.
The formal definitions of entropic independence and homogenization, along with a formal proof of such equivalence between product domination and entropic independence, are given in \Cref{section-pd-mixing}.

%
%

We show that this product domination property is guaranteed by \Cref{condition:CSI-MS} .
For technical reasons, 
we will show that a weakening of \Cref{condition:CSI-MS} is sufficient to guarantee the product domination.
\begin{condition}
\label{condition:CSI-MS-weak}
Let $\eta>1$, $\epsilon > 0$, $\zeta> 1$ be parameters.  The $\mu$ is a distribution over $\{\0,\1\}^{n}$ that satisfies:
\begin{enumerate}
\item \label{condition:CSI-weak}
$(\*\lambda* \mu)$ is $\eta$-spectrally independent for all $\*\lambda \in (0,1+\epsilon]^{n}$;
\item  \label{condition:MS-weak}
$\mu$ is $\zeta$-marginally stable.
\end{enumerate}
\end{condition}

\Cref{condition:CSI-MS-weak} (complete spectral independence with marginal stability)
weakens \Cref{condition:CSI-MS} because it only requires the marginal stability to hold for $\mu$ itself but not for $(\*\lambda* \mu)$ with other external fields $\*\lambda$.

\begin{lemma}\label{lemma-suff-condi}
For any distribution $\mu$ over $\{\0,\1\}^{n}$, if $\mu$ satisfies \Cref{condition:CSI-MS-weak} with $\eta>1$, $\epsilon > 0$ and $\zeta> 1$,
then for any $\Lambda \subseteq [n]$  and $\sigma \in \Omega(\mu_\Lambda)$,  the distribution $\mu^\sigma_{[n] \setminus \Lambda}$ is $(1/\alpha)$-product dominated, where
\begin{align} \label{eq:alpha-def}
\alpha =  \min\left\{\frac{1}{2\eta}, \frac{\log(1+\epsilon)}{\log(1+\epsilon) + \log 2 \zeta}\right\}.
\end{align}
\end{lemma}

\Cref{lemma-suff-condi} is proved in  \Cref{sec:CSI-MS-PD}.
More precisely,
the complete spectral independence (\Cref{condition:CSI-weak} in \Cref{condition:CSI-MS-weak}) guarantees that $\mu$ is $(1/\alpha)$-product dominated on the domain $D = (0,(1+\epsilon)^{1/\alpha}]^{n}$; 
and the marginal stability (\Cref{condition:MS-weak} in \Cref{condition:CSI-MS-weak}) allows to extend such product domination from $D$ to $\mathds{R}_{>0}^{n}$.

The product domination property is closely related to the \emph{uniform block factorization of entropy},
which gives rise to the entropy decay and MLSI for uniform block dynamics.

\begin{definition}[uniform block factorization of entropy \text{\cite{CP20}}]\label{definition:UBF}
Let $n \geq 1$ and $1 \leq \ell \leq n$ be two integers, and $C>0$.
A distribution $\mu$ over $\{\0,\1\}^{n}$ is said to satisfy the \emph{$\ell$-uniform block factorization of entropy with parameter $C$} if for all $f:\Omega(\mu) \to \mathds{R}_{\geq 0}$,
\begin{align*}
\Ent[\mu]{f} \leq \frac{C}{\binom{n}{\ell}} \sum_{S \in \binom{[n]}{\ell}}\mu[\Ent[S]{f}],
\end{align*}
where $\mu[\Ent[S]{f}] \triangleq \sum_{\sigma \in \Omega(\mu_{[n] \setminus S})}\mu_{[n] \setminus S}(\sigma)\cdot\Ent[\mu^\sigma]{f}$. 
\end{definition}

We show that the product domination with all pinnings implies the block factorization of entropy.
\begin{lemma}\label{lemma-pd-to-ubfe}
Let $\mu$ be a distribution over $\{\0,\1\}^{n}$ and $\alpha \in (0,1)$.
If for any $\Lambda \subseteq [n]$ and any $\sigma \in \Omega(\mu_\Lambda)$, 
the conditional marginal distribution $\mu^\sigma_{[n] \setminus \Lambda} $ is $(1/\alpha)$-product dominated, 
then for every integer $\ell \geq 1/ \alpha$, 
the distribution $\mu$ satisfies $\ell$-uniform block factorization of entropy with parameter $C=(\frac{\mathrm{e}n}{\ell})^{1/\alpha + 1}$.
\end{lemma}

\Cref{lemma-pd-to-ubfe}  is proved in \Cref{section-pd-mixing}.
The proof is based on the aforementioned equivalence between product domination and entropic independence,
while the latter is known to guarantee the uniform block factorization of entropy \cite{anari2021entropic}.



\Cref{lemma-suff-condi} and \Cref{lemma-pd-to-ubfe}
together show that \Cref{condition:CSI-MS-weak} guarantees that $\mu$ satisfies the $\ell$-uniform block factorization of entropy,
which is sufficient to imply the MLSI for the $\ell$-uniform block dynamics on $\mu$ \cite{CP20}. 
%
In order to enhance this to the MLSI for single-site Glauber dynamics without resorting to marginal lower bound,
we further apply the $k$-transformation introduced in \cite{chen2021rapid}.


\subsection{$k$-transformation and boosting of MLSI}
The \emph{$k$-transformation} operation for a multi-dimensional probability distribution with Boolean domain is formally defined as follows.

\begin{definition}[\text{$k$-transformation~\cite{chen2021rapid}}] \label{def:k-trans}
Let $\mu$ be a distribution over $\{\0, \1\}^{n}$ and $k \geq 1$ an integer.
The \emph{$k$-transformation} of  $\mu$ gives a distribution $\mu_k = \Rd(\mu, k)$ over $\{\0,\1\}^{n\times k}$, constructed as follows.

Let $\*X\sim\mu$. Then $\mu_k = \Rd(\mu, k)$ is the distribution of $\*Y\in\{\0,\1\}^{n\times k}$ constructed as follows:
\begin{itemize}
  \item if $X_i = \0$, then $Y_{(i,j)} = \0$ for all $j\in[k]$;
  \item if $X_i = \1$, then $Y_{(i,j^*)} = \1$ and $Y_{(i,j)} = \0$ for all $j\in[k]\setminus \{j^*\}$, where $j^*$ is chosen from $[k]$ uniformly and independently at random.
\end{itemize}
\end{definition}

The $k$-transformation defines a sort-of lifting operation on $\mu$. 
It effectively replaces every $i\in[n]$ with a gadget of hardcore $k$-clique.
The limiting object of the uniform block dynamics on $\mu_k$ when $k\to\infty$ gives the \emph{field dynamics} process introduced in \cite{chen2021rapid}. 

The significance of such lifting operations on $\mu$ is that
the uniform block factorizations of entropy for $\mu_k$ for all sufficiently large $k$
give rise to a boosting of MLSI for Glauber dynamics, which holds with no further restriction on  marginals.
Formally, the following theorem was proved in \cite{chen2021optimalIsing}.

\begin{lemma}[\cite{chen2021optimalIsing}]
\label{lem:FBF} 
Let $\mu$ be a distribution over $\{\0,\1\}^{n}$,  $\theta \in (0, 1)$ and $C > 0$. 
  If there is a finite $K_0=K_0(\mu,\theta,C)$ such that for all integers $k\ge K_0$, the distribution $\mu_k = \Rd(\mu, k)$ satisfies $\lceil \theta kn \rceil$-uniform block factorization of entropy with parameter $C$, then the Glauber dynamics on $\mu$ has the modified log-Sobolev constant
\begin{align*}
\rho^{\mathrm{GD}}(\mu) \geq \frac{\rho^{\mathrm{GD}}_{\min}(\theta * \mu)}{C}.	
\end{align*}
\end{lemma}
The exact statement of \Cref{lem:FBF} follows from  \cite[Lemma~2.2]{chen2021optimalIsing} and \cite[Lemma~2.9]{chen2021optimalIsing}.

It only remains to guarantee that the $k$-transformed distribution $\mu_k$ indeed satisfies the desired uniform block factorization of entropy for all sufficiently large $k$.
From the argument in \Cref{sec:outline-PD-UBF}, 
this holds as long as  \Cref{condition:CSI-MS-weak} can be verified for all such $\mu_k$, 
which is guaranteed by the following lemma. 

\begin{lemma}\label{lem:mu-to-muk}
For any distribution $\mu$ over $\{\0,\1\}^{n}$, if $\mu$ satisfies \Cref{condition:CSI-MS} with $\eta>1$, $\epsilon > 0$ and $\zeta> 1$,
then there exists a finite $k_0=10(1+\epsilon)(1+\zeta)$ 
such that for all integers $k \geq k_0$, 
the distribution $\mu_k = \Rd(\mu, k)$ satisfies  \Cref{condition:CSI-MS-weak} with parameters $(\eta', \epsilon', \zeta')$, where $\eta'=2\eta+5$, $\epsilon'=\epsilon$ and $\zeta'=2\zeta$.
\end{lemma}

\Cref{lem:mu-to-muk} basically says that \Cref{condition:CSI-MS} is almost invariant under $k$-transformation (it is not exactly invariant because \Cref{condition:CSI-MS-weak} is still weaker than \Cref{condition:CSI-MS}).
This can be formally  verified by using a natural coupling between $\mu$ and $\mu_k$.
The formal proof is in \Cref{section-muk}.


\subsection{Proof of main theorem}
We now prove \Cref{theorem-general}.
Fix $\eta > 1$, $\epsilon > 0$, $\zeta > 1$ and distribution $\mu$.
Assume that $\mu$ satisfies \Cref{condition:CSI-MS} with parameters $\eta,\epsilon$ and $\zeta$.
Due to \Cref{lem:mu-to-muk},  
there is a finite $k_0=10(1+\epsilon)(1+\zeta)$ 
such that for all integers $k \geq k_0$, 
\Cref{condition:CSI-MS-weak} holds for $\mu_k = \Rd(\mu, k)$ with parameters $\eta'=2\eta+5$, $\epsilon'=\epsilon$ and $\zeta'=2\zeta$,
which according to \Cref{lemma-suff-condi} and \Cref{lemma-pd-to-ubfe}, means that for
\begin{align*}
\alpha =  \min\left\{\frac{1}{4\eta+10}, \frac{\log(1+\epsilon)}{\log (1+\epsilon) + \log (4\zeta)}\right\},
\end{align*}
for all integers 
$\ell \geq 4\eta + 11 + \frac{\log (4\zeta)}{\log (1+\epsilon)} \geq 1/ \alpha$,
$\mu_k$ satisfies $\ell$-uniform block factorization of entropy with parameter 
\begin{align*}
C=\tp{\frac{\mathrm{e}nk}{\ell}}^{1/\alpha + 1} \leq \tp{\frac{\mathrm{e}nk}{\ell}}^{4\eta + 12 + \frac{\log (4\zeta)}{\log (1+\epsilon)}}.
\end{align*}
Fix an arbitrary $\theta \in (0,1)$. 
Consider $\ell = \ctp{\theta k n}$. 
For all sufficiently large integers $k \geq k_0$ satisfying $\theta k n \geq  4\eta +  \frac{\log 4\zeta}{\log(1+\epsilon)} +11 $,  $\mu_k$ satisfies the $\ctp{\theta k n}$-uniform block factorization of entropy with parameter
\begin{align*}
C = \tp{\frac{\mathrm{e}kn}{\ctp{\theta k n}}}^{4\eta + 12 + \frac{\log (4\zeta)}{\log (1+\epsilon)}}	\leq \tp{\frac{\mathrm{e}}{\theta}}^{4\eta + \frac{\log (4\zeta)}{\log(1+\epsilon)} + 12} \leq \tp{\frac{\mathrm{e}}{\theta}}^{30\eta + \frac{\log (4\zeta)}{\log(1+\epsilon)}}
\end{align*}
This holds for all $k \geq K_0$, where $K_0 = \max\left\{10(1+\epsilon)(1+\zeta), \frac{4\eta}{\theta n}+ \frac{\log (4\zeta)}{\theta n\log(1+\epsilon)} + \frac{11}{\theta n}\right\} $. Note that $K' = K_0(\theta)$ is finite because $\epsilon > 0$, $\zeta > 1$ and $\eta > 1$ are fixed parameters. 
\Cref{theorem-general} follows from \Cref{lem:FBF}.

\subsection{Open problems}
Compared to the classes of two-spin systems resolved in \cite{LLY13,chen2021rapid},
the uniqueness criterion (\Cref{condition-canonical-two-spin}) used in \Cref{thm:2spin-theorem} still leaves open the optimal mixing for the following classes of anti-ferromagnetic two-spin systems $(\beta,\gamma,\lambda)$ on $G$:
\begin{itemize}
\item
$\gamma>1$, the graph $G$ is a general graph with maximum degree $\Delta$, and $(\beta,\gamma,\lambda)$ is $d$-unique for all $1\le d<\Delta$ (i.e.~up-to-$\Delta$-unique),
\end{itemize}
that is, the ``skewed'' anti-ferromagnetic case (where $\gamma>1$) on general irregular graphs.

The main technical difficulty for this case is that the MLS constant for such case was very much under-studied, 
even in much sub-critical regimes, 
so there lacks a MLSI for the ``easier'' regime from where we can apply our boosting theorem (\Cref{theorem-general}) for MLSI.

Another minor technical difficulty is that the marginal stability asserted by \Cref{condition:CSI-MS} does not hold for this case in general.
However, we believe that this can be circumvented because \Cref{condition:CSI-MS} provides only a sufficient condition for the product domination property. 
And we conjecture that for any anti-ferromagnetic two-spin system satisfying the up-to-$\Delta$-uniqueness, 
the $k$-transformed distribution $\mu_k$ satisfies the proper product domination property for all sufficiently large $k$.
Verifying this conjecture, while provided  a MLSI in the easier regime, would prove the optimal mixing for the above ``skewed'' anti-ferromagnetic  two-spin systems on general graphs.

\section{Preliminaries}
\subsection{Mixing time and modified log-Sobolev constant}
Let $\Omega$ be a finite state space and $(X_t)_{t \in \mathds{Z}_{\ge 0}}$ a Markov chain on $\Omega$ with transition matrix $P$. 
The Markov chain $P$ is \emph{irreducible}, if for any $\sigma, \tau \in \Omega$, there exists $t \in \mathds{Z}_{\ge 0}$, such that $P^t(\sigma,\tau) > 0$. 
The Markov chain $P$ is \emph{aperiodic}, if $\mathrm{gcd} \{t \in \mathds{Z}_{>0} \mid P^t(\sigma,\sigma) > 0\} = 1$ holds for any $\sigma \in \Omega$.
The fundamental theorem of Markov chains says that an irreducible and aperiodic Markov chain $P$ converges to a unique stationary distribution $\mu$ over $\Omega$ such that $\mu P=\mu$.
The Markov chain $P$ is \emph{reversible} with respect to $\mu$, if the detailed balance equation $\mu(\sigma) P(\sigma,\tau) = \mu(\tau) P(\tau,\sigma)$ holds for all $\sigma,\tau \in \Omega$. Such $\mu$ satisfying the detailed balance equation must be a stationary distribution of $P$.

The mixing time of a Markov chain $P$ with stationary distribution $\mu$ is defined as
\begin{align*}
  T_{\mathrm{mix}}(\epsilon) = \max_{\sigma \in \Omega} \min_{t} \DTV{P^{t}(\sigma,\cdot)}{\mu},
\end{align*}
where $d_{\mathrm{TV}}$ denotes the \emph{total variation distance} and is defined by
\begin{align*}
  d_{\mathrm{TV}}(\mu,\nu) = \frac{1}{2} \sum_{\sigma \in \Omega} \abs{\mu(\sigma)-\nu(\sigma)}.
\end{align*}

Analysis of the mixing time can be done through establishing certain functional inequalities, such as Poincar\'e inequalities and modified log-Sobolev inequalities (MLSI).
Let $\mu$ be the stationary distribution of Markov chain $P$ on state space $\Omega$, and $(\mathds{R}^\Omega, \inner{\cdot}{\cdot}_{\mu})$ be the corresponding space with the inner-product 
\begin{align*}
  \inner{f}{g}_{\mu} \triangleq \sum_{\sigma \in \Omega} \mu(\sigma) f(\sigma) g(\sigma), \quad \forall f,g \in \mathds{R}^\Omega.
\end{align*}
The \emph{Dirichlet form} $\+E_P(\cdot,\cdot)$ is defined by
\begin{align*}
  \+E_P(f,g) \triangleq \inner{(I-P)f}{g}_{\mu}.
\end{align*}
The \emph{modified log-Sobolev (MLS) constant} is defined by
\begin{align*}
  \rho = \rho^{P}(\mu) \triangleq \inf \left\{ \frac{\+E_P(f,\log f)}{\Ent[\mu]{f}} \,\Bigm{|}\, f \in \mathds{R}_{>0}^\Omega \text{ and } \Ent[\mu]{f} \neq 0 \right\},
\end{align*}
where $\Ent[\mu]{f} \triangleq \E[\mu]{f \log f} - \E[\mu]{f} \log \E[\mu]{f}$ and we assume $0\log 0 = 0$.

The following relation between modified log-Sobolev constant and mixing time was known.
\begin{proposition}[\cite{bobkov2006modified}]
Let $P$ be an irreducible, aperiodic and reversible  Markov chain on finite state space $\Omega$ with stationary distribution $\mu$.
If all eigenvalues of $P$ are non-negative, then the mixing time satisfies 
  \begin{align*}
    T_{\mathrm{mix}}(\epsilon) \le \frac{1}{\rho^{P}(\mu)} \tp{\log \log \frac{1}{\mu_{\min}} + \log \frac{1}{2\epsilon^2}},
  \end{align*}
  where $\mu_{\min} \triangleq \min_{\sigma \in \Omega} \mu(\sigma)$.
\end{proposition}
\subsection{Entropic independence}
Let $\mu:\binom{[n]}{k} \to \mathds{R}_{\geq 0}$ be a distribution over all $k$-subsets of $[n]$. We call such distribution a \emph{homogeneous distribution}.

Let $\Omega \subseteq \binom{[n]}{k}$ denote the support of $\mu$.
Let $X$ be the downward closure of $\Omega$.  
Formally, $X$ is the smallest family such that $\Omega \subseteq X$ and if $\alpha \in X$ then $\beta \in X$ for all $\beta \subseteq \alpha$.
In other words, $X$ is the \emph{simplicial complexes} generated by $\mu$.
For any \emph{face} $\alpha \in X$, let $\abs{\alpha}$ denote the \emph{dimension} of $\alpha$.
For any integer $0 \leq j \leq  k$, let $X(j)$ denote all the faces in $X$ with dimension $j$.
\begin{definition}[down/up walk]
\label{definition-DUW}
Let $X$ be the simplicial complexes generated by a homogeneous distribution $\mu:\binom{[n]}{k} \to \mathds{R}_{\geq 0}$.
Let $0 \leq j < k$ be an integer.
\begin{itemize}
\item The \emph{down walk} $D_{k \to j}: X(k) \times X(j) \to \mathds{R}_{\geq 0}$ is defined by
\begin{align*}
\forall \alpha \in X(k),\beta \in X(j),\quad D_{k \to j}(\alpha,\beta) = \begin{cases}
 	\frac{1}{\binom{k}{j}} &\text{if } \beta \subseteq \alpha;\\
 	0 &\text{otherwise.}
 \end{cases}
\end{align*}
\item The \emph{up walk} $U_{j \to k}: X(j) \times X(k) \to \mathds{R}_{\geq 0}$ is defined by
\begin{align*}
\forall \alpha \in X(j),\beta \in X(k),\quad U_{j \to k}(\alpha,\beta) = \begin{cases}
 	\frac{\mu(\beta)}{\sum_{\gamma \in X(k): \alpha \subseteq \gamma}\mu(\gamma)} &\text{if } \alpha \subseteq \beta;\\
 	0 &\text{otherwise.}
 \end{cases}
\end{align*}
\end{itemize}
\end{definition}

The following definition of \emph{entropic independence} was introduced in~\cite{anari2021entropic}.
\begin{definition}[\text{entropic independence~\cite{anari2021entropic}}]\label{definition-ei}
Let $\alpha \in (0,1)$.
A distribution $\mu$ over $\binom{[n]}{k}$ is said to be \emph{$(1/\alpha)$-entropically independent} if for any distribution $\nu$ over $\Omega(\mu)$,
\begin{align*}
\KL{\nu D_{k\to 1}}{\mu D_{k \to 1}} \leq \frac{1}{\alpha k}\KL{\nu}{\mu}.
\end{align*}
\end{definition}

  Let $\mu$ be a distribution over $\binom{[n]}{k}$.
  For a set $R \subseteq [n]$ satisfying $\Pr[S \sim \mu]{R \subseteq S} > 0$, we use $\mu^R$ to denote the \emph{link} of $\mu$ produced by $R$. This notion was used in, for examples, \cite{oppenheim2018local,KO20,ALOV19,anari2020spectral,anari2021entropic}.
  Formally, $\mu^R$ is a distribution  over $X_R = \left\{S \setminus R \mid S \in \binom{[n]}{k} \land S \supseteq R\right\}$ defined by
  \begin{align}\label{eq-def-link}
    \forall T \in X_R, \quad \mu^R(T) \propto \mu(T \uplus R).
  \end{align}
  The following relative entropy decay result was implied by the entropic independence property~\cite{anari2021entropic}.
  We say a property holds for all links of $\mu$, if it holds for $\mu^R$ for all $R \subseteq [n]$ with $\Pr[S\sim \mu]{R\subseteq S} > 0$.
%
%
\begin{theorem}[\text{\cite[Theorem~5]{anari2021entropic}}]
\label{theorem-AJK+}
Let $\mu$ be a distribution over $\binom{[n]}{k}$ and $\alpha \in (0,1)$.
If the $(1/\alpha)$-entropical independence holds for all links of $\mu$, then for any integer $0 \leq j \leq k - \lceil 1/\alpha \rceil$, any distribution $\nu$ over $\Omega(\mu)$,
\begin{align*}
\KL{\nu D_{k \rightarrow j}}{ \mu D_{k \rightarrow j} }	\leq \tp{1-\kappa\tp{j,k,\frac{1}{\alpha}}}\KL{\nu}{\mu},
\end{align*}
where
\begin{align}
\label{eq-def-kappa}
\kappa\tp{j,k,c} \triangleq \frac{(k+1-j-c)^{c- \ctp{c} } \prod_{i=0}^{\ctp{c}-1}(k-j - i) }{(k+1)^{c}}.
\end{align}
\end{theorem}

Let $\mu_{(k)} = \mu$. For any integer $0 \leq j < k$, let $\mu_{(j)} = \mu_{(k)} D_{k\to j}$.
Given any function $f=f^{(k)}: X(k) \to \mathds{R}_{\geq 0}$,
for any integer $0 \leq j < k$, define $f^{(j)}: X(j) \to \mathds{R}_{\geq 0}$ by that $f^{(j)} = U_{j \to k}f^{(k)}$. 

\begin{lemma}
\label{lemma-local-etp-decay}
Let  $\mu=\mu_{(k)}$ be a distribution over $\binom{[n]}{k}$. Let $0 \leq j < k$ and $\kappa \in (0,1)$. 
Assume that for any distribution $\nu$ absolutely continuous with respect to $\mu$, it holds that $\KL{\nu D_{k \rightarrow j}}{ \mu D_{k \rightarrow j} }	\leq (1-\kappa)\KL{\nu}{\mu}$.
Then for any function $f^{(k)}: X(k) \to \mathds{R}_{\geq 0}$,
\begin{align*}
\Ent[\mu_{(j)}]{f^{(j)}} \leq (1-\kappa)	\Ent[\mu_{(k)}]{f^{(k)}}.
\end{align*}
\end{lemma}
\begin{proof}
The proof is standard. We include it here for completeness.
First note that if $f^{(k)} = 0$, then the lemma holds trivially.
  Without loss of generality, we assume $f \not \equiv 0$. By homogeneity, we may further assume $\E[\mu_{(k)}]{f^{(k)}} =1$. 
Let $\nu$ be  a distribution over $X(k)$ defined by $\nu(\sigma) = \mu(\sigma)f^{(k)}(\sigma)$ for all $\sigma \in X(k)$.
\begin{align*}
\KL{\nu}{\mu}  = \sum_{\sigma \in X(k)}\mu(\sigma)\frac{\nu(\sigma)}{\mu(\sigma)} \log \frac{\nu(\sigma)}{\mu(\sigma)} = \EE[\mu]{f^{(k)} \log f^{(k)}} \overset{(\ast)}{=}  \Ent[\mu_{(k)}]{f^{(k)}},
\end{align*}
where $(\ast)$ holds because $\E[\mu]{f^{(k)}} = 1$ and $\mu=\mu_{(k)}$. 
Let $\nu_{(j)} = \nu D_{k\to j}$ and $\mu_{(j)} = \mu D_{k \to j}$.
We have
\begin{align*}
\KL{\nu D_{k \rightarrow j}}{ \mu D_{k \rightarrow j} }	= \sum_{\sigma \in X(j)}	\mu_{(j)}(\sigma)\frac{\nu_{(j)}(\sigma)}{\mu_{(j)}(\sigma)}\log\frac{\nu_{(j)}(\sigma)}{\mu_{(j)}(\sigma)} = \Ent[\mu_{(j)}]{f^{(j)}},
\end{align*}
where the last equation holds because $f^{(j)}(\sigma) = U_{j \to k}f(\sigma) =  \sum_{\alpha \in X(k): \sigma \subseteq \alpha}U_{j \to k}(\sigma,\alpha)\frac{\nu(\alpha)}{\mu(\alpha)} =\frac{\nu_{(j)}(\sigma)}{\mu_{(j)}(\sigma)}$ and $\E[\mu_{(j)}]{f^{(j)}} = \sum_{\sigma \in X(j)}\mu_{(j)}(\sigma)\frac{\nu_{(j)}(\sigma)}{\mu_{(j)}(\sigma)} = 1$.
This proves the lemma.
\end{proof}

\section{Factorization of Entropy via Product Domination}\label{section-pd-mixing}
In this section, we prove \Cref{lemma-pd-to-ubfe},
the implication from the product domination to the uniform block factorization of entropy.
We first prove an equivalence relation (\Cref{lemma-equvalent-EI-PD}) between product domination and entropic independence~\cite{anari2021entropic}; and then  the uniform block factorization of entropy is established through the entropic independence (\Cref{prop:ei-implies-block}). 

\subsection{Product domination and entropic independence}
Recall the following equivalent algebraic definition of entropic independence \cite{anari2021entropic}.
\begin{definition}[\text{algebraic definition of entropic independence~\cite{anari2021entropic}}]\label{definition-entropic-independence}
Let $\alpha \in (0,1)$.
A homogeneous distribution $\pi$ over $\binom{[n]}{k}$ is $(1/\alpha)$-entropically independent if 
\begin{align}\label{eq-def-ei-alg}
 \forall (z_1, \cdots, z_n) \in \mathds{R}^{n}_{> 0},\quad g_\pi (z_1^\alpha, \cdots, z_n^\alpha)^{\frac{1}{k\alpha}} \leq \frac{1}{k}\sum_{i = 1}^n \Pr[S \sim \pi]{i \in S} z_i .	
\end{align}
Furthermore, $\pi$ is $(1/\alpha)$-entropically independent on $D \subseteq \mathds{R}^{n}_{> 0}$, if above holds for all $(z_1,\cdots,z_n) \in D$. 

\end{definition}
\begin{remark}
The original algebraic definition of entropic independence~\cite[Theorem 4]{anari2021entropic}	requires~\eqref{eq-def-ei-alg} holds for all  $(z_1, \cdots, z_n) \in \mathds{R}^{n}_{\geq 0}$. The two definitions are equivalent by continuity.
\end{remark}

\begin{remark}
\Cref{definition-entropic-independence} is equivalent to \Cref{definition-ei} (see \cite[Theorem 4]{anari2021entropic}).
\end{remark}

For any distribution $\mu$ over $\{\0, \1\}^{n}$. The \emph{homogenization} of $\mu$, denoted by $\mu^{\mathrm{hom}}$, is a distribution over $\binom{[n]\cup [\bar{n}]}{n}$, where $[\bar{n}] = \{\bar{1},\bar{2},\ldots,\bar{n}\}$.
For any configuration $\sigma \in \{\0, \1\}^{n}$, we define 
\[
S_\sigma \triangleq \{i \mid \sigma_i = \1\} \cup \{\bar{i} \mid \sigma_i = \0\},
\]
then the homogenization $\mu^{\-{hom}}$ is defined by
\begin{align*}
  \forall \sigma \in \{\0, \1\}^{n}, \quad \mu^{\-{hom}}(S_\sigma) \triangleq \mu(\sigma),
\end{align*}
and $\mu^{\-{hom}}(T) = 0$ for any $T$  that cannot be expressed as $S_\sigma$ for some $\sigma \in \{\0, \1\}^{n}$.


The following lemma gives the relation between product domination and entropic independence.
\begin{lemma}\label{lemma-equvalent-EI-PD}
Let $\mu$ be a distribution $\mu$ over $\{\0, \1\}^{n}$ and $\alpha \in (0,1)$. Let $D \subseteq \mathds{R}^{n}_{> 0}$ and define
\begin{align*}
D^{\mathrm{hom}} = \left\{(z_1,z_2,\ldots,z_n,z_{\bar{1}},z_{\bar{2}},\ldots, z_{\bar{n}}) \in \mathds{R}^{2n}_{>0} \,\bigg|\, \tp{\frac{z_1}{z_{\bar{1}}},\frac{z_2}{z_{\bar{2}}},\ldots,\frac{z_n}{z_{\bar{n}}}} \in D \right\}.	
\end{align*}
%
$\mu$ is ($1/\alpha$)-product dominated on $D$ if and only if $\mu^{\mathrm{hom}}$ is $(1/\alpha)$-entropically independent on $D^{\mathrm{hom}}$.	

In particular, $\mu$ is $(1/\alpha)$-product dominated if and only if $\mu^{\mathrm{hom}}$ is $(1/\alpha)$-entropically independent.
\end{lemma}




\begin{proof}
Denote $\pi=\mu^{\mathrm{hom}}$.
  We first prove the sufficiency.
  Note that  we have
\[
g_\pi(x_1,x_2,\ldots,x_n,1,1,\ldots,1)= g_\mu(x_1,x_2,\ldots,x_n),
\]
 $\Pr[S\sim \pi]{i\in S} = \mu_i(\1)$ and $\Pr[S\sim \pi]{\bar{i} \in S} = \mu_i(\0)$.
Therefore, $\mu$ being $(1/\alpha)$-product dominated on $D$ means that the generating function of $\pi = \mu^{\mathrm{hom}}$ satisfies that for all $(x_1,x_2,\ldots,x_n) \in D$,
\begin{align*}
g_\pi \tp{x_1^\alpha, x_2^\alpha, \cdots, x_n^\alpha, 1, \cdots, 1} ^{\frac{1}{\alpha}} \leq 	\prod_{i=1}^n\tp{ \Pr[S \sim \pi]{i \in S} x_i + \Pr[S \sim \pi]{\bar{i} \in S} }.	
\end{align*}
Hence, for any $(z_1,z_2,\ldots,z_n,z_{\bar{1}},z_{\bar{2}},\ldots,z_{\bar{n}}) \in D^{\mathrm{hom}}$,
\begin{align*}
  g_\pi \tp{\tp{\frac{z_1}{z_{\bar{1}}}}^\alpha, \tp{\frac{z_2}{z_{\bar{2}}}}^\alpha, \cdots, \tp{\frac{z_n}{z_{\bar{n}}}}^\alpha, 1, \cdots, 1} ^{\frac{1}{\alpha}} \leq 	\prod_{i=1}^n\tp{ \Pr[S \sim \pi]{i \in S} \frac{z_i}{z_{\bar{i}}} + \Pr[S \sim \pi]{\bar{i} \in S} }.
\end{align*}
Multiplying both sides by $(z_{\bar{1}} z_{\bar{2}} \cdots  z_{\bar{n}})$ gives 
\begin{align}\label{eq-eq-proof}
g_\pi (z_1^\alpha, z_2^\alpha, \cdots, z_n^\alpha,z_{\bar{1}}^\alpha, z_{\bar{2}}^\alpha, \cdots, z_{\bar{n}}^\alpha)^{\frac{1}{\alpha}} \leq \prod_{i=1}^n\tp{ \Pr[S \sim \pi]{i \in S} z_i + \Pr[S \sim \pi]{\bar{i} \in S} z_{\bar{i}} }. 
\end{align}
Applying AM-GM inequality, for any $(z_1,z_2,\ldots,z_n,z_{\bar{1}},z_{\bar{2}},\ldots,z_{\bar{n}}) \in D^{\mathrm{hom}}$,
\begin{align*}
g_\pi(z_1^\alpha, z_2^\alpha,\ldots, z_n^\alpha, z_{\bar{1}}^\alpha, z_{\bar{2}}^\alpha,\ldots, z_{\bar{n}}^\alpha)^{\frac{1}{n\alpha}} &\leq 
  \prod_{i=1}^n\tp{ \Pr[S \sim \pi]{i \in S} z_i + \Pr[S \sim \pi]{\bar{i} \in S} z_{\bar{i}} }^{\frac{1}{n}}\\
\text{(by AM-GM)}\quad   &\leq 	\frac{1}{n}\sum_{i = 1}^n \tp{\Pr[S \sim \pi]{i \in S} z_i + \Pr[S \sim \pi]{\bar{i} \in S} z_{\bar{i}}},
\end{align*}
which implies that $\pi = \mu^{\mathrm{hom}}$ is $(1/\alpha)$-entropically independent on $D^{\mathrm{hom}}$ by \Cref{definition-entropic-independence}.

Next, we prove the necessity. 
Fix arbitrary $(x_1, x_2,\ldots,x_n) \in D$. Define $z_1,\ldots,z_n$ and $z_{\bar{1}},\ldots,z_{\bar{n}}$ respectively as
\begin{align*}
\forall	i \in [n], \quad z_i = \frac{x_i}{x_i \Pr[S\sim \pi]{i \in S} + \Pr[S\sim \pi]{\bar{i} \in S}}\quad\text{and}\quad z_{\bar{i}} = \frac{1}{x_i \Pr[S\sim \pi]{i \in S} + \Pr[S\sim \pi]{\bar{i} \in S}}.
\end{align*}
It is straightforward to verify that $(z_1,\ldots,z_n,z_{\bar{1}},\ldots,z_{\bar{n}}) \in D^{\mathrm{hom}}$ and $\Pr[S \sim \pi]{i \in S} z_i + \Pr[S \sim \pi]{\bar{i} \in S} z_{\bar{i}} = 1$. 
Therefore, $\pi = \mu^{\mathrm{hom}}$ being $(1/\alpha)$-entropically independent means that the generating function of $\pi$ satisfies
\begin{align*}
g_\pi (z_1^\alpha, \cdots, z_n^\alpha,z_1^\alpha,\ldots,z_{\bar{n}}^\alpha)^{\frac{1}{\alpha}} \leq \tp{\frac{1}{n}\sum_{i = 1}^n \tp{\Pr[S \sim \pi]{i \in S} z_i +\Pr[S \sim \pi]{\bar{i} \in S} z_{\bar{i}}}}^n = 1.	
\end{align*}
Note that $\prod_{i=1}^n \tp{\Pr[S\sim \pi]{i \in S} z_i + \Pr[S\sim \pi]{\bar{i} \in S} z_{\bar{i}}}=1$. Therefore,
\begin{align*}
g_\pi (z_1^\alpha, \cdots, z_n^\alpha,z_{\bar{1}}^\alpha, \cdots, z_{\bar{n}}^\alpha)^{\frac{1}{\alpha}} \leq 1 = \prod_{i=1}^n\tp{ \Pr[S \sim \pi]{i \in S} z_i + \Pr[S \sim \pi]{\bar{i} \in S} z_{\bar{i}} }.	
\end{align*}
Dividing both sides by $(z_{\bar{1}} z_{\bar{2}} \ldots z_{\bar{n}})$ gives
\begin{align*}
  g_\pi \tp{\tp{\frac{z_1}{z_{\bar{1}}}}^\alpha, \tp{\frac{z_2}{z_{\bar{2}}}}^\alpha, \cdots, \tp{\frac{z_n}{z_{\bar{n}}}}^\alpha, 1, \cdots, 1} ^{\frac{1}{\alpha}} \leq 	\prod_{i=1}^n\tp{ \Pr[S \sim \pi]{i \in S} \frac{z_n}{z_{\bar{n}}} + \Pr[S \sim \pi]{\bar{i} \in S} }.
\end{align*}
Note that $\frac{z_i}{z_{\bar{i}}} = x_i$ and recall that $g_\pi(x_1,x_2,\ldots,x_n,1,1,\ldots,1) = g_\mu(x_1,x_2,\ldots,x_n)$, $\Pr[S\sim \pi]{i\in S} = \mu_i(\1)$ and $\Pr[S\sim \pi]{\bar{i} \in S} = \mu_i(\0)$. We have the following holds for all $(x_1, x_2,\ldots,x_n) \in D$
\begin{align*}
g_\mu(x_1^\alpha,x_2^\alpha,\ldots,x_n^\alpha)^{\frac{1}{\alpha}}\leq  \prod_{i=1}^n\tp{ \mu_i(\1) x_i + \mu_i(\0) },
\end{align*}
which implies $\mu$ is $(1/\alpha)$-product dominated on $D$.
\end{proof}


\subsection{Entropic independence and block factorization of entropy}

We now use the entropic independence to obtain the uniform block factorization of entropy.
Recall that the link of a distribution is defined in~\eqref{eq-def-link}.
\begin{proposition}\label{prop:ei-implies-block}
Let  $\mu$ be a distribution over $\{\0,\1\}^{n}$ and $\pi = \mu^{\mathrm{hom}}$ over $\binom{[n] \cup [\bar{n}]}{n}$ its homogenization.
Let $\alpha \in (0,1)$.
If $(1/\alpha)$-entropic independence holds for all links of $\pi$, then for any $\ctp{1/\alpha} \leq \ell \leq n$, $\mu$ satisfies $\ell$-uniform block factorization of entropy with $C=\kappa(n-\ell,n,{1}/{\alpha})^{-1}$, where $\kappa(\cdot)$ is defined in \eqref{eq-def-kappa}.
\end{proposition}

The proof of \Cref{prop:ei-implies-block} is standard. We include it here for completeness.

\begin{proof}[Proof of \Cref{prop:ei-implies-block}]
Fix any function $f: \Omega(\mu) \to \mathds{R}_{\geq 0}$.
We construct $f^{(n)}: \Omega({\pi}) \to \mathds{R}_{\geq 0}$ as that $f^{(n)}(S_\sigma) = f(\sigma)$ for all $\sigma \in \Omega(\mu)$, where $S_\sigma = \{i \mid \sigma_i = \1\} \cup \{\overline{i} \mid \sigma_i = \0\}$.
Let $X$ denote the simplicial complexes generated by $\pi$.
Let $U_{\cdot}$ and $D_{\cdot}$ denote the up walk and down walk on $X$ (\Cref{definition-DUW}).
Let $\pi_{(n)} = \pi$ and $\pi_{(j)} = \pi_{(n)} D_{n \to j}$ for all $0 \leq j < n$.
%
Let $f^{(j)} = U_{j \to n}f^{(n)}$ for all $0 \leq j < n$.

Recall the notation $\mu[\Ent[S]{f}]$ used in \Cref{definition:UBF}:
\[
\mu[\Ent[S]{f}] \triangleq \sum_{\sigma \in \Omega(\mu_{[n] \setminus S})}\mu_{[n]\setminus S}(\sigma)\Ent[\mu^\sigma]{f}.
\]
The following lemma is proved in~\cite{chen2020optimal} (see the proof of Lemma~2.6 in the full version of~\cite{chen2020optimal}).

\begin{lemma}[\text{\cite{chen2020optimal}}]
\label{lemma-UBF=ED}
Let $\mu$ be a distribution over $\{\0,\1\}^{n}$.
For any $0\leq j \leq n$,
it holds that 
\begin{align}
\frac{1}{\binom{n}{j}}\sum_{S \in \binom{[n]}{j}}\mu[\Ent[S]{f}]&= \Ent[\pi_{(n)}]{f^{(n)}} - \Ent[\pi_{(n-j)}]{f^{(n-j)}}.\label{eq-CLV-2}
\end{align}
\end{lemma}


Note that all conditional marginal distributions induced by $\pi$ are $(1/\alpha)$-entropically independent. By \Cref{theorem-AJK+} and \Cref{lemma-local-etp-decay}, for any $f: \Omega(\mu) \to \mathds{R}_{\geq 0}$ and $0 \le j \le n - \ctp{1/\alpha}$,
  \begin{align}
    \label{eq-proof-e-decay}
    \Ent[\pi_{(j)}]{f^{(j)}}  \leq \tp{1-\kappa\tp{j,n,1/\alpha}}\Ent[\pi_{(n)}]{f^{(n)}},
  \end{align}
where $\kappa(\cdot)$ is defined in~\eqref{eq-def-kappa}. Hence, for any $ \ctp{1/\alpha} \leq \ell \leq n$,
\begin{align*}
  \Ent[\mu]{f} &\overset{(\ast)}{=} \Ent[\pi_{(n)}]{f^{(n)}}
  = \Ent[\pi_{(n)}]{f^{(n)}} -  \Ent[\pi_{(n-\ell)}]{f^{(n-\ell)}} + \Ent[\pi_{(n-\ell)}]{f^{(n-\ell)}}\\
  \text{(by~\eqref{eq-proof-e-decay} and~\Cref{lemma-UBF=ED})}\quad &\leq \frac{1}{\binom{n}{\ell}}\sum_{S \in \binom{[n]}{\ell}}\mu[\Ent[S]{f}] + \tp{1- \kappa\tp{n-\ell,n,1/\alpha}} \Ent[\pi_{(n)}]{f^{(n)}}\\
  &\overset{(\star)}{=} \frac{1}{\binom{n}{\ell}}\sum_{S \in \binom{[n]}{\ell}}\mu[\Ent[S]{f}] + \tp{1- \kappa\tp{n-\ell,n,1/\alpha}} \Ent[\mu]{f}
\end{align*}
where $(\ast)$ and $(\star)$ hold due to the definitions of $\pi_{(n)}$ and $f^{(n)}$. This implies that
\begin{align*}
	\Ent[\mu]{f} \leq \frac{\kappa(n-\ell,n,\frac{1}{\alpha})^{-1}}{\binom{n}{\ell}} \sum_{S \in \binom{[n]}{\ell}}\mu[\Ent[S]{f}]. &\qedhere
\end{align*}
\end{proof}

\subsection{Block factorization of entropy via product domination}
We are now ready to prove \Cref{lemma-pd-to-ubfe}.
\begin{proof}[Proof of \Cref{lemma-pd-to-ubfe}]
 We interpret $\mu$ as a distribution over the power set $2^{[n]}$.
 
 Let $\pi = \mu^{\mathrm{hom}}$  over $\binom{[n]\cup[\bar{n}]}{n}$ be its homogenization.
There is a one-to-one correspondence between conditional distribution in $\mu$ and links of $\pi$. 
Recall the link defined in~\eqref{eq-def-link}.
Fix any link $\pi^R$ of $\pi$. 
%
%
%
%
%
It is straightforward to verify that there exists a partial configuration $\sigma \in \Omega(\mu_\Lambda)$ such that
\begin{align*}
  \tp{\mu^\sigma_{V \setminus \Lambda}}^{\mathrm{hom}} = \pi^R. 
\end{align*}
 By \Cref{lemma-equvalent-EI-PD}, assumption of \Cref{lemma-pd-to-ubfe} and the monotonicity of entropic independence (see~\Cref{definition-ei}),
 the $\ctp{1/\alpha}$-entropical independence holds for all links of $\pi$.
Due to \Cref{prop:ei-implies-block}, for any $\ctp{1/\alpha} \leq \ell \leq n$, 
the distribution $\mu$ satisfies $\ell$-uniform block factorization of entropy with parameter 
 \begin{align*}
 C
 &=\frac{1}{\kappa(n-\ell,n,\ctp{1/\alpha})}
 = \binom{n}{\ctp{1/\alpha}} \Big/ \binom{\ell}{\ctp{1/\alpha}} 
 \leq \tp{\frac{\mathrm{e}n}{\ctp{1/\alpha}}}^{\ctp{1/\alpha}} \Big/ \tp{\frac{\ell}{^{\ctp{1/\alpha}} } }^{\ctp{1/\alpha}} 
 = \tp{\frac{\mathrm{e}n}{\ell}}^{{\ctp{1/\alpha}} } 
 \leq \tp{\frac{\mathrm{e}n}{\ell}}^{\frac{1}{\alpha} + 1}. &\qedhere
 \end{align*}
 \end{proof}

\section{Product Domination from Marginally Stable Spectral Independence}\label{sec:CSI-MS-PD}

%
In this section, we prove \Cref{lemma-suff-condi}, establishing of the product domination property (\Cref{def:product-dominated})
through the spectral independence and marginal stability properties guaranteed in \Cref{condition:CSI-MS-weak}.

We first define the complete spectral independence, which will be used in the following sections.
\begin{definition}[complete spectral independence]\label{definition-complete-SI} 
  Let $\eta > 1$ and $\epsilon > 0$. A distribution $\mu$ over $\{-1,+1\}^n$ is said to be \emph{$(\eta,\epsilon)$-completely spectrally independent}, if $(\*\lambda * \mu)$ is $\eta$-spectrally independent for any $\*\lambda \in (0,1+\epsilon]^n$.
\end{definition} 
Define a function $F_{\mu,\alpha} : \mathds{R}^n_{> 0} \rightarrow \mathds{R}$ by
\begin{align}\label{eq-def-f-mu-a}
 F_{\mu,\alpha}(z_1,z_2,\ldots,z_n) \triangleq \frac{g_\mu(z_1^\alpha,z_2^\alpha,\ldots,z_n^\alpha)^{{\frac{1}{\alpha}}}}{\prod_{i=1}^n \tp{\mu_i(\1) z_i + \mu_i(-1)}}.	
\end{align}

It is not hard to see that $F_{\mu,\alpha}\le 1$ implies that $\mu$ is $(1/\alpha)$-product dominated.
Moreover, the following lemmas transform \Cref{condition:CSI-MS-weak} to the following conditions regarding function $F_{\mu,\alpha}$.

\begin{lemma}\label{lemma-verify-1}
Let $\eta > 1$ and $\epsilon > 0$. 
If a distribution $\mu$ over $\{\0,\1\}^n$ is  $(\eta,\epsilon)$-completely spectrally independent,
then for any $0 < \alpha \leq 1/ (2\eta)$, 
it holds that $F_{\mu,\alpha}(x) \leq 1$ for all $x\in \left(0, (1 + \epsilon)^{1/\alpha}\right]^n$.
\end{lemma}

\begin{lemma}\label{lemma-verify-2}
Let $\zeta >1$.
If  a distribution $\mu$  over $\{\0,\1\}^n$ is $\zeta$-marginally stable,
then for any $\alpha \in (0,1)$, any $x \in \mathds{R}_{> 0}^n$, any $i \in [n]$, 
if $x_i \geq (2\zeta)^{\frac{1}{1 - \alpha}} $, then
    \begin{align*}
    \frac{\partial F_{\mu,\alpha}}{\partial z_i}	\bigg|_{z = x} \leq 0. 
    \end{align*}
\end{lemma}

The complete spectral independence implies the product domination in $\left(0, (1 + \epsilon)^{1/\alpha}\right]^n$ through \Cref{lemma-verify-1},
which is extended to all $ \mathds{R}_{> 0}^n$ through the monotonicity in \Cref{lemma-verify-2} implied by the marginal stability.

  It remains to ensure the complete spectral independence and the marginal stability in \Cref{condition:CSI-MS-weak} closed under pinning, which is straightforward because their definitions already  consider all pinnings.
  \begin{fact}\label{fact:pinning-SI}
    Let $\eta>1,\zeta$ and $\epsilon>0$. If a distribution $\mu$ over $\{\0,\1\}^n$ is  $(\eta,\epsilon)$-completely spectrally independent and $\zeta$-marginally stable, then
    these properties also hold for $\mu^\sigma_{[n] \setminus\Lambda}$ for arbitrary $\Lambda \subseteq [n]$ and $\sigma \in \Omega(\mu_{\Lambda})$.
  \end{fact}

  \begin{proof}[Proof of \Cref{lemma-suff-condi}]
  Denote $\nu = \mu^\sigma_{[n] \setminus \Lambda}$. 
  Without loss of generality, suppose $[n] \setminus \Lambda = [m] = \{1,2,\ldots,m\}$.
  By the definition of the function $F_{\nu,\alpha}$, it suffices to show that 
  \begin{align*}
  	\forall x \in \mathds{R}_{>0}^m,\quad F_{\nu,\alpha}(x) \leq 1.
  \end{align*}
  Denote $D=\left(0, (1 + \epsilon)^{1/\alpha}\right]^m$.
  By \Cref{fact:pinning-SI} and \Cref{lemma-verify-1}, since $\alpha \le 1/(2\eta)$, $F_{\nu,\alpha}(x) \le 1$ for any $x \in D$. 
  
  Therefore, it remains to take care of those $x \not \in D$.
  Fix an arbitrary $x \in  \mathds{R}_{>0}^m\setminus D$. 
  Define $\tilde{x} \in \mathds{R}_{>0}^m$ as that $\tilde{x}_i=\min\{x_i, (1 + \epsilon)^{1/\alpha}\}$ for all $i\in[m]$.
  Obviously $\tilde{x} \in D$, and hence $F_{\nu,\alpha}(\tilde{x}) \leq 1$.
  We only need to show that $F_{\nu,\alpha}({x})\le F_{\nu,\alpha}(\tilde{x})$.
  
  Denote $M=\{i\in[m]\mid x_i>(1 + \epsilon)^{1/\alpha}\}$.  
 By mean value theorem, there exists $\theta \in (0,1)$ such that
 \begin{align*}
    F_{\nu,\alpha}(x) - F_{\nu,\alpha}(\tilde{x}) 
    = \left\langle(x-\tilde{x}), \nabla F_{\nu,\alpha}(\theta x + (1-\theta) \tilde{x})\right\rangle 
    = \sum_{i \in M}^m (x_i - \tilde{x}_i)\frac{\partial F}{\partial z_i}\bigg|_{z = \theta x + (1-\theta) \tilde{x}},
  \end{align*}
  where the last equation holds because $x_i = \tilde{x}_i$ for all $i\in [m]\setminus M$.
  Fix any $i \in M$.
  It holds that $x_i > \tilde{x}_i = (1 + \epsilon)^{1/\alpha}$, thus $(\theta x + (1-\theta) \tilde{x})_i > (1 + \epsilon)^{1/\alpha}$.
    Note that by the choice of $\alpha$ in \eqref{eq:alpha-def}, it holds that 
  \begin{align} \label{eq:alpha-aux}
    (2\zeta)^{\frac{1}{1-\alpha}} \leq (1 + \epsilon)^{1/\alpha} < (\theta x + (1 - \theta)\tilde{x})_i.
  \end{align}
  Combining \eqref{eq:alpha-aux} with \Cref{fact:pinning-SI} and \Cref{lemma-verify-2}, we have
  \begin{align*}
  	F_{\nu,\alpha}(x) - F_{\nu,\alpha}(\tilde{x}) = \sum_{i \in M}^m (x_i - \tilde{x}_i)\frac{\partial F}{\partial z_i}\bigg|_{z = \theta x + (1-\theta) \tilde{x}} \leq 0.
  \end{align*}
 Hence $F_{\nu,\alpha}(x) \leq F_{\nu,\alpha}(\tilde{x}) \leq 1$. This concludes the proof.
  \end{proof}

\subsection{Fractional log-concavity from complete spectral independence (proof of \Cref{lemma-verify-1})}
The following lemma was implicit in \cite{anari2021entropicII, AASV21,anari2021entropic}.

\begin{lemma}\label{lemma-aflc}
Let $\eta > 1$ and $\epsilon > 0$. 
If a distribution $\mu$ over $\{\0,\1\}^n$ is $(\eta,\epsilon)$-completely spectrally independent, 
then for any $0 < \alpha \leq 1/(2\eta)$, 
the function $\log g_{\mu^{\mathrm{hom}}}(z_1^\alpha,\ldots,z_n^\alpha,z_{\bar{1}}^\alpha,\ldots,g_{\bar{n}}^\alpha)$
is concave on 
\[
\Lambda_{\alpha,\epsilon} \triangleq \left\{(z_1,\ldots,z_{n},z_{\bar{1}},\ldots,z_{\bar{n}}) \in \mathds{R}_{>0}^{2n} \mid  \forall i \in [n], 0 < z_i \leq z_{\bar{i}}(1+\epsilon)^{1/\alpha}\right\},
\]
where ${\mu^{\mathrm{hom}}}$ is $\mu$'s homogenization over $\binom{[n]\cup [\bar{n}]}{n}$ and $g_{\mu^{\mathrm{hom}}}$ is its generating function.
\end{lemma}
The concavity property in \Cref{lemma-aflc}  is called the ``$\alpha$-fractional log-concavity''
of ${\mu^{\mathrm{hom}}}$ \cite{AASV21,anari2021entropic}.
%
%
\begin{proof}[Proof of \Cref{lemma-verify-1}]
Fix $0 < \alpha \leq 1/(2\eta)$.
Define the 1-homogeneous function as
\begin{align*}
f(z_1,\ldots,z_{n},z_{\bar{1}},\ldots,z_{\bar{n}})=g_{\mu^{\mathrm{hom}}}(z_1^\alpha,\ldots,z_n^\alpha,z_{\bar{1}}^\alpha,\ldots,z_{\bar{n}}^\alpha)^{\frac{1}{\alpha n}}.
\end{align*}
Note that $f$ is concave on $\Lambda_{\alpha,\epsilon}$.
This is because $g_{\mu^{\mathrm{hom}}}(z_1^\alpha,\ldots,z_n^\alpha,z_{\bar{1}}^\alpha,\ldots,g_{\bar{n}}^\alpha)$ is $\alpha n$-homogeneous,
and by~\Cref{lemma-aflc}, it is also log-concave as a function of $(z_1,\ldots,z_{n},z_{\bar{1}},\ldots,z_{\bar{n}})\in \Lambda_{\alpha,\epsilon}$,
which implies the concavity of $f$ on $\Lambda_{\alpha,\epsilon}$ by \cite[Lemma 25]{anari2021entropic}.
Therefore, for any $(z_1,\ldots,z_{n},z_{\bar{1}},\ldots,z_{\bar{n}}) \in \Lambda_{\alpha,\epsilon}$,
\begin{align*}
f(z_1,\ldots,z_{n},z_{\bar{1}},\ldots,z_{\bar{n}}) 
&\leq f(1,1,\ldots,1) + \sum_{i \in [n] \cup [\bar{n}]}\frac{\partial f }{\partial z_i}(1,1,\ldots,1) (z_i - 1) \\
&= \sum_{i \in [n] \cup [\bar{n}]}\frac{\partial f }{\partial z_i}(1,1,\ldots,1)z_i,
\end{align*}
where the equation holds since $f$ is 1-homogeneous.
Note that  $\mu^{\mathrm{hom}}$ is a distribution over $\binom{[n]\cup [\bar{n}]}{n}$ and 
 $\frac{\partial f }{\partial z_i}(1,1,\ldots,1) = \frac{1}{n}\Pr[S \sim \mu^{\mathrm{hom}}]{i \in S}$ for all $i \in [n] \cup [\bar{n}]$.
Therefore, for any $(z_1,\ldots,z_{n},z_{\bar{1}},\ldots,z_{\bar{n}}) \in \Lambda_{\alpha,\epsilon}$,
\begin{align*}
g_{\pi}(z_1^\alpha,\ldots,z_n^\alpha,z_{\bar{1}}^\alpha,\ldots,z_{\bar{n}}^\alpha)^{\frac{1}{\alpha n}} 
&= f(z_1,\ldots,z_{n},z_{\bar{1}},\ldots,z_{\bar{n}}) 
\leq \frac{1}{n}\sum_{i=1}^n\tp{\Pr[S\sim \mu^{\mathrm{hom}}]{i \in S} + \Pr[S \sim \mu^{\mathrm{hom}}]{\bar{i} \in S}}. 	
\end{align*}
This means that $\mu^{\mathrm{hom}}$ is $(1/\alpha)$-entropically independent over $\Lambda_{\alpha,\epsilon}$.
Then by \Cref{lemma-equvalent-EI-PD}, $\mu$ is $(1/\alpha)$-product dominated on $\left(0,(1+\epsilon)^{\frac{1}{\alpha}}\right]^n$.
By definition of product domination, for any $(x_1,\ldots,x_n) \in \left(0,(1+\epsilon)^{\frac{1}{\alpha}}\right]^n$,
\begin{align*}
 F_{\mu,\alpha}(x_1,x_2,\ldots,x_n) = \frac{g_\mu(x_1^\alpha,x_2^\alpha,\ldots,x_n^\alpha)^{{\frac{1}{\alpha}}}}{\prod_{i=1}^n \tp{\mu_i(\1) x_i + \mu_i(\0)}} \leq 1.	&\qedhere
\end{align*}
\end{proof}

It remains to formally verify \Cref{lemma-aflc}. 
A variant of the lemma was proved in \cite{anari2021entropicII} assuming the spectral domination property for correlation matrix. 
\Cref{lemma-aflc} can be proved in the same way.

\begin{definition}[\text{signed correlation matrix~\cite{anari2021entropicII}}]\label{definition-correlation-matrix}
Let $\mu$ be a distribution over $\{\0,\1\}^n$.
The correlation matrix $\Psi^{\mathrm{Cor}}_\mu \in \mathds{R}_{\geq 0}^{n \times n}$ is defined by  
\begin{align*}
\forall i,j \in [n], \quad \Psi^{\mathrm{Cor}}_\mu(i,j) = \begin{cases}
    \mu_i(\0)&\text{if } i = j;\\
    \mu_j^{i \gets \1}(\1) - \mu_j(\1) &\text{if } i\neq j \text{ and } \1 \in \Omega(\mu_i) ;\\
 	0 &\text{otherwise}.
 \end{cases}
\end{align*}	
\end{definition}
\begin{definition}[\text{signed influence matrix~\cite{anari2020spectral}}]\label{definition-sign-inf-matrix}
Let $\mu$ be a distribution over $\{-1,+1\}^{n}$.
The signed influence matrix $\Psi^{\mathrm{Inf}}_\mu \in \mathds{R}^{n \times n}$ is defined by  
\begin{align*}
\forall i,j \in [n], \quad \Psi^{\mathrm{Inf}}_\mu(i,j) = \begin{cases}
    \mu_j^{i\gets \1}(\1) - \mu_{j}^{i \gets \0}(\1) &\text{if } i \neq j \text{ and }\Omega(\mu_i)=\{\0,\1\} ;\\
    0& \text{otherwise}.
 \end{cases}
\end{align*}	
\end{definition}

\begin{remark}
The influence matrix $\Psi_\mu$ in \Cref{definition-SI} is satisfies that $\Psi_\mu(i,j) = |\Psi^{\mathrm{Inf}}_\mu(i,j)|$. 
\end{remark}

\begin{lemma}[\text{\cite[Corollary 8.1.19]{horn2012matrix}}] \label{lem:CP-SR-GE0}
  Let $A, B \in \mathds{R}^{n \times n}$ and suppose $B$ is non-negative.
  If $\abs{A} \leq B$, then $\rho\tp{A} \leq \rho\tp{\abs{A}} \leq \rho\tp{B}$.
\end{lemma}


The following relation between influence matrix and correlation matrix was proved in~\cite{AASV21}.
\begin{lemma}[\text{\cite{AASV21}}]
\label{lemma-sp-AASV}
The spectrum of $\Psi^{\mathrm{Cor}}_{\mu^{\mathrm{hom}}}$ is the union of $\{\lambda_i + 1\}_{1 \leq i \leq n}$ and $n$ copies of 0, where $\lambda_1\geq \lambda_2\geq \ldots \geq \lambda_n$ are eigenvalues of $\Psi^{\mathrm{Inf}}_{\mu}$.
\end{lemma}

\begin{proof}[Proof of \Cref{lemma-aflc}]
By the proof of Proposition~19 in~\cite{anari2021entropicII}, we only need to verify that for any $\vec{v} = (v_1,\ldots,v_{n},v_{\bar{1}},\ldots,v_{\bar{n}}) \in \Lambda_{\alpha,\epsilon}$, denoted $\vec{v}^\alpha =  (v_1^\alpha,\ldots,v_{n}^\alpha,v_{\bar{1}}^\alpha,\ldots,v_{\bar{n}}^\alpha)$, it holds that
\begin{align*}
\lambda_{\max}\tp{\Psi^{\mathrm{Cor}}_{\vec{v}^\alpha * \mu^{\mathrm{hom}}}} \leq \frac{1}{\alpha}.	
\end{align*}
Note that $\vec{v}^\alpha * \mu^{\mathrm{hom}}$ is the homogenization of $u * \mu$, where $u=(u_i)_{1\le i\le n}$ has $u_i = \tp{\frac{v_i}{v_{\bar{i}}}}^\alpha \leq (1+\epsilon)$ for all $i\in[n]$.
Therefore, we have $\lambda_{\max}\tp{\Psi^{\mathrm{Cor}}_{\vec{v}^\alpha * \mu^{\mathrm{hom}}}}= \lambda_{\max}\tp{\Psi^{\mathrm{Inf}}_{u * \mu}} +1$ by \Cref{lemma-sp-AASV}; and $\lambda_{\max}\tp{\Psi^{\mathrm{Inf}}_{u * \mu}}\le \eta$ by  \Cref{lem:CP-SR-GE0} and the $(\eta,\epsilon)$-complete spectral independence of $\mu$. Together, we have $\lambda_{\max}\tp{\Psi^{\mathrm{Cor}}_{\vec{v}^\alpha * \mu^{\mathrm{hom}}}}\le \eta+1\le \frac{1}{\alpha}$.
\end{proof}


\subsection{Monotonicity from marginal stability (proof of \Cref{lemma-verify-2})}

For any $x = (x_1,x_2,\ldots,x_n) \in \mathds{R}_{>0}^n$, 
  \begin{align*}
    \frac{\partial F_{\mu, \alpha}}{\partial z_i} \bigg|_{z = x} = \frac{F_{\mu,\alpha}(x)}{x_i}\tp{ \frac{x_i^{\alpha}}{g_\mu(x^\alpha_1,x^\alpha_2,\ldots,x^\alpha_n)} \frac{\partial g_\mu}{\partial z_i} \bigg|_{z = (x_1^\alpha,x_2^\alpha,\ldots,x_n^\alpha)} - \frac{\mu_i(\1) x_i}{\mu_i(\1) x_i + \mu_i(\0)}}.
  \end{align*}
  Observe that  
  \begin{align*}
    \frac{x_i^{\alpha}}{g_\mu(x^\alpha_1,x^\alpha_2,\ldots,x^\alpha_n)} \frac{\partial g_\mu}{\partial z_i} \bigg|_{z = (x_1^\alpha,x_2^\alpha,\ldots,x_n^\alpha)} = \frac{\sum_{\sigma: \sigma_i = \1}\mu(\sigma) \prod_{j: \sigma_j = \1} x^\alpha_j}{\sum_{\sigma}\mu(\sigma)\prod_{j: \sigma_j = \1}x_j^\alpha} = (x^\alpha * \mu)_i(\1),
  \end{align*}
  where $x^\alpha = (x_1^\alpha,x_2^\alpha,\ldots,x_n^\alpha)$. 
  Furthermore, we can assume without loss of generality that $\mu_i(\1) > 0$ and $\mu_i(\0) > 0$ because otherwise $\frac{\partial F_{\mu, \alpha}}{\partial z_i} \bigg|_{z = x} = 0$ for all $x \in \mathds{R}_{>0}^n$.
  Therefore, we have
  \begin{align*}
   \tp{ \frac{x_i}{F_{\mu,\alpha}}}\cdot  \frac{\partial F_{\mu, \alpha}}{\partial z_i} \bigg|_{z = x} &= (x^\alpha * \mu)_i(\1) - \frac{\mu_i(\1) x_i}{\mu_i(\1) x_i + \mu_i(\0)}
    = \tp{1+\frac{(x^\alpha * \mu)_i(\0)}{(x^\alpha * \mu)_i(\1)}}^{-1}- \tp{1+\frac{\mu_i(\0)}{\mu_i(\1) x_i}}^{-1}.
   \end{align*}
   Note that $(1+x)^{-1}$ is decreasing in $x > 0$. To prove $\frac{\partial F_{\mu, \alpha}}{\partial z_i} \bigg|_{z = x}  \leq 0$, it suffices to verify  $\frac{(x^\alpha * \mu)_i(\0)}{(x^\alpha * \mu)_i(\1)} \geq \frac{\mu_i(\0)}{\mu_i(\1) x_i}$, or equivalently, $\frac{(x^\alpha * \mu)_i(\1)}{(x^\alpha * \mu)_i(\0)} \leq x_i \frac{\mu_i(\1)}{\mu_i(\0)}$.
Indeed, it holds that
 \begin{align*}
  \frac{(x^\alpha * \mu)_i(\1)}{(x^\alpha * \mu)_i(\0)} &= \frac{\sum_{\sigma \in \Omega(\mu_{[n] \setminus \{i\}})}(x^\alpha * \mu)_{[n] \setminus \{i\}}(\sigma) \cdot (x^\alpha * \mu)^\sigma_i(\1)}{ \sum_{\sigma \in \Omega(\mu_{[n] \setminus \{i\}})}(x^\alpha * \mu)_{[n] \setminus \{i\}}(\sigma) \cdot (x^\alpha * \mu)^\sigma_i(\0)}\\
   &= \frac{\sum_{\sigma \in \Omega(\mu_{[n] \setminus \{i\}})}(x^\alpha * \mu)_{[n] \setminus \{i\}}(\sigma) \cdot \mu^\sigma_i(\1) \cdot x^\alpha_i}{ \sum_{\sigma \in \Omega(\mu_{[n] \setminus \{i\}})} (x^\alpha * \mu)_{[n] \setminus \{i\}}(\sigma) \cdot \mu^\sigma_i(\0)}\\
   & \leq x^\alpha_i \max_{\sigma \in \Omega(\mu_{[n] \setminus \{i\}})}\frac{\mu^\sigma_i(\1)}{\mu^{\sigma}_i(\0)}.
 \end{align*}
In above, we enumerate all $\sigma$ in $\Omega(\mu_{[n] \setminus \{i\}})$ because $\mu$ and $(x^\alpha * \mu)$ have the same support.
By \Cref{def:bounded-distribution}, it holds that for all possible partial configuration $\sigma$, we have $R^\sigma_i + R^\sigma_i / R_i \leq 2\zeta$, where $R_i=\mu_i(+1)/\mu_i(-1)$, which implies $R^\sigma_i \leq 2\zeta (1 + 1/R_i)^{-1} = 2\zeta \cdot \mu_i(\1)$.
Hence, it holds that,
\begin{align*}
  \frac{(x^\alpha * \mu)_i(\1)}{(x^\alpha * \mu)_i(\0)} \leq x^\alpha_i\max_{\sigma \in \Omega(\mu_{[n] \setminus \{i\}})}\frac{\mu^\sigma_i(\1)}{\mu^{\sigma}_i(\0)} \leq x_i^\alpha \cdot 2\zeta \cdot \mu_{i}(+1) \leq x_i \frac{\mu_i(\1)}{\mu_i(\0)},
\end{align*}
where the last inequality comes from the fact that $x_i^{\alpha - 1} \cdot 2\zeta \leq 1$, which is guaranteed by $x_i \geq (2\zeta)^{\frac{1}{1 - \alpha}}$.


\section{Invariants of $k$-Transformation}\label{section-muk}
In this section, we prove \Cref{lem:mu-to-muk}, 
that the spectral independence and marginal stability properties stated in \Cref{condition:CSI-MS} are roughly invariant under $k$-transformation (\Cref{def:k-trans}).

This is proved by two lemmas. Recall that concept of complete spectral independence (\Cref{definition-complete-SI}).
\begin{lemma}\label{lemma-part-1}
Let $\eta,\epsilon > 0$.
If a distribution $\mu$ over $\{\0,\1\}^{n}$ is $(\eta,\epsilon)$-completely spectrally independent and 	 
\begin{align*}
\mu_{-1}^{\min} \triangleq \min_{i \in [n]}\min_{\sigma \in \Omega(\mu_{[n]\setminus \{i\}})}\mu^\sigma_i(-1) > 0, 
\end{align*}
then there exists a finite $k_0 = 10(1+\epsilon)/\mu_{-1}^{\min}$ such that for all integers $k \geq k_0$, the $k$-transformed distribution $\mu_k= \Rd(\mu, k)$ is $(2\eta+5,\epsilon)$-completely spectrally independent.
\end{lemma}

To state the next lemma, we also define the concept of complete marginal stability.
\begin{definition}[complete marginal stability]\label{definition-complete-MS}
  Let $\zeta > 1$. A distribution $\mu$ over $\{\0,\1\}^{n}$ is said to be completely $\zeta$-marginally stable if $(\*\lambda * \mu)$ is $\zeta$-marginally stable for any $\*\lambda \in (0,1]^n$.
\end{definition}

\begin{lemma}\label{lemma-part-II} 
  Let $\zeta > 1$. 
  If a distribution $\mu$ over $\{\0,\1\}^{n}$ is completely $\zeta$-marginally stable, then for any integer $k \geq 1$, the $k$-transformed distribution $\mu_k= \Rd(\mu, k)$ is $2\zeta$-marginally stable.
\end{lemma}

\Cref{lem:mu-to-muk} follows immediately from \Cref{lemma-part-1} and \Cref{lemma-part-II}.
\begin{proof}[Proof of \Cref{lem:mu-to-muk}]
  By \Cref{condition:CSI-MS},  $\mu$ is completely $\zeta$-marginally stable. Then for any $i \in [n]$ and $\sigma \in \Omega(\mu_{[n] \setminus \{i\} })$, it holds that
  \begin{align*}
    \mu_i^\sigma(\0)=\tp{\frac{\mu_i^\sigma(\1)}{\mu_i^\sigma(\0)} + 1}^{-1} 
    \ge \frac{1}{1+\zeta}.
  \end{align*}
  Let $k_0 = 10(1+\epsilon)(1+\zeta)$. By \Cref{lemma-part-1}, the distribution $\mu_k$ is $(2\eta+5,\epsilon)$-completely spectrally independent for all $k \ge k_0$. 
  By \Cref{lemma-part-II}, the distribution $\mu_k$ is $4\zeta^2$-marginally stable for all $k\geq 1$. 
\end{proof}

\subsection{Complete spectral independence of $\mu_k$ (proof of \Cref{lemma-part-1})} 
The  \emph{correlation matrix} (\Cref{definition-correlation-matrix}) was introduced in~\cite{AASV21}. 
We consider the \emph{absolute correlation matrix}.
\begin{definition}[\text{absolute correlation matrix~\cite{anari2021entropicII}}]\label{definition-abs-correlation-matrix}
Let $\mu$ be a distribution over $\{-1,+1\}^{[n]}$.
The absolute correlation matrix $\Psi^{\mathrm{AbsCor}}_\mu \in \mathds{R}_{\geq 0}^{[n] \times [n]}$ is defined by  
\begin{align*}
\forall i,j \in [n], \quad \Psi^{\mathrm{AbsCor}}_\mu(i,j) \triangleq \abs{\Psi^{\mathrm{Cor}}_\mu(i,j)},
\end{align*}	
where $\Psi^{\mathrm{Cor}}_\mu$ is the correlation matrix in \Cref{definition-correlation-matrix}.
\end{definition}

The  spectral independence (\Cref{definition-SI}) and complete spectral independence (\Cref{definition-complete-SI}) are defined using absolute influence matrix. Similarly, we can define using absolute correlation matrix.  

\begin{definition}[limited correlation]
Let $\eta,\eps > 0$.
A distribution $\mu$ is said to have \emph{$\eta$-limited correlation} if for any $\Lambda \subseteq V$, any $\sigma \in \Omega(\mu_\Lambda)$, the spectral radius of the absolute correlation matrix $\Psi^{\mathrm{AbsCor}}_{\mu^\sigma}$ satisfies
\begin{align*}
\rho\tp{\Psi^{\mathrm{AbsCor}}_{\mu^\sigma}} \leq \eta.
\end{align*}
$\mu$ is said to have \emph{$(\eta,\eps)$-complete limited correlation} if $(\*\lambda * \mu)$ has $\eta$-limited correlation for all $\*\lambda \in (0,1+\epsilon]^n$.
\end{definition}


\begin{lemma}\label{lemma-SI-to-Cor}
  Let $\eta, \epsilon > 0$.
  Let $\mu$ be a distribution over $\{\0, \1\}^{[n]}$.
  If $\mu$ is $(\eta, \epsilon)$-completely spectrally independent, then $\mu$ has $(\eta + 1, \epsilon)$-complete limited correlation.
\end{lemma}

\begin{lemma} \label{lem:muk-C-SI}
 Let $\eta, \epsilon > 0$.
  Let $\mu$ be a distribution over $\{\0, \1\}^{[n]}$ satisfying
 \begin{align}
 \label{eq-def-mu-min-mi}
\mu_{-1}^{\min} \triangleq \min_{v \in [n]}\min_{\sigma \in \Omega(\mu_{[n]\setminus \{v\}})}\mu^\sigma_v(-1) > 0.
\end{align}
Let $k_0 = k_0(\mu,\epsilon) = {10(1+\epsilon)}/{\mu^{\min}_{-1}} > 0$ be a finite real number.
  If $\mu$ has $(\eta, \epsilon)$-complete limited correlation, then for each integer $k \geq k_0$, it holds that $\mu_k$ is $(2\eta + 3, \epsilon)$-completely spectrally independent.
\end{lemma}

\Cref{lemma-part-1} is a straightforward corollary of the above two lemmas.
We then prove \Cref{lemma-SI-to-Cor} in \Cref{sec-lemma-SI-to-Cor}, and prove \Cref{lem:muk-C-SI} in \Cref{sec-lem:muk-C-SI} respectively.

\subsubsection{Proof of \Cref{lemma-SI-to-Cor}}\label{sec-lemma-SI-to-Cor}
The following lemma is a well-known fact for non-negative matrix.

\begin{lemma}[\text{\cite[Lemma 8.3.1]{horn2012matrix}}]
  Let $A \in \mathds{R}^{n \times n}_{\ge 0}$ be a non-negative matrix. The spectral radius $\rho(A)$ equals to the maximum eigenvalue $\lambda_{\max}(A)$. 
  Consequently, $\rho(A+I) = \rho(A) + 1$.
\end{lemma}
The next lemma is the relation between the influence matrix and the correlation matrix. 
Recall that the signed influence matrix is defined in~\Cref{definition-sign-inf-matrix}.

\begin{lemma} \label{lem:cor-inf}
  Let $\mu$ be a distribution over $\{\0, \1\}^{[n]}$ satisfying $\mu_i(\0) > 0$ for all $i \in [n]$, it holds that 
  \begin{align*}
  	\Psi^{\-{Inf}}_\mu = \-{diag}^{-1}\tp{\{\mu_i(\0)\}_{i\in [n]}} \Psi^{\-{Cor}}_\mu - I,
  \end{align*}
  where $\-{diag}^{-1}\tp{\{\mu_i(\0)\}_{i\in [n]}}$ is a diagonal matrix satisfying $\-{diag}^{-1}\tp{\{\mu_i(\0)\}_{i\in [n]}}(i,i) = \frac{1}{\mu_i(\0)}$, and $I$ is the $n$-by-$n$ identity matrix.
\end{lemma}

\begin{proof}
Suppose $i \neq j$. If $\Omega(\mu_i)= \{\0\}$ or $\Omega(\mu_i) = \{\1\}$, then it holds that $\Psi^{\mathrm{Inf}}_\mu(i,j) = \Psi^{\mathrm{Cor}}_\mu(i,j) = 0$. Suppose $\Omega(\mu_i) = \{\0,\1\}$. It holds that
\begin{align*}
  \Psi^{\-{Inf}}_\mu(i, j)
  &=  \mu_j^{i\gets \1}(\1) - \mu_{j}^{i \gets \0}(\1)\\
  &= \frac{\Pr[X \sim \mu]{X_i = \1 \land X_j = \1}}{\Pr[X \sim \mu]{X_i = \1}} - \frac{\Pr[X\sim \mu]{X_j = \1} -\Pr[X \sim \mu]{X_i = \1 \land X_j = \1}}{\Pr[X \sim \mu]{X_i = \0}}\\
  &= \frac{\Pr[X \sim \mu]{X_i = \1 \land X_j = \1} - \Pr[X\sim \mu]{X_i= \1}\Pr[X\sim \mu]{X_j = \1}}{\Pr[X \sim \mu]{X_i = \0}\Pr[X \sim \mu]{X_i = \1}}\\
  &= \frac{1}{\mu_i(\0)} \Psi^{\-{Cor}}_\mu (i, j).
\end{align*}
By definition, if $i=j$, then $\Psi^{\-{Inf}}_\mu(i,i) = 0$, and thus we have $\Psi^{\-{Inf}}_\mu = \-{diag}^{-1}\tp{\{\mu_i(\0)\}_{i\in [n]}} \Psi^{\-{Cor}}_\mu - I$.
\end{proof}

Now, we are ready to prove \Cref{lemma-SI-to-Cor}.
\begin{proof}[Proof of \Cref{lemma-SI-to-Cor}]
We use $\Psi^{\-{AbsInf}}_{\cdot}$ to denote the absolute influence matrix in \Cref{definition-SI}.
By \Cref{definition-complete-SI} and \Cref{definition-abs-correlation-matrix}, it suffices to prove that for any distribution $\mu$ over $\{\0,\1\}^{[n]}$, it holds that 
\begin{align}\label{eq-rhocor-rhoinf}
	\rho\tp{\Psi^{\mathrm{AbsCor}}_\mu} \leq \rho\tp{\Psi^{\mathrm{AbsInf}}_\mu} + 1.
\end{align}
\Cref{lemma-SI-to-Cor} is a straightforward corollary of the above inequality. 

Note that for any $i \in [n]$ such that $\mu_i(\0) = 0$, the $i$-th row and $i$-th column in $\Psi^{\mathrm{Cor}}_{\mu}$ and $\Psi^{\mathrm{Inf}}_\mu$ are all 0. 
Hence, it suffices to consider $\Psi^{\mathrm{Cor}}_{\mu_S}$  and $\Psi^{\mathrm{Inf}}_{\mu_S}$, where $S = \{ i \in [n] \mid \mu_i(\0) > 0\}$.
Without loss of generality, we can assume that the distribution $\mu$ satisfies $\mu_i(\0) > 0$ for all $i \in [n]$. By \Cref{lem:cor-inf},
\begin{align*}
\Psi^{\-{Cor}}_\mu  = \-{diag}\tp{\{\mu_i(\0)\}_{i\in [n]}}(\Psi^{\-{Inf}}_\mu + I)
\end{align*}
Note that $\Psi^{\-{Inf}}_\mu(i,i) = 0$ for all $i \in [n]$. For any $i,j \in [n]$, it holds that
\begin{align*}
\Psi^{\-{AbsCor}}_\mu(i,j) \leq  \Psi^{\-{AbsInf}}_\mu(i,j) + I(i,j) 
\end{align*}
because $0 < \mu_i(\0) \le 1$ for all $i \in [n]$. This implies~\eqref{eq-rhocor-rhoinf}.
\end{proof}

\subsubsection{Proof of \Cref{lem:muk-C-SI}}\label{sec-lem:muk-C-SI}
We use the following definitions and lemmas to prove \Cref{lem:muk-C-SI}.
Let $\mu$ be a distribution over $\{\0,\1\}^{n}$.
For any integer $k \geq 1$, let $\mu_k$ denote the $k$-transformation of $\mu$ (\Cref{def:k-trans}).
We use $V = [n]$ to denote the variable set of $\mu$ and $V_k = [n] \times [k]$ to denote the variable set of $\mu_k$. 
For each $v \in [n]$ and $i \in [k]$, we use $v_i$ to denote the pair $(v, i) \in V_k$. 
For any $\Lambda \subseteq V_k$, we use $\mu_{k,\Lambda}$ to denote the marginal distribution on $\Lambda$ projected from $\mu_k$. 
We simply denote $\mu_{k,\{v_i\}}$ by $\mu_{k,v_i}$.

\begin{lemma} \label{lem:mu-C-Cor => muk-C-Cor}
Let $\eta,\epsilon > 0$.
  If $\mu$  has $(\eta, \epsilon)$-complete limited correlation, 
  then for any integer $k \geq 1$, $\mu_k$ has $(\eta + 2, \epsilon)$-complete limited correlation.
\end{lemma}
\begin{lemma} \label{lem:muk-pin-ratio}
Let $\epsilon \geq 0$ and $k \in \mathds{Z}_{> 0}$. 
Let $\*z \in (0, \epsilon]^{V_k}$, $\Lambda \subseteq V_k$, $v_i \in V_k \setminus \Lambda$.
For any $\sigma \in \Omega(\mu_{k, \Lambda})$ where $\mu_{k, v_i}^\sigma(\1) > 0$, there exist $\*x \in (0, \epsilon]^V$ satisfying $x_v = 1$, a subset $R \subseteq V$ satisfying $v \notin R$, and a  partial configuration $\tau \in \Omega(\mu_R)$ such that 
  \begin{align*}
    \frac{\tp{\*z * \mu_k}^{\sigma}_{v_i}(\0)}{\tp{\*z * \mu_k}^{\sigma}_{v_i}(\1)}
    &= \frac{k}{z_{v_i}}\tp{\frac{\tp{\*x * \mu}^{\tau}_v(\0)}{\tp{\*x * \mu}^\tau_v(\1)} + \frac{1}{k} \sum_{v_j \in C_v \setminus (\Lambda \cup\{v_i\})} z_{v_j}},
  \end{align*}
where  $C_v \triangleq \{v_i \mid i \in [k]\}$.
\end{lemma}

\Cref{lem:mu-C-Cor => muk-C-Cor} can be proved by going through the proof of \cite[Proposition 26]{anari2021entropicII}.
\Cref{lem:muk-pin-ratio} is a technical lemma that relates $\mu$ to $\mu_k$ with local fields and pinnings.
We first use \Cref{lem:mu-C-Cor => muk-C-Cor} and \Cref{lem:muk-pin-ratio} to prove \Cref{lem:muk-C-SI}, and then prove \Cref{lem:mu-C-Cor => muk-C-Cor} and \Cref{lem:muk-pin-ratio}.
 




\begin{proof}[Proof of \Cref{lem:muk-C-SI} assuming  \Cref{lem:mu-C-Cor => muk-C-Cor} and \Cref{lem:muk-pin-ratio}]
  By \Cref{lem:mu-C-Cor => muk-C-Cor}, we know that $\mu_k$ has $(\eta + 2, \epsilon)$-complete  limited correlation.
  Fix $\Lambda \subseteq V_k$, $\sigma \in \Omega(\mu_{k, \Lambda})$, and $\*z \in (0, 1 + \epsilon]^{V_k}$. It holds that
  \begin{align*}
    \rho\tp{\Psi^{\-{AbsCor}}_{\tp{\*z * \mu_k}^\sigma}} \leq \eta + 2.
  \end{align*}
  Let $\pi = \tp{\*z * \mu_k}^\sigma_{V_k \setminus \Lambda}$, which is obtained by projecting $\tp{\*z * \mu_k}^\sigma$ on subset $V_k \setminus \Lambda$. 
  By the definition of absolute correlation matrix, for any $v_i \in \Lambda$, the row and the column in $\Psi^{\-{AbsCor}}_{(\*z * \mu)^\sigma}$ corresponding to $v_i$ only contain zeros. We have
  \begin{align*}
  	\rho\tp{\Psi^{\-{AbsCor}}_{\pi}} \leq \eta + 2.
  \end{align*}
  Let $k_0(\mu,\epsilon) = {10(1+\epsilon)}/{\mu^{\min}_{-1}}$, where $\mu^{\min}_{-1}$ is defined in \eqref{eq-def-mu-min-mi}. Note that $k_0$ is finite because $\mu^{\min}_{-1} > 0$.
  For all integer $k \geq k_0$,  we claim that
  \begin{align} \label{eq:aux-inf-cor}
    \min_{v_i \in V_k \setminus \Lambda} \pi_{v_i}(\0) \geq \frac{2}{3}.
  \end{align}
  Recall that we use $\Psi^{\-{AbsInf}}_{\cdot}$ to denote the absolute influence matrix in \Cref{definition-SI}.
  Then, by \Cref{lem:cor-inf}, 
  \begin{align*}
    \rho\tp{\Psi^{\-{AbsInf}}_{\pi}}
    &= \rho\tp{\-{diag}^{-1}(\{\pi_{v_i}(\0)\}_{v_i \in V_k \setminus \Lambda}) \cdot \Psi^{\-{AbsCor}}_{\pi} - I}
      \end{align*}
Note that the diagonal of $\-{diag}^{-1}(\{\pi_{v_i}(\0)\}_{v_i \in V_k \setminus \Lambda}) \cdot \Psi^{\-{AbsCor}}_{\pi} - I$ are a set of zeros. By \Cref{lem:CP-SR-GE0} and \eqref{eq:aux-inf-cor},
\begin{align*}
 \rho\tp{\Psi^{\-{AbsInf}}_{\pi}} \leq \rho\tp{\frac{3}{2} \Psi^{\-{AbsCor}}_{\pi}} = \frac{3}{2} \rho\tp{\Psi^{\-{AbsCor}}_{\pi}} \leq 2\eta + 3.
\end{align*}
By the definition of absolute influence matrix, for any $v_i \in \Lambda$, the row and the column in $\Psi^{\-{AbsInf}}_{(\*z * \mu_k)^\sigma}$ corresponding to $v_i$ only contain zeros. We have
\begin{align*}
\rho\tp{\Psi^{\-{AbsInf}}_{(\*z * \mu_k)^\sigma}} \leq 2\eta + 3.
\end{align*}

  Finally, we only need to verify inequality \ref{eq:aux-inf-cor}. 
  To do this, we only need to show that for each $v_i \in V_k \setminus \Lambda$, it holds that
  \begin{align*}
   \frac{\tp{\*z * \mu_k}^\sigma_{v_i}(\1)}{\tp{\*z * \mu_k}^\sigma_{v_i}(\0)} \leq \frac{1}{2}.
  \end{align*}
  When $\mu_{k,v_i}^\sigma(\1) = 0$, this holds trivially.
  Otherwise when $\mu_{k,v_i}^\sigma(\1) > 0$, by \Cref{lem:muk-pin-ratio}, there exists $\*x \in \mathds{R}^{[n]}_{>0}$ where $x_v = 1$ and a feasible partial configuration $\tau \in \Omega(\mu_R)$, where $R \subseteq V$ and $v \notin R$, such that
  \begin{align*}
    \frac{\tp{\*z * \mu_k}^\sigma_{v_i}(\1)}{\tp{\*z * \mu_k}^\sigma_{v_i}(\0)} 
    &= \frac{z_{v_i}}{k}\tp{\frac{\tp{\*x * \mu}^{\tau}_v(\0)}{\tp{\*x * \mu}^\tau_v(\1)} + \frac{1}{k} \sum_{v_j \in C_v \setminus (\Lambda \cup\{v_i\})} z_{v_j}}^{-1}\\
    &\leq \frac{z_{v_i}}{k} \frac{\tp{\*x * \mu}^\tau_v(\1)}{\tp{\*x * \mu}^\tau_v(\0)}\\
    (\ast)\quad &\leq \frac{1 + \epsilon}{k} \max_{\sigma \in \Omega(\mu_{V \setminus \{v\}})} \frac{\mu^\sigma_v(\1)}{\mu^\sigma_v(\0)}\\
     & \leq \frac{1 + \epsilon}{k} \max_{\sigma \in \Omega(\mu_{V \setminus \{v\}})} \frac{1}{\mu^\sigma_v(\0)} = \frac{1+\epsilon}{k \mu^{\min}_{-1}},  
  \end{align*}
  which is less than $\frac{1}{2}$ when $k \geq \frac{10(1+\epsilon)}{\mu^{\min}_{-1}}$.
  Inequality $(\ast)$ holds because (1) $0 < z_{v_i} \leq 1 + \varepsilon$; (2) the fact that the value of $v$ is not fixed by $\tau$; (3) $x_v = 1$.
\end{proof}

A version of \Cref{lem:mu-C-Cor => muk-C-Cor} with signed correlation matrix was proved  in~\cite{anari2021entropicII}.
We give a proof of \Cref{lem:mu-C-Cor => muk-C-Cor} by applying the same argument there. We include the proof for completeness.
\begin{proof}[Proof of \Cref{lem:mu-C-Cor => muk-C-Cor}]
Let $\mu$ be a distribution over $\{\0,\1\}^{V}$, where $V = [n]$.
Define a more general $\vec{k} = (k_1,k_2,\ldots,k_n) \in \mathds{Z}_{> 0}$ transformation, which transforms $\mu$ to a new distribution $\mu_{\vec{k}}$, where $\mu_{\vec{k}}$ is defined over $\{\0,\1\}^{V_{\vec{k}}}$ 
and $V_{\vec{k}} = \{(i,j) \mid 1 \le i \le n, 1 \le j \le k_i\}$.
For each $v \in [n]$, $i \in [k_v]$, we use $v_i$ to denote $(v,i)$.
To sample $\*Y\sim \mu_{\vec{k}}$, we first sample $\*X \sim \mu$, and then for any $v \in V$
\begin{itemize}
	\item if $X_v = \0$, then let $Y_{v_i} = \0$ for all $i \in [k_v]$;
	\item if $X_v = \1$, then sample $j^* \in [k_v]$ u.a.r., set $Y_{v_{j^*}} = \1$ and $Y_{v_j} = \0$ for all $j \in [k_v] \setminus \{j^*\}$.
\end{itemize}
It is straightforward to verify the $k$-transformation in \Cref{def:k-trans} is a special case when $\vec{k}$ is a constant vector with value $k$. 

We prove the following results. 
For any $\vec{k} \in \mathds{Z}_{> 0}^n$, any $\*y \in \mathds{R}^{k_1 + \cdots + k_n}_{> 0}$, it holds that 
  \begin{align} \label{eq:abs-cor-ineq}
    \rho\tp{\Psi^{\-{AbsCor}}_{\*y * \mu_{\vec{k}}}} &\leq \rho\tp{\Psi^{\-{AbsCor}}_{\*x * \mu}} + 2, \text{ where } \forall i \in [n], \*x_i \triangleq \frac{1}{k_i} \sum_{j = 1}^{k_i} y_{(i,j)}.
  \end{align}
  
  We first use~\eqref{eq:abs-cor-ineq} to prove the lemma. 
  We need to prove that for any partial configuration $\sigma \in \{\0\,\1\}^\Lambda$ of $\mu_k$, where $\Lambda \subseteq V_k$, it holds that for any $\*z \in (0,1+\epsilon]^{nk}$,
  \begin{align*}
  	 \rho\tp{\Psi^{\-{AbsCor}}_{(\*z * \mu_k)^\sigma}} =  \rho\tp{\Psi^{\-{AbsCor}}_{\*z *\mu_{k}^\sigma }} \le \eta+2. 
  \end{align*}
  By \cite[Lemma 15]{anari2021entropicII}, for any feasible condition $\sigma \in \{\0\,\1\}^\Lambda$ with respect to $\mu_k$, there exists a feasible condition $\tau$ with respect to $\mu$, local fields $\*\lambda \in (0,1]^n$ together with a vector $\vec{k} \in \mathds{Z}_{> 0}$ such that 
  \begin{align*}
  	\rho\tp{\Psi^{\-{AbsCor}}_{\mu_{k}^\sigma }} = \rho\tp{\Psi^{\-{AbsCor}}_{(\*\lambda * \mu^\tau)_{\vec{k}}}},
  \end{align*}
  where $(\*\lambda * \mu^\tau)_{\vec{k}}$ is obtained by applying $\vec{k}$-transformation on $\*\lambda * \mu^\tau$. 
  Using~\eqref{eq:abs-cor-ineq} on $\mu^\sigma_k$ implies that 
  \begin{align*}
  	\rho\tp{\Psi^{\-{AbsCor}}_{\*z * \mu_{k}^\sigma }} \leq \rho\tp{\Psi^{\-{AbsCor}}_{(\*x \odot \*\lambda) * \mu^\tau}} + 2,
  \end{align*}
  where for all $i \in [n]$, $x_i \triangleq \frac{1}{k_i} \sum_{j = 1}^{k_i} z_{(i,j)}$ and $(\*x \odot \*\lambda) \in (0,1+\epsilon]^{n}$ satisfying $(\*x \odot \*\lambda)_v = x_v\lambda_v \leq 1 + \epsilon$. 
  Since $\mu$ has $(\eta, \epsilon)$-complete limited correlation, we have 
  \begin{align*}
  	 	\rho\tp{\Psi^{\-{AbsCor}}_{\mu_{k}^\sigma }} = \rho\tp{\Psi^{\-{AbsCor}}_{(\*\lambda * \mu^\tau)_{\vec{k}}}} \leq \rho\tp{\Psi^{\-{AbsCor}}_{(\*x \odot \*\lambda) * \mu^\tau}} + 2 = \rho\tp{\Psi^{\-{AbsCor}}_{((\*x \odot \*\lambda) * \mu)^\tau}} + 2 \leq \eta +2.
  \end{align*}



  Now, we only need to verify \Cref{eq:abs-cor-ineq}.
  For convenience, we denote $\Psi^{\-{AbsCor}}_{\*y * \mu_{\vec{k}}}$ as $\Psi_{\vec{k}}$ and $\Psi^{\-{AbsCor}}_{\*x * \mu}$ as $\Psi$ respectively.
  Without loss of generality, we may assume $\1 \in \Omega(\mu_i)$ for all $i \in [n]$. Suppose $\1 \notin \Omega(\mu_i)$ for some $i \in [n]$. 
  Then the $i$-th row and the $i$-th column of $\Psi$ are all zeros, and the rows $(i,j)$ and columns $(i,j)$ for $j \in [k_i]$ in $\Psi_{\vec{k}}$ are all zeros. 
  Hence, the variable $i$ and all variables $(i,j)$ for $j \in [k_i]$ have fixed value and they do not affect the spectral radiuses of correlation matrices.
  In this case, we can simply consider the distribution $\mu_{[n] \setminus \{i\}}$ and its transformations. 
  
  Let $\widehat{\Psi}_{\vec{k}}$ be a matrix with the same size as $\Psi_{\vec{k}}$ defined as
  \begin{align*}
    \forall u,v\in [n], i\in[k_u],j\in[k_v], \quad \widehat{\Psi}_{\vec{k}}(u_i, v_j) \triangleq
    \begin{cases}
      1 + \textstyle(\*y * \mu_{\vec{k}})_{u_i}(\1), & u = v \text{ and } i = j; \\
      \Psi_{\vec{k}}(u_i, v_j), & \text{otherwise.}
    \end{cases}
  \end{align*}
  From this definition, we know that $\Psi_{\vec{k}}(u_i,v_j) \leq \widehat{\Psi}_{\vec{k}}(u_i,v_j)$, and by \Cref{lem:CP-SR-GE0}, it holds that 
  \begin{align}\label{eq-rho-1}
  	\rho\tp{\Psi_{\vec{k}}} \leq \rho\tp{\widehat{\Psi}_{\vec{k}}}.
  \end{align}
  Let $\widehat{\Psi}$ be another matrix with the same size as $\Psi$ defined as
  \begin{align*}
    \forall u,v \in [n], \quad \widehat{\Psi}(u, v) \triangleq \sum_{h \in [k_v]} \widehat{\Psi}_{\vec{k}}(u_1, v_h).
  \end{align*}
  In the above definition,  $\widehat{\Psi}(u,v)$ is the sum over all $\widehat{\Psi}_{\vec{k}}(u_1, v_h)$ for $h \in [k_v]$. The following claim shows that the $u_1$ in the definition can be replaced by any $u_i$ for $i \in [k_u]$. The claim will be proved later. 
  \begin{claim}\label{claim-u1-ui}
  	For any $ u,v \in [n]$ and $i \in [k_u]$, it holds that $\widehat{\Psi}(u, v) = \sum_{h \in [k_v]} \widehat{\Psi}_{\vec{k}}(u_i, v_h)$.
  \end{claim}

  To prove~\eqref{eq:abs-cor-ineq}, we prove the following two inequalities
  \begin{align}
  	\rho\tp{\widehat{\Psi}} &\leq \rho\tp{\Psi} + 2 \label{eq-rho-2}\\
  	\rho\tp{\widehat{\Psi}_{\vec{k}}} &\leq \max\left\{\rho\tp{\widehat{\Psi}}, 1\right\}. \label{eq-rho-3}
  \end{align}
   Combining~\eqref{eq-rho-1},~\eqref{eq-rho-2} and~\eqref{eq-rho-3}, we have
  \begin{align*}
  		\rho\tp{\Psi_{\vec{k}}}  \leq \rho\tp{\widehat{\Psi}_{\vec{k}}} \leq \max\left\{\rho\tp{\widehat{\Psi}}, 1\right\} \leq \rho\tp{\Psi} + 2,
  \end{align*}
  which proves~\eqref{eq:abs-cor-ineq}.

  We first prove \eqref{eq-rho-2}.
    By the definition of $\widehat{\Psi}$, we know that for any $u,v \in [n]$, if $u = v$, it holds that 
  \begin{align*}
  	\widehat{\Psi}(u, u) &= 1 + (\*y * \mu_{\vec{k}})_{u_1}(\1) + \sum_{j \in [k_u] \setminus \{1\}}\Psi_{\vec{k}}(u_1, u_j) \overset{(\ast)}{\leq} 1 + \sum_{j \in [k_u]}(\*y * \mu_{\vec{k}})_{u_j}(\1).
  \end{align*}
  where $(\ast)$ holds because it is straightforward to see $\Psi_{\vec{k}}(u_1, u_j) = (\*y * \mu_{\vec{k}})_{u_j}(\1)$ for all $j \in [k_u] \setminus \{1\}$.
  We have the following claim about the distribution $(\*y * \mu_{\vec{k}})$.
  \begin{claim}\label{claim-muk}
  For any distinct $u,v \in [n]$, any $i \in [k_u]$ and $j \in [k_v]$, it holds that
  \begin{align*}
  	(\*y * \mu_{\vec{k}})^{u_i \gets \1}_{v_j} (\1) = \frac{y_{v_j}(\*x * \mu)^{u \gets 1}_v(\1)}{\sum_{\ell \in [k_v]}y_{v_\ell}} \quad\text{and}\quad (\*y * \mu_{\vec{k}})_{v_j} (\1) = \frac{y_{v_j}(\*x * \mu)_v(\1)}{\sum_{\ell \in [k_v]}y_{v_\ell}},
  \end{align*}  
  where $\*x$ is defined in~\eqref{eq:abs-cor-ineq}.
  \end{claim}
  
  We first prove the lemma assuming \Cref{claim-muk}, and then prove \Cref{claim-muk}.
  By~\Cref{claim-muk}, it is straightforward to verify that  
   \begin{align}\label{eq-relation-1}
  	\widehat{\Psi}(u, u)  = 1 + \sum_{j \in [k_u]}(\*y * \mu_{\vec{k}})_{u_j}(\1) =1+ (\*x * \mu)_u(\1).
   \end{align}


  For any  $u,v \in [n]$, if $u \neq v$, it holds that 
  \begin{align}\label{eq-ind-1}
  	\widehat{\Psi}(u, v) &= \sum_{h \in [k_v]} \widehat{\Psi}_{\vec{k}}(u_1, v_h) = \sum_{h \in [k_v]} \abs{(\*y * \mu_{\vec{k}})^{u_1 \gets \1}_{v_h}(\1) - (\*y * \mu_{\vec{k}})_{v_h}(\1)}\notag\\
  	&=\sum_{h \in [k_v]}\abs{\frac{y_{v_h}\tp{(\*x * \mu)_v^{u \gets \1}(\1) - (\*x * \mu)_v(\1)}}{\sum_{\ell \in [k_v] }y_{v_\ell}}} = \abs{(\*x * \mu)_v^{u \gets \1}(\1) - (\*x * \mu)_v(\1)},
  \end{align}
  where~\eqref{eq-ind-1} holds because of \Cref{claim-muk},
  which implies that 
  \begin{align}\label{eq-relation-2}
  \forall u,v \in [n] \text{ with } u \neq v, \quad 	\widehat{\Psi}(u, v) = \Psi(u,v).
  \end{align}
  Combining~\eqref{eq-relation-1} and~\eqref{eq-relation-2}, we have
  \begin{align*}
  \forall u,v \in [n], \quad \widehat{\Psi}(u,v) \leq \Psi(u,v) + 2I(u,v).
  \end{align*}
  Since both $\widehat{\Psi}$ and $\Psi$ are non-negative matrices,  by \Cref{lem:CP-SR-GE0},~\eqref{eq-rho-2} holds.

  Now, we prove~\eqref{eq-rho-3}. 
  By~\eqref{eq-relation-2} and the definition of correlation matrix, it is straightforward to verify $\-{diag}(\{(\*x * \mu)_i(\1)\}_{i \in [n]}) \widehat{\Psi}$ is a symmetric matrix.

  Hence, $\widehat{\Psi}$ has an orthogonal eigenbasis $f^1, \cdots, f^n$ with respect to the inner product $\inner{\cdot}{\cdot}_{(\*x * \mu)}$ with corresponding real eigenvalues $\lambda_1 \geq \cdots \geq \lambda_n$, where the inner product $\inner{\cdot}{\cdot}_{(\*x * \mu)}$ is defined by
  \begin{align*}
  \forall a,b\in \mathds{R}^{[n]},\quad
  	\inner{a}{b}_{(\*x * \mu)}  = \sum_{i = 1}^n a_i b_i (\*x * \mu)_i(\1). 
  \end{align*}
  One could verify that for each eigenvector $f^t = (f^t_1, \cdots, f^t_n) \in \mathds{R}^n$, the vector
  \begin{align*}
    F^t = (\underbrace{f^t_1, \cdots, f^t_1}_{\text{$k_1$ times}}, \underbrace{f^t_2, \cdots, f^t_2}_{\text{$k_2$ times}}, \cdots, \underbrace{f^t_n, \cdots, f^t_n}_{\text{$k_n$ times}})
  \end{align*}
  is an eigenvector of $\widehat{\Psi}_{\vec{k}}$ with eigenvalue $\lambda_t$.
  That is, for any $u \in [n], i \in [k_u]$, it holds that 
  \begin{align*}
    \tp{\widehat{\Psi}_{\vec{k}}(u_i, \cdot) F^t}_{u_i}
    = \sum_{v \in [n]} f^t_v \sum_{j \in [k_v]} \widehat{\Psi}_{\vec{k}}(u_i, v_j)
    \overset{(\star)}{=} \sum_{v \in [n]} \widehat{\Psi}(u, v) f^t_v 
    = \widehat{\Psi}(u, \cdot) f^t
    = \lambda_t f^t_u = \lambda_t F^t_{u_i}.
  \end{align*}
  Equation $(\star)$ holds due to \Cref{claim-u1-ui}, i.e. 
  \begin{align*}
 \forall i \in [k_u],\quad   \widehat{\Psi}(u, v) = \sum_{j \in [k_v]} \widehat{\Psi}_{\vec{k}}(u_1, v_j) \overset{(\ast)}{=} \sum_{j \in [k_v]} \widehat{\Psi}_{\vec{k}}(u_i, v_j).
  \end{align*}
  Moreover, for each $u \in [n]$, if we pick $f[u]^{1}, f[u]^{2}, \cdots, f[u]^{k_u - 1}$ as an orthogonal basis with respect to the inner product $\inner{\cdot}{\cdot}_{(y_{u_i})_{i \in [k_u]}}$ of the vector space $\left\{f \in \mathds{R}^{k_u} \mid \inner{f}{1}_{(y_{u_i})_{i \in [k_u]}} = 0 \right\}$, where the inner product is defined by $\inner{a}{b}_{(y_{u_i})_{i \in [k_u]}} = \sum_{i = 1}^{k_u}a(i)b(i)y_{u_i}$, and then the vector
  \begin{align*}
    F[u]^t = (\underbrace{0, \cdots, 0}_{\text{$k_1$ times}}, \cdots, \underbrace{0, \cdots, 0}_{\text{$k_{u-1}$ times}}, f[u]^t_1, f[u]^t_2, \cdots, f[u]^t_{k_u}, \underbrace{0, \cdots, 0}_{\text{$k_{u+1}$ times}}, \cdots, \underbrace{0, \cdots, 0}_{\text{$k_n$ times}})
  \end{align*}
  is an eigenvector of $\widehat{\Psi}_{\vec{k}}$ with eigenvalue $1$.
  This is because:
 \begin{enumerate}
\item for any $i \in [k_u]$, it holds that 
  \begin{align*}
  (\widehat{\Psi}_{\vec{k}}F[u]^t)_{u_i} = \sum_{j \in [k_u]} \widehat{\Psi}_{\vec{k}}(u_i,u_j)F[u]^t_{u_j} = F[u]^t_{u_i} + \sum_{j \in [k_u]}(\*y * \mu_{\vec{k}})_{u_j}(\1)F[u]^t_{u_j} = F[u]^t_{u_i},
  \end{align*}
  where the last equation holds because 
  \begin{align*}
  \sum_{j \in [k_u]}(\*y * {\mu}_{\vec{k}})_{u_j}(\1)F[u]^t_{u_j} = \frac{\mu_u(\1)}{Z\cdot k_u} \sum_{j \in [k_u]} y_{u_j} F[u]^t_{u_j} = \frac{\mu_u(\1)}{Z\cdot k_u}\inner{f[u]^t}{1}_{(y_{u_i})_{i \in [k_u]}} = 0,
  \end{align*}
  where $Z$ is defined by 
 \begin{align}\label{eq-def-paritition-Z}
  	Z &\triangleq \sum_{Y \in \Omega(\mu_{\vec{k}})} \mu_{\vec{k}}(Y)\prod_{v_i \in Y^{-1}(\1)}{y_{v_i}}\\ 
	&= \sum_{X \in \Omega(\mu)}\mu(X) \prod_{v \in X^{-1}(\1)}\sum_{i \in [k_v]}\frac{y_{v_i}}{k_v}\notag\\ 
	&=\sum_{X \in \Omega(\mu)}\mu(X) \prod_{v \in X^{-1}(\1)}x_v; \notag
  \end{align}
\item 
for any $v \neq u$ and $i \in [k_v]$, it hold that 
  \begin{align*}
  	(\widehat{\Psi}_{\vec{k}}F[u]^t)_{v_i}  
	&= \sum_{j \in [k_u]}\widehat{\Psi}_{\vec{k}}(v_i,u_j)F[u]^t_{u_j}\\ 
	&= \sum_{j \in [k_u]}\abs{ (\*y * \mu_{\vec{k}})_{u_j}^{v_i \gets \1}(\1) - (\*y * \mu_{\vec{k}})_{u_j}(\1)  }f[u]_j^t\\
  	\text{(by \Cref{claim-muk})}\quad &= \sum_{j \in [k_u]}\abs{ \frac{y_{u_j} \tp{(\*x * \mu)^{v \gets \1}_u(\1) - (\*x * \mu)_u(\1)} }{\sum_{\ell \in [k_u]}y_{u_\ell}} }f[u]^t_j\\
  	 &= \abs{(\*x * \mu)^{v \gets \1}_u(\1) - (\*x * \mu)_u(\1)}\inner{f[u]^t}{1}_{(y_{u_j})_{j \in [k_u]}}\\
	 &= 0.
  \end{align*}
 \end{enumerate}
Finally, note that
  \begin{align*}
    \left\{F^1, \cdots F^n\right\} \cup \left\{F[u]^1, \cdots, F[u]^{k_u - 1}\right\}_{u = 1}^n
  \end{align*}
  forms an orthogonal eigenbasis of $\widehat{\Psi}_{\vec{k}}$ with respect to the inner product $\inner{\cdot}{\cdot}_{(\*y * \mu_{\vec{k}})}$, where the inner product is defined by $\inner{a}{b}_{(\*y * \mu_{\vec{k}})} = \sum_{v_i \in V_{\vec{k}}}a(v_i)b(v_i)(\*y * \mu_{\vec{k}})_{v_i}(\1)$. Formally, for any distinct $1\leq u,v \leq n$
  \begin{align*}
  	\inner{F^u}{F^v}_{(\*y * \mu_{\vec{k}})} 
	&= \sum_{w=1}^n f^u_w f^v_w \sum_{i \in [k_w]} (\*y * \mu_{\vec{k}})_{w_i}(\1)\\ 
	&= \frac{1}{Z}\sum_{w=1}^n f^u_w f^v_w \sum_{X \in \Omega(\mu):X_w = 1}\mu(X)\prod_{w' \in \sigma^{-1}(X)} \frac{y_{w'}}{k_{w'}}\\
  	&= \sum_{w=1}^n f^u_w f^v_w (\*x * \mu)_{w}(\1)  = \inner{f^u}{f^v}_{(\*x * \mu)}\\
	&= 0.  
  \end{align*}
  For any $u \in [n]$, any $v \in [n]$ and $i \in [k_v - 1]$,
  \begin{align*}
  	\inner{F^u}{F[v]^i}_{(\*y * \mu_{\vec{k}})} 
	&= f^u_v  \sum_{j \in [k_v]}f[v]^i_j  (\*y * \mu_{\vec{k}})_{v_j}(\1)\\ 
	&= \frac{f^u_v \mu_v(\1)}{Zk_v} \sum_{j \in [k_v]}f[v]^i_j y_{v_j}\\
  	&=\frac{f^u_v \mu_v(\1)}{Zk_v} \inner{f[v]^i}{1}_{(y_{v_j})_{j \in [k_v]}}\\ 
	&= 0.
  \end{align*}
  For any $u \in [n]$, any distinct $i, j \in [k_u -1]$,
  \begin{align*}
  	\inner{F[u]^i}{F[u]^j}_{(\*y * \mu_{\vec{k}})} 	
	&=  \sum_{\ell \in [k_u]}f[u]^i_\ell f[u]^j_\ell (\*y * \mu_{\vec{k}})_{u_\ell}(\1)\\ 
	&= \frac{\mu_u(\1)}{Zk_u} \sum_{\ell \in [k_u]}f[u]^i_\ell f[u]^j_\ell y_{u_\ell}\\
  	&=\frac{\mu_u(\1)}{Zk_u}  \inner{f[u]^i}{f[u]^j}_{(y_{u_\ell})_{\ell \in [k_u]}}\\ 
	&= 0.
  \end{align*}
  For any distinct $v,u \in [n]$, $i \in [k_v - 1]$ and $j \in [k_v-1]$, it is straightforward to verify
  \begin{align*}
  		\inner{F[u]^i}{F[v]^j}_{(\*y * \mu_{\vec{k}})} = 0.
  \end{align*}
  Hence, the spectrum of $\widehat{\Psi}_{\vec{k}}$ is
  \begin{align*}
    \left\{\lambda_1, \cdots, \lambda_n\right\} \cup \left\{ 1^{(t)} \right\}_{t=1}^{\sum_{u \in [n]} k_u - n},
  \end{align*}
  So, we know that $\rho\tp{\widehat{\Psi}_{\vec{k}}} \leq \max\left\{\rho\tp{\widehat{\Psi}}, 1\right\}$.
  \end{proof}
  
  \begin{proof}[Proofs of \Cref{claim-u1-ui} and \Cref{claim-muk}]
We first prove \Cref{claim-muk}, then use \Cref{claim-muk} to prove \Cref{claim-u1-ui}.
By the definition of conditional probability, we have
  \begin{align*}
  	(\* y * \mu_{\vec{k}})^{u_i \gets 1}_{v_j}(\1) &= \tp{\sum_{\substack{\sigma \in \Omega(\mu_{\vec{k}}):\\\sigma(u_i) = \1 \land \sigma(v_j) = \1 }}\mu_{\vec{k}}(\sigma)\prod_{w_\ell \in \sigma^{-1}(\1)}y_{w_j}} \Big/ \tp{\sum_{\substack{\sigma \in \Omega(\mu_{\vec{k}}):\\\sigma(u_i) = \1}}\mu_{\vec{k}}(\sigma)\prod_{w_\ell \in \sigma^{-1}(\1)}y_{w_j}}\\
   \end{align*}
  The numerator equals to 
  \begin{align*}
  &\sum_{\substack{\tau \in \Omega(\mu):\\\tau(u) = \1 \land \tau(v) = \1 }}\frac{\mu(\tau)y_{u_i}y_{v_j}}{k_uk_v} \prod_{w \in \tau^{-1}(\1) \setminus \{u,v\}} \sum_{\ell \in [k_w]}\frac{y_{w_\ell}}{k_w} \\
  = &\sum_{\substack{\tau \in \Omega(\mu):\\\tau(u) = \1 \land \tau(v) = \1 }}\frac{\mu(\tau)y_{u_i}y_{v_j}}{k_uk_v} \prod_{w \in \tau^{-1}(\1) \setminus \{u,v\}} x_w\\
  =\,& \frac{y_{v_j}}{\sum_{\ell \in [k_v]}y_{v_\ell}} \cdot \frac{y_{u_i}}{\sum_{\ell \in [k_u]}y_{u_\ell}} \cdot \sum_{\substack{\tau \in \Omega(\mu):\\\tau(u) = \1 \land \tau(v) = \1 }}\mu(\tau)\prod_{w \in \tau^{-1}(\1)} x_w.
  \end{align*} 
  The denominator equals to
  \begin{align*}
  	&\sum_{\substack{\tau \in \Omega(\mu):\\ \tau(u) = \1 }}\frac{\mu(\tau)y_{u_i}}{k_v} \prod_{w \in \tau^{-1}(\1) \setminus \{u\}} \sum_{\ell \in [k_w]}\frac{y_{w_\ell}}{k_w}  
= \frac{y_{u_i}}{\sum_{\ell \in [k_u]}y_{u_\ell}}\sum_{\substack{\tau \in \Omega(\mu):\\ \tau(u) = \1 }}\mu(\tau) \prod_{w \in \tau^{-1}(\1)} x_w
  \end{align*}
  Hence, we have
  \begin{align*}
  (\* y * \mu_{\vec{k}})^{u_i \gets 1}_{v_j}(\1) &= 	\frac{y_{v_j}}{\sum_{\ell \in [k_v]}y_{v_\ell}} \cdot  \tp{\sum_{\substack{\tau \in \Omega(\mu):\\\tau(u) = \1 \land \tau(v) = \1 }}\mu(\tau)\prod_{w \in \tau^{-1}(\1)} x_w} \Big / \tp{\sum_{\substack{\tau \in \Omega(\mu):\\ \tau(u) = \1 }}\mu(\tau) \prod_{w \in \tau^{-1}(\1)} x_w}\\
  &= 	\frac{y_{v_j}}{\sum_{\ell \in [k_v]}y_{v_\ell}} \cdot (\*x * \mu)^{u \gets \1}_v(\1).
  \end{align*}
  Recall $Z$ defined in~\eqref{eq-def-paritition-Z}. We have
  \begin{align*}
  	(\*y * \mu_{\vec{k}})_{u_i} (\1) 
	&= \frac{1}{Z}\sum_{\substack{\sigma \in \Omega(\mu_{\vec{k}}):\\\sigma(u_i) = \1}}\mu_{\vec{k}}(\sigma)\prod_{w_\ell \in \sigma^{-1}(\1)}y_{w_j} \\
	&= 	\frac{1}{Z}\sum_{\substack{\tau \in \Omega(\mu):\\ \tau(u) = \1 }}\frac{\mu(\tau)y_{u_i}}{k_u} \prod_{w \in \tau^{-1}(\1) \setminus \{u\}} \sum_{\ell \in [k_w]}\frac{y_{w_\ell}}{k_w}\\
  	& = \frac{y_{u_i}}{\sum_{\ell \in [k_u]}y_{u_\ell}} \cdot \frac{1}{Z} \cdot \sum_{\substack{\tau \in \Omega(\mu):\\ \tau(u) = \1 }}\mu(\tau) \prod_{w \in \tau^{-1}(\1)} x_w\\
  	&=  \frac{y_{u_i}}{\sum_{\ell \in [k_u]}y_{u_\ell}} \cdot \mu_u(\1). 
  \end{align*}
  Next, we prove \Cref{claim-u1-ui}. By definition, we have for any $ u,v \in [n] $ with $u \neq v$, any $i \in [k_u]$,  we have
  \begin{align*}
  \widehat{\Psi}(u, v) = \sum_{h \in [k_u]} \widehat{\Psi}_{\vec{k}}(u_1, v_h) 
  &= \sum_{h \in [k_v]}\abs{ (\*y * \mu_{\vec{k}})^{u_1 \gets \1}_{v_h}(\1) - (\*y * \mu_{\vec{k}})_{v_h}(\1)}\\
  \text{(by \Cref{claim-muk})}\quad &= \sum_{h \in [k_v]}\abs{ \frac{y_{v_h} \tp{ (\*x * \mu)^{u \gets \1}_v(\1) - (\*x * \mu)_{v}(\1)} }{\sum_{\ell \in [k_v]}y_{v_\ell}} }\\
   \text{(by \Cref{claim-muk})}\quad &=  \sum_{h \in [k_v]}\abs{ (\*y * \mu_{\vec{k}})^{u_i \gets \1}_{v_h}(\1) - (\*y * \mu_{\vec{k}})_{v_h}(\1)}\\
   &=  \sum_{h \in [k_v]} \widehat{\Psi}_{\vec{k}}(u_i, v_h).  
  \end{align*}
  For any $u \in [n]$ and $i \in [k_u]$, we have
  \begin{align*}
  	\widehat{\Psi}(u, u) = \sum_{h \in [k_u]} \widehat{\Psi}_{\vec{k}}(u_1, u_h)
  	&= 1 + (\*y * \mu_{\vec{k}})_{u_1}(\1) + \sum_{j\in [k_u] \setminus \{1\}}\abs{ (\*y * \mu_{\vec{k}})_{u_j}^{u_1 \gets \1}(\1) - (\*y * \mu_{\vec{k}})_{u_j}(\1) }\\ 
  	\tp{\text{by }(\*y * \mu_{\vec{k}})_{u_j}^{u_1 \gets \1}(\1) = 0}\quad&= 1 + \sum_{j\in [k_u]}(\*y * \mu_{\vec{k}})_{u_j}(\1)\\
  	&= 1 + (\*y * \mu_{\vec{k}})_{u_i}(\1) + \sum_{j\in [k_u] \setminus \{i\}}\abs{ (\*y * \mu_{\vec{k}})_{u_j}^{u_i \gets \1}(\1) - (\*y * \mu_{\vec{k}})_{u_j}(\1) }\\
  	&= \sum_{h \in [k_u]} \widehat{\Psi}_{\vec{k}}(u_i, u_h). \qedhere
  \end{align*}
  \end{proof}

\begin{proof}[Proof of \Cref{lem:muk-pin-ratio}]
  Recall $V = [n]$. First, define
  \begin{align}\label{eq-def-proof-R}
    R_{-} &\triangleq \left\{u \in V \mid \forall i \in [k], u_i \in \Lambda \land \sigma_{u_i} = \0\right\}, \notag\\
    R_{+} &\triangleq \left\{u \in V \mid \exists i \in [k], u_i \in \Lambda \land \sigma_{u_i} = \1\right\}, \notag\\
    R &\triangleq R_{-} \uplus R_{+}.
  \end{align}
  Let $\tau \in \{\0, \1\}^{R}$ be
  \begin{align}\label{eq-def-proof-tau}
    \forall u \in R,\quad \tau_u \triangleq
    \begin{cases}
      \0, & u \in R_{-} \\
      \1, & u \in R_{+}.
    \end{cases}
  \end{align}
  Now, note that
\begin{align*}
  \tp{\*z * \mu_k}^{\sigma}_{v_i}(\1)
  &= \frac{\Pr[Y \sim \*z * \mu_k]{Y_{v_i} = \1 \land Y_{\Lambda} = \sigma}}{\Pr[Y \sim \*z * \mu_k]{Y_{\Lambda} = \sigma}}
    \quad \text{and} \quad
    \tp{\*z * \mu_k}^{\sigma}_{v_i}(\0)
  = \frac{\Pr[Y \sim \*z * \mu_k]{Y_{v_i} = \0 \land Y_{\Lambda} = \sigma}}{\Pr[Y \sim \*z * \mu_k]{Y_{\Lambda} = \sigma}}.
\end{align*}

Note that $v_i \in V_k \setminus \Lambda$ and $\mu^\sigma_{v_i}(\1) > 0$.
We first show that $v \notin R$.
Suppose $v \in R_+$. Since $v_i \in V_k \setminus \Lambda$, there exists $v_j \in C_v$ such that $\sigma_{v_j} = \1$, and thus $\mu^\sigma_{v_i}(\1) = 0$, but $\mu^\sigma_{v_i}(\1) > 0$.
Suppose $v \in R_-$. It must hold that $v_j \in \Lambda$ for all $j \in [k]$, but $v_i \notin \Lambda$.
Hence, it holds that $v \notin R$.
   
Define the partition function 
\begin{align*}
Z \triangleq \sum_{Y \in \Omega(\mu_k)} \mu(Y^\star) \prod_{u_j \in V_k: Y_{u_j} = 1} \frac{z_{u_j}}{k},\quad \text{where }\forall u \in V,	
Y^\star(u) = \begin{cases}
\1 &\text{if } \exists j \in [k], Y(u_j) = \1\\
\0 &\text{if } \forall j \in [k], Y(u_j) =\0.
\end{cases}
\end{align*}
 For any $u \in [n]$, let $S_u \triangleq  C_u \setminus \Lambda$, where $C_u = \{u_i \mid i \in [k]\}$.
We have
\begin{align*}
  \Pr[Y \sim \*z * \mu_k]{Y_{v_i} = \1 \land Y_{\Lambda} = \sigma}
  &= \frac{1}{Z} \sum_{Y \in \Omega(\mu_k)} \mu(Y^\star) \cdot  \tp{\prod_{u_j \in V_k: Y_{u_j} = 1} \frac{z_{u_j}}{k}}\cdot \mathds{1}[Y_{v_i} = 1 \land Y_{\Lambda} = \sigma]\\
  &= \frac{1}{Z} \sum_{\substack{ Y \in \Omega(\mu_k):\\ Y_{v_i} = \1 \land Y_{\Lambda} = \sigma}} \mu(Y^\star) \cdot \prod_{u_j \in V_k \setminus \sigma^{-1}(\1): Y_{u_j} = \1} \frac{z_{u_j}}{k} \cdot \prod_{u_j \in \sigma^{-1}(\1)} \frac{z_{u_j}}{k}\\
(\ast)  &= \frac{1}{Z} \sum_{\substack{ X \in \Omega(\mu):\\ X_{v} = \1 \land X_R = \tau}} \mu(X) \cdot \frac{z_{v_i}}{k} \cdot \prod_{\substack{u \in V\setminus R : \\ u \neq v \land X_u = \1}} \tp{ \sum_{u_j \in S_u} \frac{z_{u_j}}{k}} \cdot \prod_{u_j \in \sigma^{-1}(\1)} \frac{z_{u_j}}{k}.
\end{align*}
In equation~$(\ast)$, we enumerate all $X= Y^\star$. Since $Y_{v_i} =\1$, it holds that $X_v = \1$. For any $u \in V\setminus R$ and $u \neq v$, if $X_u = Y^\star_u = \1$, we must select one $u_j \in C_u \setminus \Lambda = S_u$ to set  $Y_{u_j} = \1$, which gives the factor $\prod_{\substack{u \in V\setminus R : \\ u \neq v \land X_u = \1}} \tp{\frac{1}{k} \sum_{u_j \in S_u} z_{u_j}}$ in ($\ast$). 
Similarly, it holds that
\begin{align*}
  &\Pr[Y \sim \*z * \mu_k]{Y_{v_i} = \0 \land Y_{\Lambda} = \sigma} \\
  = &\frac{1}{Z} \sum_{Y \in \Omega(\mu_k)} \mu(Y^\star) \cdot \tp{\prod_{\substack{ u_j \in V_k \setminus \sigma^{-1}(1):\\ Y_{u_j} = 1}} \frac{z_{u_j}}{k} \cdot \prod_{u_j \in \sigma^{-1}(1)} \frac{z_{u_j}}{k}} \cdot \mathds{1}[Y_{v_i} = \0 \land Y_{\Lambda} = \sigma]\\
  = &\frac{1}{Z} \sum_{\substack{X \in \Omega(\mu)\\X_R=\tau}} \mu(X) \cdot \tp{\mathds{1}[X_v = \0] + \mathds{1}[X_v = \1] \tp{\frac{1}{k} \sum_{v_j \in S_v \setminus \{v_i\}} z_{v_j}}} \cdot \prod_{{\substack{u \in V \setminus R : \\ u \neq v \land X_u = \1}} } \tp{ \sum_{u_j \in S_u} \frac{z_{u_j}}{k}} \cdot \prod_{u_j \in \sigma^{-1}(\1)} \frac{z_{u_j}}{k}.
\end{align*}
Since $Y_{v_i} = \0$ and $X = Y^\star$. If $X_v = \0$, then $Y_{v_j} = 0$ for all $j \in [k]$; if $X_v = \1$, since $Y_{v_i} = \0$, there exists $v_j \in S_v \setminus \{v_i\}$ such that $Y_{v_j} = \1$. This gives the factor $\tp{\mathds{1}[X_v = \0] + \mathds{1}[X_v = \1] \tp{\frac{1}{k} \sum_{v_j \in S_v \setminus \{v_i\}} z_{v_j}}}$ in above formula. 
Hence, it holds that
\begin{align*}
  \frac{\tp{\*z * \mu_k}^{\sigma}_{v_i}(\0)}{\tp{\*z * \mu_k}^{\sigma}_{v_i}(\1)}
  &= \frac{\Pr[Y \sim \*z * \mu_k]{Y_{v_i} = \0 \land Y_{\Lambda} = \sigma}}{\Pr[Y \sim \*z * \mu_k]{Y_{v_i} = \1 \land Y_{\Lambda} = \sigma}}
    = \frac{k}{z_{v_i}} \tp{\frac{\tp{\*x * \mu}^\tau_v(\0)}{\tp{\*x * \mu}^\tau_v(\1)} + \frac{1}{k} \sum_{v_j \in S_v \setminus \{v_i\}} z_{v_j}},
\end{align*}
where
\begin{align}
  \forall u \in V,\quad x_u &\triangleq \begin{cases}
    \frac{1}{k} \sum_{u_j \in S_u} z_{u_j},& u \in V \setminus R \text{ and } u \neq v \\
    1,& u \in R \text{ or } u = v. 
    \end{cases}.
\end{align}
Note that $\*z \in (0,1+\epsilon]^{V_k}$ implies $\*x \in (0,1+\epsilon]^{V}$.
\end{proof}

\subsection{Marginal stability of $\mu_k$ (Proof of \Cref{lemma-part-II})}


Fix a subset $\Lambda \subseteq V_k$ and a feasible configuration $\sigma \in \Omega\tp{\mu_{k, V_k \setminus \Lambda}}$ on $V_k \setminus \Lambda$. 
Fix a variable $v_i \in  \Lambda$ and a feasible configuration $\tau \in \Omega{\tp{\mu^\sigma_{k, \Lambda \setminus \{v_i\}}}}$ on $ \Lambda \setminus \{v_i\}$.
Our goal is to verify the following inequalities: 

\begin{align}\label{eq-bounded-ratio-1}
  \frac{\mu_{k,v_i}^{\sigma\cup \tau}(\1)}{\mu_{k,v_i}^{\sigma \cup \tau}(\0)} \leq \zeta,
\end{align}
\begin{align} \label{eq-bounded-ratio-2}
  \frac{\mu_{k,v_i}^{\sigma\cup \tau}(\1)}{\mu_{k,v_i}^{\sigma \cup \tau}(\0)} \leq 2\zeta \cdot \frac{\mu^\sigma_{k,v_i}(\1)}{\mu^\sigma_{k,v_i}(\0)}.
\end{align}

We first show that~\eqref{eq-bounded-ratio-1} and~\eqref{eq-bounded-ratio-2} together indeed guarantee the marginal stability of $\mu_k$ that we want. By~\eqref{eq-bounded-ratio-1}, we know that for any $\gamma \in \Omega(\mu_{k,V_k \setminus \{v_i\} })$, it holds that $ {\mu_{k,v_i}^{\gamma}(\1)}/{\mu_{k,v_i}^{\gamma}(\0)} \leq \zeta$, which implies for any partial pinning $\rho \in \Omega(\mu_{k,S})$, where $S\subseteq V_k \setminus \{v_i\}$, ${\mu_{k,v_i}^{\rho}(\1)}/{\mu_{k,v_i}^{\rho}(\0)} \leq \zeta$. 
Next, consider $H \subseteq S$. It holds that 
\begin{align*}
	\frac{\mu_{k,v_i}^{\rho}(\1)}{\mu_{k,v_i}^{\rho}(\0)} \leq \max_{\phi \in \Omega\tp{\mu_{k,V_k \setminus (S \cup \{v_i\}) }^\rho} } \frac{\mu_{k,v_i}^{\rho \cup \phi}(\1)}{\mu_{k,v_i}^{\rho \cup \phi}(\0)} \leq 2\zeta \frac{\mu_{k,v_i}^{\rho_H}(\1)}{\mu_{k,v_i}^{\rho_H}(\0)},
\end{align*}
where in the last inequality we use~\eqref{eq-bounded-ratio-2} with $\sigma = \rho_H$ and $\sigma\cup \tau = \rho \cup \phi$.

Our proof is reduced to verifying ~\eqref{eq-bounded-ratio-1} and \eqref{eq-bounded-ratio-2}.
In the rest part of this section, without loss of generality, we may assume that $\mu^{\sigma\cup\tau}_{k, v_i}(\1) > 0$ and $\mu^\sigma_{k, v_i}(\1) > 0$, since when $\mu^{\sigma\cup \tau}_{k, v_i}(\1) = 0$, \eqref{eq-bounded-ratio-1}, \eqref{eq-bounded-ratio-2} hold trivially; and note that $\mu^\sigma_{k, v_i}(\1) = 0$ implies $\mu^{\sigma\cup \tau}_{k, v_i}(\1) = 0$.

We first proof \eqref{eq-bounded-ratio-1}.
Note that $\sigma \cup \tau$ is a configuration on $V_k \setminus \{v_i\}$.
We use \Cref{lem:muk-pin-ratio} with $\epsilon = 1$ and $\*z = \*1$.
By \Cref{lem:muk-pin-ratio}, there exist $\*x \in (0, 1]^V$ satisfying $x_v = 1$, a subset $R = V \setminus \{v\}$, and a configuration $\rho \in \Omega(\mu_R)$ such that 
 \begin{align}\label{eq-ratio-nu}
   \frac{\mu_{k,v_i}^{\sigma\cup \tau}(\0)}{\mu_{k,v_i}^{\sigma \cup \tau}(\1)}
    &= k\cdot \frac{\tp{\*x * \mu}^{\rho}_v(\0)}{\tp{\*x * \mu}^\rho_v(\1)}.
  \end{align}
  Specifically, by~\eqref{eq-def-proof-R} and \eqref{eq-def-proof-tau}, we have 
    \begin{align*}
  	R_-&=\left\{u \in V \mid \forall j \in [k], u_j \neq v_i \land (\sigma\cup\tau)_{u_j} = \0\right\},\\
  	R_+&=\left\{u \in V \mid \exists j \in [k], u_j \neq v_i \land (\sigma\cup\tau)_{u_j} = \1\right\},\\
  	R &= R_- \cup R_+ = V \setminus \{v\},
      \end{align*}
and $\rho$ defined by 
    \begin{align*}
      \forall u \in R,\quad \rho_u &\triangleq
    \begin{cases}
      \0, & u \in R_{-} \\
      \1, & u \in R_{+}
    \end{cases}.
      \end{align*}
%
 Since $x_v = 1$ and $v \notin R$ have already hold by \Cref{lem:muk-pin-ratio}, we have
  \begin{align}\label{eq-lower-bound>0}
  	\frac{\tp{\*x * \mu}^{\rho}_v(\0)}{\tp{\*x * \mu}^\rho_v(\1)} 
  	\overset{(\ast)}{\geq} \frac{1}{\zeta} > 0,
  \end{align}
  where inequality $(\ast)$ holds because $(\*x * \mu)$ is $\zeta$-marginally stable.
  This proves \eqref{eq-bounded-ratio-1}.
  
Now, we bound \eqref{eq-bounded-ratio-2}. Recall that we assume $\mu^\sigma_{k, v_i}(\1) > 0$.
By \Cref{lem:muk-pin-ratio}, there exist $\*x' \in (0, 1]^V$ satisfying $x'_v = 1$, a subset $R'\subseteq V$ with $v \notin R'$, and a configuration $\rho' \in \Omega(\mu_{R'})$ such that 
\begin{align}\label{eq-ratio-2}
  \frac{\mu^\sigma_{k,v_i}(\0)}{\mu^\sigma_{k,v_i}(\1)}
    &= k\tp{\frac{\tp{\*x' * \mu}^{\rho'}_v(\0)}{\tp{\*x' * \mu}^{\rho'}_v(\1)} + \frac{1}{k} \sum_{v_j \in (C_v \cap \Lambda) \setminus \{v_i\}} 1} \leq k \cdot \frac{\tp{\*x' * \mu}^{\rho'}_v(\0)}{\tp{\*x' * \mu}^{\rho'}_v(\1)} + k = \frac{k}{\tp{\*x' * \mu}^{\rho'}_v(\1)}.
\end{align}
  By~\eqref{eq-def-proof-R} and  \eqref{eq-def-proof-tau}, we have
  \begin{align*}
  	R'_-&=\left\{u \in V \mid \forall j \in [k], u_j \in V_k \setminus \Lambda \land \sigma_{u_j} = \0\right\}\\
  	R'_+&=\left\{u \in V \mid \exists j \in [k], u_j \in V_k \setminus \Lambda \land \sigma_{u_j} = \1\right\}\\
  	R' &= R'_- \cup R'_+\\
    \forall u \in R',\quad \rho'_u &\triangleq
    \begin{cases}
      \0, & u \in R_{-}' \\
      \1, & u \in R_{+}'.
    \end{cases}
  \end{align*}

  Before we progress, recall that we have assumed $\mu^{\sigma\cup\tau}_{k, v_i}(\1) > 0$ and $\mu^\sigma_{k, v_i}(\1) > 0$. Combining this fact with \eqref{eq-ratio-nu}, \eqref{eq-ratio-2}, it holds that $(\*x * \mu)^\rho_v(\1) > 0$ and $(\*x' * \mu)^{\rho'}_v(\1) > 0$.
  
  Now, in order to prove \eqref{eq-bounded-ratio-2}, we claim that
  \begin{align} \label{eq-bounded-ratio-2-aux}
    \frac{1}{\tp{\*x' * \mu}^{\rho'}_v(\1)} \leq 2\zeta \frac{\tp{\*x * \mu}^{\rho}_v(\0)}{\tp{\*x * \mu}^\rho_v(\1)}.
  \end{align}
  Combining \eqref{eq-ratio-nu}, \eqref{eq-ratio-2}, and \eqref{eq-bounded-ratio-2-aux}, it holds that
  \begin{align*}
  \frac{\mu^\sigma_{k,v_i}(\0)}{\mu^\sigma_{k,v_i}(\1)}
    \leq 2\zeta \frac{\mu_{k,v_i}^{\sigma\cup \tau}(\0)}{\mu_{k,v_i}^{\sigma \cup \tau}(\1)},
  \end{align*}
  and this proves \eqref{eq-bounded-ratio-2}.

  Now, we only left to prove \eqref{eq-bounded-ratio-2-aux}, which, by some calculation, is equivalent to
  \begin{align}\label{eq:verify-ms}
    \frac{(\*x * \mu)^\rho_v(\1)}{(\*x * \mu)^\rho_v(\0)} + \frac{(\*x * \mu)^\rho_v(\1)}{(\*x * \mu)^\rho_v(\0)} \bigg/ \frac{\tp{\*x' * \mu}^{\rho'}_v(\1)}{\tp{\*x' * \mu}^{\rho'}_v(\0)} \leq 2\zeta.
  \end{align}
  Both the first and the second term could be bounded
  by the complete marginal stability of $\zeta$.
  In particular, $(\*x * \mu)^\rho_v = (\*x' * \mu)^\rho_v$ holds by the the fact  $\rho \in \{\0,\1\}^{V \setminus \{ v\}}$ and $x_v = x'_v$.
  Therefore, the second term of~\eqref{eq:verify-ms} can be bounded by
  \begin{align*}
    \frac{(\*x * \mu)^\rho_v(\1)}{(\*x * \mu)^\rho_v(\0)} \bigg/ \frac{\tp{\*x' * \mu}^{\rho'}_v(\1)}{\tp{\*x' * \mu}^{\rho'}_v(\0)}
    =
    \frac{(\*x' * \mu)^\rho_v(\1)}{(\*x' * \mu)^\rho_v(\0)} \bigg/ \frac{\tp{\*x' * \mu}^{\rho'}_v(\1)}{\tp{\*x' * \mu}^{\rho'}_v(\0)} \leq \zeta,
  \end{align*}
  where the inequality holds by $\rho_{R'} = \rho'$ the $\zeta$-marginal stability of $\*x' * \mu$.

  \color{black}

\section{Applications to Anti-Ferromagnetic Two-Spin  Systems}\label{section-app}
%
%
In this section, we apply \Cref{theorem-general} to anti-ferromagnetic 2-spin systems and prove the lower bound on the modified log-Sobolev (MLS) constant for anti-ferromagnetic two-spin systems in \Cref{thm:2spin-theorem}.
Given the modified log-Sobolev bound, the mixing time bound in \Cref{thm:2spin-theorem} is standard, whose calculation is postponed to \Cref{sec:append}.

Let $\+I = (G=(V,E), \beta, \gamma, \lambda)$ be an anti-ferromagnetic two-spin system with Gibbs distribution $\mu$, where 
\begin{align}
0 \leq \beta \leq \gamma, \lambda,\gamma > 0 \quad\text{and}\quad \beta\gamma < 1. \label{eq:anti-ferro-regime}
\end{align}
Let $n=|V|$ and $\Delta\geq 3$ denote the maximum degree of $G$. 
Suppose that $\+I$ satisfies \Cref{condition-canonical-two-spin}, that is:
\begin{itemize}
    \item $(\beta, \gamma, \lambda)$ is $(\Delta-1)$-unique with gap $\delta \in (0,1)$; 
    \item $G$ is regular or $\gamma \leq 1$. 
  \end{itemize}

The following fact is folklore. A formal proof  is provided in \Cref{section-app-monotone}.
 \begin{proposition}\label{prop:LLY}
  Let $(\beta,\gamma,\lambda)$ satisfy \eqref{eq:anti-ferro-regime}.
 Let $\Delta \geq 3$ be an integer and $\delta \in (0,1)$.
 If $\gamma \leq 1$, then $(\beta,\gamma,\lambda)$ is up-to-$\Delta$ unique with gap $\delta$ if and only if $(\beta,\gamma,\lambda)$ is $(\Delta-1)$-unique with gap $\delta$.
  \end{proposition}  

With this, we can assume that $\+I$ satisfies the following condition that is equivalent to \Cref{condition-canonical-two-spin}.
\begin{condition}\label{condition-proof-two-spin}
Let $\delta \in (0,1)$.
The anti-ferromagnetic two-spin system $\+I = (G, \beta, \gamma, \lambda)$ with maximum degree $\Delta=\Delta_G \geq 3$ satisfies one of the following two conditions
\begin{itemize}
	\item $\gamma \leq 1$ and $(\lambda,\beta,\gamma)$ is up-to-$\Delta$ unique with gap $\delta$; 
	\item $\gamma > 1$, $(\lambda,\beta,\gamma)$ is $(\Delta-1)$-unique with gap $\delta$,   and $G$ is $\Delta$-regular.
\end{itemize}
\end{condition} 
We will show that the modified log-Sobolev  constant $\rho^{\-{GD}}(\mu)$ of Glauber dynamics on $\mu$ is at least $\frac{1}{C(\delta) n}$ for some $C(\delta)=\exp(O(1/\delta))$.

As a preprocessing of $\mu$, we apply the flipping operation used in~\cite{chen2021rapid}.
\begin{definition}[\text{flipping operation}]
  Let $\mu$ be a distribution over $\{\0, \1\}^n$, and $\*\chi \in \{\0, \1\}^n$ be a direction vector. The flipped distribution $\pi = \-{flip}(\mu, \*\chi)$ over $\{\0, \1\}^n$ is defined as
\begin{align*}
  \forall \sigma \in \{\0, \1\}^n, \quad \pi(\sigma) \triangleq \mu(\*\chi \odot \sigma),
\end{align*}
where $(\*\chi \odot \sigma)_i \triangleq \chi_i\sigma_i$ for all $i \in [n]$.

In particular, if  $\*\chi_i = \chi \in \{\0, \1\} $ for all $i \in [n]$, we denote $\-{flip}(\mu, \*\chi)$ by $\-{flip}(\mu, \chi)$.
\end{definition}
%
Let $\chi=\chi(\+I) \in \{\0, \1\}$ be a direction indicator defined by
\begin{align} \label{eq:uniform-flip}
  \chi &\triangleq
    \begin{cases}
    \1, \quad \lambda \leq \tp{\frac{\gamma}{\beta}}^{\Delta/2}, \\
    \0, \quad \text{otherwise.}
    \end{cases}
\end{align}
Let $\pi = \-{flip}(\mu,\chi)$. By definition, $\pi$ is the Gibbs distribution of $\+I_{
\-{flip}}=(G, \bar{\beta},\bar{\gamma},\bar{\lambda})$, where 
\begin{align}\label{eq-def-bar}
(\bar{\beta},  \bar{\gamma}, \bar{\lambda})=
	\begin{cases}
	(\beta,  \gamma,  \lambda) &\text{if } \lambda \leq \tp{\frac{\gamma}{\beta}}^{\Delta / 2},\\
	(\gamma, \beta,  \frac{1}{\lambda}) &\text{if } \lambda > \tp{\frac{\gamma}{\beta}}^{\Delta / 2}	.
	\end{cases}
\end{align}
Note that either $\+I_{\-{flip}} = \+I$ or $\+I_{\-{flip}}$ is obtained by flipping the roles between $-1$ and $+1$ in $\+I$.
The following two observation about $\+I_{\-{flip}}$ are straightforward to verify.
\begin{observation}\label{observation-I-flip}
$\bar{\beta} \geq 0, \bar{\gamma} > 0$, $\bar{\beta}\bar{\gamma} < 1$ and $0 < \bar{\lambda} \leq \tp{\frac{\bar{\gamma}}{\bar{\beta}}}^{\Delta / 2}$.	
\end{observation}


\begin{observation} \label{lemma-regular-canonical}
$\rho^{\-{GD}}(\mu) = \rho^{\-{GD}}(\pi)$.
\end{observation}

The next lemma analyzes the modified log-Sobolev constant $\rho^{\-{GD}}(\pi)$ for flipped distribution $\pi = \-{flip}(\mu,\chi)$.
\begin{lemma}\label{theorem-regular}
Let $0 < \delta < 1$. If $\+I$ with Gibbs distribution $\mu$ satisfies \Cref{condition-proof-two-spin} with parameter $\delta$, then 
\begin{align*}
\rho^{\-{GD}}(\pi) \geq \frac{1}{C(\delta) n},	
\end{align*}
where $\rho^{\-{GD}}(\pi)$ is the  modified log-Sobolev constant for the Glauber dynamics on $\pi = \-{flip}(\mu,\chi)$ with $\chi$ defined in \eqref{eq:uniform-flip}, and $C(\delta) = \exp(O(1/\delta))$ is a constant depending only on $\delta$. 
\end{lemma}
The MLS bound in \Cref{thm:2spin-theorem} is a direct consequence of \Cref{prop:LLY},  \Cref{lemma-regular-canonical} and \Cref{theorem-regular}.
\Cref{theorem-regular} can be proved by \Cref{theorem-general} together with the following three lemmas.
\begin{lemma}[complete spectral independence] \label{lem:verify-C-SI}
  $\pi$ is $(\frac{288}{\delta}, \frac{\delta}{2})$-completely spectrally independent.
\end{lemma}
\begin{lemma}[complete marginal stability] \label{lem:verify-cbmr}
  $\pi$ is completely $\exp(12^5)$-marginally stable.
\end{lemma}
\begin{lemma}[MLSI in subcritical regime] \label{lem:good-regime-mLSI}
  For any $0 < \theta \leq 12^{-6}$, it holds that $\rho^{\mathrm{GD}}_{\min}(\theta * \pi) \geq \frac{1}{4n}$.
\end{lemma}

\begin{proof}[Proof of \Cref{theorem-regular}]
	By Lemmas \ref{lem:verify-C-SI}-\ref{lem:good-regime-mLSI}, \Cref{theorem-general} and setting $\theta = 12^{-6}$, we have
  \begin{align*}
    \rho(\pi) \ge 10^{-7\tp{\frac{9000}{\delta} + \frac{10^{10}}{\log(1+\frac{\delta}{2})}}} \frac{1}{4n} \overset{(\star)}{\ge} \frac{10^{-10^{12}/\delta}}{n} = \frac{1}{\exp(O(1/\delta)) n}
  \end{align*}
  where $(\star)$ is due to that $\log (1+x) \ge \frac{x}{2}$ for all $x \in [0,1]$. 
\end{proof}

\subsection{Verifying complete spectral independence}
In this section, we prove \Cref{lem:verify-C-SI}.
Let $\+I = (G,\beta,\gamma,\lambda)$  be an anti-ferromagnetic two-spin system instance with Gibbs distribution $\mu$ satisfying  \Cref{condition-proof-two-spin} with parameter $\delta \in (0,1)$.
%
Let $\pi = \-{flip}(\mu, \chi)$ be the flipped distribution, where $\chi$ is defined in~\eqref{eq:uniform-flip}.
We have the following lemma. 
The proof is given in \Cref{sec:proof-unique-relax}.
\begin{lemma} \label{lem:unique-relax}
Let $\delta \in (0, 1)$ and $\+I = (G, \beta, \gamma, \lambda)$ be an instance of anti-ferromagnetic two-spin systems, then
\begin{itemize}
\item For all $1 \leq d < \Delta$, $(\beta, \gamma, \lambda)$ is $d$-unique with gap $\delta$ implies $(\beta, \gamma, (1 + \frac{\delta}{2})^\chi \lambda)$ is $d$-unique with gap $\frac{\delta}{2}$.
	\item If $\Delta$ further satisfies $\Delta - 1 > \tp{1-\frac{\delta}{2}} \overline{\Delta}$ where $\overline{\Delta} = \frac{1+\sqrt{\beta \gamma}}{1-\sqrt{\beta \gamma}}$, then it holds that  
	\begin{align*}
		\lambda \le \tp{\frac{\gamma}{\beta}}^{\Delta/2} &\implies  \tp{1+\frac{\delta}{2}}^\chi \lambda < \tp{\frac{\gamma}{\beta}}^{\Delta/2}\\
		\lambda > \tp{\frac{\gamma}{\beta}}^{\Delta/2} &\implies  \tp{1+\frac{\delta}{2}}^\chi \lambda > \tp{\frac{\gamma}{\beta}}^{\Delta/2}.
	\end{align*}
\end{itemize}
\end{lemma}

Next, we need to use the following definition introduced in~\cite{chen2021rapid}.
\begin{definition} [complete spectral independence in a direction]
  Let $\eta,\epsilon \geq 0$ and $\*\chi \in \{\0, \1\}^n$.
  A distribution $\mu$ over $\{\0, \1\}^n$ is said to be $(\eta,\epsilon)$-completely spectrally independent in direction $\*\chi$ if $\*\theta^{\*\chi} * \mu$ is $(\eta,\epsilon)$-spectrally independent for all $\*\theta \in (0, 1+\epsilon]^V$, where $(\*\theta^{\*\chi})_v = \theta_v^{\chi_v}$ for all $v \in V$.
  
 In particular, if $\*\chi$ is a constant vector such that $\chi_v = \chi$ for all $v \in [n]$, we say $\mu$ is $(\eta,\epsilon)$-completely spectrally independent in direction $\chi$ for simplicity.
\end{definition}


We need the following lemma, whose proof is given in \Cref{sec-C-SI-mu}.
\begin{lemma} \label{lem:C-SI-mu}
For any anti-ferromagnetic two spin system instance $\+I=(G,\beta,\gamma,\lambda)$ satisfying \Cref{condition-proof-two-spin} with parameter $\delta \in (0,1)$, let $\mu$ denote the Gibbs distribution of $\+I$, 
 $\mu$ is $(\frac{144}{\delta},0)$-completely spectrally independent in direction $\chi(\+I)$ defined in~\eqref{eq:uniform-flip}, formally,
 \begin{align*}
 	\chi(\+I) = 
    \begin{cases}
    \1, \quad \lambda \leq \tp{\frac{\gamma}{\beta}}^{\Delta/2} \\
    \0, \quad \text{otherwise.}
    \end{cases}
 \end{align*}
 Furthermore, if $\Delta$ satisfies $\Delta - 1 \leq \tp{1-\delta} \overline{\Delta}$ where $\overline{\Delta} = \frac{1+\sqrt{\beta \gamma}}{1-\sqrt{\beta \gamma}}$, above result holds for any $\chi(\+I) \in \{\0,\1\}$.
\end{lemma}

We are now ready to prove \Cref{lem:verify-C-SI}.

\begin{proof}[Proof of \Cref{lem:verify-C-SI}]
Let $\nu \triangleq (1 + \frac{\delta}{2})^\chi * \mu$, which is the Gibbs distribution of the anti-ferromagnetic  two-spin  system $\+J = (G, \beta, \gamma, (1 + \frac{\delta}{2})^\chi \lambda)$.
%
We prove that $\nu$ is $(\frac{288}{\delta},0)$-completely spectrally independent in direction $\chi=\chi(\+I)$ defined in~\eqref{eq:uniform-flip}.
By \Cref{lem:unique-relax} and the fact that $\+J$ shares the same parameters $\beta,\gamma$ and graph $G$ with $\+I$, we know that $\+J$ satisfies \Cref{condition-proof-two-spin} with parameter $\frac{\delta}{2}$.
Recall that $\overline{\Delta} = \frac{1+\sqrt{\beta \gamma}}{1-\sqrt{\beta \gamma}}$.
We consider the following two cases.

Case $\Delta - 1 \le (1-\frac{\delta}{2}) \overline{\Delta}$. 
By the further more part \Cref{lem:C-SI-mu} (remark that we use \Cref{lem:C-SI-mu}  with parameter $\delta/2$), $\nu$ is $(\frac{288}{\delta},0)$-completely spectrally independent in direction $\chi$.

Case $\Delta - 1 > (1-\frac{\delta}{2}) \overline{\Delta}$. 
By \Cref{lem:C-SI-mu}, the Gibbs distribution $\nu$ is $(\frac{288}{\delta},0)$-completely spectrally independent in direction $\chi(\+J)$.
The second part of \Cref{lem:unique-relax} 
shows that (1) if $\chi = \1$, then $\chi(\+J) = \1$ (2) if $\chi = \0$, then $\chi(\+J) = \0$, which implies $\chi = \chi(\+J)$. Hence, $\nu$ is $(\frac{288}{\delta},0)$-completely spectrally independent in direction $\chi$.



Lastly, we verify that $\pi$ is $(\frac{288}{\delta},\frac{\delta}{2})$-completely spectrally independent.
Recall that $\Psi^{\-{AbsInf}}_{\cdot}$ is the absolute influence matrix defined in~\Cref{definition-SI}.
Let $\Lambda \subseteq V$ and $\sigma \in \Omega(\pi_{V\setminus \Lambda})$, it is straightforward to check that for any  $\*\phi \in (0, 1 + \frac{\delta}{2}]^V$,
\begin{align*}
  \Psi^{\-{AbsInf}}_{(\*\phi * \pi)^\sigma_\Lambda} = \Psi^{\-{AbsInf}}_{(\*\phi^{\chi} * \mu)^{\chi \odot \sigma}_{\Lambda}},
\end{align*}
where $(\chi \odot \sigma)_v = \chi \cdot \sigma_v$ for $v \in \Lambda$ and $(\*\phi^\chi)_v = \phi_v^{\chi_v}$ for $v \in V$.
Let $\*\theta \in (0, 1]^V$ such that $\theta_v = \phi_v / (1 + \frac{\delta}{2})$ for all $v \in V$, it holds that $\*\phi^\chi * \mu = \*\theta^\chi * ((1 + \frac{\delta}{2})^\chi * \mu) = \*\theta^\chi * \nu$, and 
\begin{align*}
  \Psi^{\-{AbsInf}}_{(\*\phi * \pi)^\sigma_\Lambda} = \Psi^{\-{AbsInf}}_{(\*\theta^\chi * \nu)^{\chi \odot \sigma}_\Lambda}.
\end{align*}
Since $\nu$ is $(\frac{288}{\delta},0)$-completely spectrally independent in direction $\chi$, $\pi$ is $(\frac{288}{\delta}, \frac{\delta}{2})$-completely spectrally independent.
\end{proof}
\subsubsection{Gap manipulation}\label{sec:proof-unique-relax}
In this section, we prove \Cref{lem:unique-relax}.
We need the following result.
\begin{lemma}[\text{\cite[Proposition 8.6]{chen2021rapid}}] \label{lem:lambda-unique}
Let $\beta,\gamma,\lambda$ be real numbers satisfying $0 \leq \beta \leq \gamma$, $\gamma > 0,\lambda > 0$ and $\beta\gamma < 1$. 

If $\beta = 0$, then the following holds for all integer $d \geq 1$:
  \begin{itemize}
  \item $(0, \gamma, \lambda)$ is $d$-unique with gap $\delta$ iff $\lambda \leq \lambda_{c, \delta}(d) = \frac{(1 - \delta)d^d \gamma^{d+1}}{(d - 1 + \delta)^{d+1}}$.
  \end{itemize}
  
  Assume $\beta > 0$. Let $\overline{\Delta} \triangleq \frac{1 + \sqrt{\beta\gamma}}{1 - \sqrt{\beta\gamma}}$. The following hold for all integers $d \geq 1$:
  \begin{itemize}
  \item If $d \leq (1 - \delta)\overline{\Delta}$, then $(\beta, \gamma, \lambda)$ is $d$-unique with gap $\delta$ for all $\lambda > 0$.
  \item If $d > (1 - \delta)\overline{\Delta}$, let $\zeta_\delta(d) \triangleq d(1 - \beta\gamma) - (1 - \delta)(1 + \beta\gamma)$,
    \begin{align*}
      x_{1,\delta}(d) = \frac{\zeta_\delta(d) - \sqrt{\zeta_\delta(d)^2 - 4(1 - \delta)^2\beta\gamma}}{2(1 - \delta)\beta}
      \quad \text{and} \quad
      x_{2,\delta}(d) = \frac{\zeta_\delta(d) + \sqrt{\zeta_\delta(d)^2 - 4(1 - \delta)^2\beta\gamma}}{2(1 - \delta)\beta},
    \end{align*}
    and for $i \in \{1, 2\}$, let
    \begin{align*}
      \lambda_{i, \delta}(d) = x_{i,\delta}(d) \tp{\frac{x_{i,\delta}(d) + \gamma}{\beta x_{i,\delta}(d) + 1}}^d.
    \end{align*}
    It holds that $\lambda_{1,\delta}(d)\lambda_{2,\delta}(d) = \tp{\frac{\gamma}{\beta}}^{d+1}$ and $\lambda_{1, \delta}(d) < \tp{\frac{\gamma}{\beta}}^{(d+1)/2} < \lambda_{2,\delta}(d)$.
    And $(\beta, \gamma, \lambda)$ is $d$-unique if and only if $\lambda \in (0, \lambda_{1,\delta}(d)] \cup [\lambda_{2, \delta}(d), \infty)$.
  \end{itemize}
\end{lemma}


To prove \Cref{lem:unique-relax}, for all $1 \leq d < \Delta$, 
we will show that if $(\beta,\gamma,\lambda)$ is $d$-unique with gap $\delta$, then $(\beta, \gamma, (1 + \frac{\delta}{2})^\chi \lambda)$ is $d$-unique with gap $\delta/2$.
We consider $3$ cases:
(1) $\beta = 0$;
(2) $\beta > 0$ and $1 \leq d \leq (1 - \delta)\overline{\Delta}$;
(3) $\beta > 0$ and $(1 - \delta)\overline{\Delta} < d < \Delta$.

\paragraph{\textbf{Case (1): $\beta = 0$}}
Since $\beta = 0$, it holds that $\chi = \1$.
By \Cref{lem:lambda-unique}, our goal is to show that
\begin{align*}
  (1 + \frac{\delta}{2}) \frac{(1 - \delta) d^d \gamma^{d+1}}{(d - 1 + \delta)^{d+1}} = (1 + \frac{\delta}{2}) \lambda_{c,\delta}(d) \leq \lambda_{c, \delta/2}(d) = \frac{(1 - \frac{\delta}{2}) d^d \gamma^{d+1}}{(d - 1 + \frac{\delta}{2})^{d+1}},
\end{align*}
which holds because $(1 - \delta)(1 + \delta/2) \leq (1 - \delta/2)$.

\paragraph{\textbf{Case (2): $\beta > 0$ and $1 \leq d \leq (1 - \delta)\overline{\Delta}$}}
In this case, \Cref{lem:lambda-unique} tells us that $(\beta, \gamma, (1 + \frac{\delta}{2})^{\chi}\lambda)$ is $d$-unique with gap $\delta$, and hence it is $d$-unique with gap $\frac{\delta}{2}$.

\paragraph{\textbf{Case (3): $\beta > 0$ and $(1 - \delta)\overline{\Delta} < d < \Delta$}}
Without loss of generality, we assume that $(\Delta - 1) > (1 - \delta)\overline{\Delta}$.

Fix an integer $d$ such that $(1 - \delta)\overline{\Delta} < d < \Delta$.
We consider $4$ sub-cases:
\textbf{(i)} $\lambda \leq \lambda_{1, \delta}(d)$ and $\chi = \0$;
\textbf{(ii)} $\lambda \geq \lambda_{2, \delta}(d)$ and $\chi = \1$.
\textbf{(iii)} $\lambda \leq \lambda_{1, \delta}(d)$ and $\chi = \1$;
\textbf{(iv)} $\lambda \geq \lambda_{2, \delta}(d)$ and $\chi = \0$.

Note that without loss of generality, we always assume that $\lambda_{1, \delta/2}(d)$ and $\lambda_{2,\delta/2}(d)$ are well defined.
Otherwise, $d < (1 - \frac{\delta}{2})\overline{\Delta}$, and by \Cref{lem:lambda-unique}, it holds that $(\beta, \gamma, (1 + \frac{\delta}{2}\lambda))$ is $d$-unique with gap $\frac{\delta}{2}$.

For case \textbf{(i)}, it holds that $(1 + \frac{\delta}{2})^{-1}\lambda \leq \lambda \leq \lambda_{1, \delta}(d)$, and the proof is done by levering \Cref{lem:lambda-unique}.

The case \textbf{(ii)} could be proved in the same manner as the case (i).

To prove case \textbf{(iii)}, by \Cref{lem:lambda-unique}, it suffices for us to show that $(1 + \frac{\delta}{2})\lambda_{1,\delta}(d) \leq \lambda_{1,\delta/2}(d)$, which is already done by the previous work~\cite[Proof of Proposition 66]{anari2021entropicII}.
We remark that their proof works for all $\beta,\gamma > 0$ satisfying $\beta\gamma < 1$.

We left to prove case \textbf{(iv)}.
Note that if we fix the parameter $d, \delta$, then $\lambda_{1, \delta}(d)$ and $\lambda_{2, \delta}(d)$ are actually functions of $\beta, \gamma$.
For convenience, we denote them as $\lambda_{1, \delta}(d; \beta, \gamma)$ and $\lambda_{2, \delta}(d; \beta, \gamma)$, respectively.
Let $\beta' = \gamma, \gamma' = \beta$, it holds that
\begin{align*}
	\lambda_{2,\delta}(d;\beta,\gamma) = 1/\lambda_{1,\delta}(d;\beta',\gamma').
\end{align*}
It suffices to show that $(1+\delta/2)^{-1}\lambda_{2,\delta}(d;\beta,\gamma) \geq \lambda_{2,\delta/2}(d;\beta,\gamma) $, which is equivalent to $(1+\delta/2)\lambda_{1,\delta}(d;\beta',\gamma') \leq \lambda_{1,\delta/2}(d;\beta',\gamma')$,
which is proved in case \textbf{(iii)}.

Finally, we prove the second part in \Cref{lem:unique-relax}.
Let $d = \Delta - 1$, by our assumption in \Cref{lem:unique-relax}, it holds that $d > (1-\delta/2)\overline{\Delta} > (1 - \delta)\overline{\Delta}$.
By \Cref{lem:lambda-unique}, we have the following two results (1) $\lambda_{1,\delta}(d)$ and $\lambda_{2,\delta}(d)$ exist; (2) $\lambda_{1,\delta/2}(d)$ and $\lambda_{2,\delta/2}(d)$ exist. 

\begin{itemize}
	\item If  $\lambda \leq (\frac{\gamma}{\beta})^{\Delta/2}$, then $\chi = \1$. By \textbf{case (3.iii)}, $(1 + \frac{\delta}{2})\lambda \leq (1 + \frac{\delta}{2})\lambda_{1,\delta}(d) \leq \lambda_{1,\delta/2}(d) < (\frac{\gamma}{\beta})^{\Delta/2}$.
\item If $\lambda > (\frac{\gamma}{\beta})^{\Delta/2}$, then $\chi = \0$. Let $\beta' = \gamma, \gamma' = \beta, \lambda' = 1/\lambda$, then by \textbf{case (3.iv)}, it holds that $(1 + \frac{\delta}{2}) \lambda' \leq (1 + \frac{\delta}{2}) \lambda_{1,\delta}(d;\beta',\gamma')  \leq \lambda_{1,\delta/2}(d; \beta', \gamma') < (\frac{\gamma'}{\beta'})^{\Delta/2}$, which implies $(1 + \frac{\delta}{2})^{-1} \lambda > (\frac{\gamma}{\beta})^{\Delta/2}$.
\end{itemize}

\subsubsection{Complete spectral independence of $\mu$ in direction $\chi$}\label{sec-C-SI-mu}
We prove \Cref{lem:C-SI-mu}.
Fix an anti-ferromagnetic two-spin system instance  $\+I = (G,\beta,\gamma,\lambda)$ satisfying \Cref{condition-proof-two-spin} with parameter $\delta \in (0,1)$.
Let $\mu$ denote the Gibbs distribution of $\+I$.
We prove that $\mu$ is $(\frac{144}{\delta},0)$-completely spectrally independent in direction $\chi$  defined in \eqref{eq:uniform-flip}.
Fix an arbitrary $\*\theta \in (0, 1]^V$. We show that $\nu \triangleq \*\theta^\chi * \mu$ is $(\frac{72}{\delta},0)$-spectrally independent, which implies the lemma.
Note that $\nu$ is the Gibbs distribution of the  two-spin  system defined by the the tuple $(G = (V, E), \beta, \gamma,\*\lambda)$, where $\*\lambda = (\lambda_v)_{v \in V} \in \mathds{R}_{>0}^V$ satisfies $\lambda_v = \theta_v^\chi \lambda$.

First, we introduce some notations and results.
For $\lambda > 0$, integer $d \geq 0$, consider tree recursion for \emph{log-marginal-ratios} $H_{\lambda, d}: [-\infty, +\infty]^d \to [-\infty, +\infty]$, 
\begin{align*}
  H_{\lambda, d}(y_1, \cdots, y_d) \triangleq \log \lambda + \sum_{i=1}^d \log \tp{\frac{\beta e^{y_i} + 1}{e^{y_i} + \gamma}}
\end{align*}
For $y \in [-\infty, +\infty]$, let
\begin{align*}
  h(y) &\triangleq - \frac{(1 - \beta\gamma)e^y}{(\beta e^y + 1)(e^y + \gamma)}.
\end{align*}
For real number $\lambda > 0$, integer $d > 0$, we define the intervals $J_{\lambda, d}$ as follow
\begin{align*}
  J_{\lambda, d} = \begin{cases}
    \left[-\infty, \log\tp{\frac{\lambda}{\gamma^d}}\right] & \text{if } \beta = 0; \\
    \left[\log\tp{\lambda\beta^d}, \log\tp{\frac{\lambda}{\gamma^d}}\right] & \text{if } 0 < \beta\gamma \leq 1.
    \end{cases}
\end{align*}
Specially, when $\lambda > 0$ and $d = 0$, let $J_{\lambda, 0} = \{\log \lambda\}$.

We use the following known results about two-spin systems.
\begin{lemma} [\text{\cite[Theorem 8.8]{chen2021rapid}, \cite{chen2020rapid}}] \label{lem: two-spin -SI}
Let $\nu$ be the Gibbs distribution of a two-spin system defined by graph $G = (V, E)$, and parameters $\beta, \gamma \in \mathds{R},  \*\lambda \in \mathds{R}^V$ such that $0 \leq \beta \leq \gamma$, $\gamma > 0$, $\beta\gamma < 1$, and $\lambda_v > 0$ for all $v \in V$.
For every $v \in V$, let $d_v \triangleq \Delta_v - 1$ where $\Delta_v$ is the degree of $v$ in $G$.
If there exists $\alpha, c > 0$ such that 
\begin{enumerate}
\item for every $v \in V$ with $d_v \geq 1$ and every $(y_1, \cdots, y_{d_v}) \in [-\infty, +\infty]^{d_v}$, it holds that
  \begin{align*}
    \sum_{i=1}^{d_v} \sqrt{\abs{h(y)}\abs{h(y_i)}} \leq 1 - \alpha,
  \end{align*}
  where $y = H_{\lambda_v, d_v}(y_1, \cdots, y_{d_v})$;
\item for every $v \in V$, every $y_v \in J_{\lambda_v, d_v}$, it holds that
  \begin{align*}
    \abs{h(y_v)} \leq \frac{c}{\Delta},
  \end{align*}
\end{enumerate}
then $\nu$ is $(\frac{2c}{\alpha},0)$-spectrally independent.
\end{lemma}

\begin{lemma} [\text{\cite[Theorem 8.11]{chen2021rapid}, \cite{LLY13}}] \label{lem:contraction}
  Let $d \geq 1$ be an integer, and let $\beta, \gamma, \lambda$ be real numbers satisfying that $0 \leq \beta \leq \gamma$, $\gamma > 0$, $\lambda > 0$, and $\beta\gamma < 1$.
  For any $\delta \in (0, 1)$, if $(\beta, \gamma, \lambda)$ is $d$-unique with gap $\delta$, then for every $(y_1, \cdots, y_d) \in [-\infty, +\infty]^d$ and $y = H_{\lambda, d}(y_1, y_2, \cdots, y_d)$, it holds that
  \begin{align*}
    \sum_{i=1}^d \sqrt{\abs{h(y)}\abs{h(y_i)}} < 1 - \frac{\delta}{2}.
  \end{align*}
\end{lemma}

\begin{lemma} [\text{\cite[Lemma 36]{chen2020rapid}}]\label{lem:boundedness}
  Let $\Delta \geq 3$ be an integer, and let $\beta, \gamma, \lambda$ be real numbers satisfying that $0 \leq \beta \leq \gamma$, $\gamma > 0$, $\lambda > 0$, and $\beta\gamma < 1$.
  Suppose $(\beta, \gamma, \lambda)$ is $(\Delta - 1)$-unique. It holds that
  \begin{itemize}
  \item if $\gamma \leq 1$, then for $0 \leq d < \Delta$, and every $y \in J_{\lambda, d}$, it holds that $\abs{h(y)} \leq \frac{18}{\Delta}$;
  \item if $G$ is $\Delta$-regular, then for $d = \Delta - 1$ and every $y \in J_{\lambda, d}$, it holds that $\abs{h(y)} \leq \frac{18}{\Delta}$.
  \end{itemize}
\end{lemma}

\begin{remark}
  The exact statement of \Cref{lem:boundedness} is slightly different from \cite[Lemma 36]{chen2020rapid}, but it can be verified by going through the same proof for \cite[Lemma 36]{chen2020rapid}.
  For completeness, a proof of \Cref{lem:boundedness} is provided in \Cref{sec:append-D}.
\end{remark}

By \Cref{lem: two-spin -SI}, \Cref{lem:contraction}, and \Cref{lem:boundedness}, to prove  that $\nu$ is $(\frac{72}{\delta},0)$-spectrally independent, we only need to prove one of the following two results
\begin{itemize}
	\item $(\beta,\gamma,\theta^\chi \lambda)$ is up-to-$\Delta$ unique;
	\item $G$ is $\Delta$-regular and $(\beta,\gamma,\theta^\chi \lambda)$ is $(\Delta-1)$-unique.
\end{itemize}
Note that the spin system $\+I = (G,\beta,\gamma,\lambda)$ in \Cref{lem:C-SI-mu} satisfies \Cref{condition-proof-two-spin} with parameter $\delta$. The above two results can be proved by the following lemma.

\begin{lemma}
Let $0 < \delta < 1$.
Let $G = (V,E)$ be a graph with maximum degree $\Delta \geq 3$.
Let $\beta, \gamma, \lambda$ be real numbers satisfying that $0 \leq \beta \leq \gamma$, $ \gamma> 0$, $\lambda > 0$, and $\beta\gamma < 1$. 
Let $\chi$ be the parameter defined in~\eqref{eq:uniform-flip}, $\delta \in (0, 1)$, and $\theta \in (0, 1]$, it holds that
\begin{itemize}
\item if $\gamma \leq 1$, then $(\beta,\gamma,\lambda)$ is up-to-$\Delta$ unique with gap $\delta$ implies $(\beta,\gamma,\theta^\chi \lambda)$ is up-to-$\Delta$ unique with gap $\delta$;
\item $(\beta,\gamma,\lambda)$ is $(\Delta - 1)$-unique with gap $\delta$ implies $(\beta,\gamma,\theta^\chi \lambda)$ is $(\Delta - 1$)-unique with gap $\delta$.
\end{itemize}
\end{lemma}


\begin{proof}
We prove the first part of the lemma. Assume $\gamma \leq 1$.
By definition, we need to prove that for every $1 \leq d < \Delta$, $(\beta, \gamma, \theta^\chi \lambda)$ is $d$-unique with gap $\delta$.
We consider $3$ cases:
(1) $\beta = 0$;
(2) $\beta > 0$ and $d \leq (1 - \delta)\overline{\Delta}$;
(3) $\beta > 0$ and $d > (1 - \delta)\overline{\Delta}$,
where $\overline{\Delta} \triangleq \frac{1 + \sqrt{\beta\gamma}}{1 - \sqrt{\beta\gamma}}$.

\paragraph{\textbf{Case (1): $\beta = 0$}}
Fix $1 \leq d < \Delta$.
In this case, it holds that $\chi = \1$.
Hence, it holds that $\theta^\chi \lambda \leq \lambda \leq \lambda_{c, \delta}$, where $\lambda_{c, \delta}$ is defined in \Cref{lem:lambda-unique}.
By \Cref{lem:lambda-unique}, we have $(\beta, \gamma, \lambda)$ is $d$-unique with gap $\delta$.

\paragraph{\textbf{Case (2): $\beta > 0$ and $d \leq (1 - \delta)\overline{\Delta}$}}
In this case, $(\beta, \gamma, \theta^\chi \lambda)$ is $d$-unique with gap $\delta$ due to \Cref{lem:lambda-unique}.

\paragraph{\textbf{Case (3): $\beta > 0$ and $d > (1 - \delta)\overline{\Delta}$}}
To handle this case, we need the following result.

\begin{lemma}[\text{\cite[Lemma 21 (7)]{LLY13}}] \label{lem:lambda-2-part}
  Let $\Delta \geq 3$ be an integer, and let $\beta, \gamma, \lambda$ be real numbers such that $0\leq \beta \leq \gamma \leq 1$, $\gamma > 0$, $\lambda > 0$.
  Let $\delta \in (0, 1)$ be a real number.
  Then $(\beta, \gamma, \lambda)$ is up-to-$\Delta$ unique with gap $\delta$ if and only if $\lambda \in (0, \lambda_{1, \delta}] \cup [\lambda_{2,\delta}, \infty)$ where
  \begin{align*}
    \lambda_{1,\delta} &\triangleq \min_{(1 - \delta)\overline{\Delta} < d < \Delta} \lambda_{1,\delta}(d) \\
    \lambda_{2,\delta} &\triangleq \max_{(1 - \delta)\overline{\Delta} < d < \Delta} \lambda_{2,\delta}(d),
  \end{align*}
where $\lambda_{1,\delta}(d)$ and $\lambda_{2,\delta}(d)$ are defined in \Cref{lem:lambda-unique}.
\end{lemma}
\Cref{lem:lambda-2-part} can be verified by routinely going through the proof in~\cite{LLY13} and taking the gap $\delta$ into consideration.

We assume that $(\Delta - 1) > (1 - \delta)\overline{\Delta}$. Otherwise, the integer $(1-\delta) \overline{\Delta} < d < \Delta$ does not exist.
If $\chi = \1$, then it holds that $\lambda \leq \tp{\frac{\gamma}{\beta}}^{\Delta/2}$.
By \Cref{lem:lambda-unique}, it holds that $\lambda \leq \lambda_{1,\delta}(\Delta - 1) < \lambda_{2,\delta}(\Delta - 1) \leq \lambda_{2,\delta}$.
Hence, by \Cref{lem:lambda-2-part}, we could conclude that $\lambda \leq \lambda_{1, \delta}$.
Hence for all $(1 - \delta)\overline{\Delta} < d < \Delta$, it holds that $\theta^\chi \lambda \leq \lambda \leq \lambda_{1,\delta}(d)$.
By \Cref{lem:lambda-unique}, it holds that $(\beta, \gamma, \theta^\chi \lambda)$ is $d$-unique with gap $\delta$.
The case $\chi = -1$ can be proved in a similar way.

We prove the second part of the lemma. Again, we consider three cases: (1) $\beta = 0$; (2) $\beta > 0$ and $\Delta - 1 \leq (1-\delta)\overline{\Delta}$; (3) $\beta > 0$ and $\Delta - 1 > (1-\delta)\overline{\Delta}$. Case (1) and (2) follow from the same proof. For case (3), we cannot use \Cref{lem:lambda-2-part} because we no longer have $\gamma \leq 1$. However, for the second part, we only need to prove $(\beta,\gamma,\theta^\chi \lambda)$ is $(\Delta-1)$-unique.
If $\chi = \1$, then it holds that $\lambda \leq \tp{\frac{\gamma}{\beta}}^{\Delta/2}$.
By \Cref{lem:lambda-unique}, it holds that $\lambda \leq \lambda_{1,\delta}(\Delta - 1)$.
By \Cref{lem:lambda-unique}, it holds that $(\beta, \gamma, \theta^\chi \lambda)$ is $(\Delta-1)$-unique with gap $\delta$.
The case $\chi = -1$ can be proved in a similar way.
\end{proof}

Finally, we prove the \emph{furthermore} part of \Cref{lem:C-SI-mu}, which states that if $(\Delta - 1) \leq (1 - \delta)\overline{\Delta}$, then for $\*\theta \in (0, 1]^V$, $\nu \triangleq \*\theta^\chi * \mu$ is $(\frac{144}{\delta},0)$-spectrally independent for all $\chi \in \{\0, \1\}$.

\begin{lemma} [\text{\cite[Lemma 36]{chen2020rapid}}]\label{lem:boundedness-1}
  Let $\Delta \geq 3$ be an integer, and let $\beta, \gamma, \lambda$ be real numbers satisfying that $0 \leq \beta \leq \gamma$, $\gamma > 0$, $\lambda > 0$, $\beta\gamma < 1$, and $\sqrt{\beta\gamma} > \frac{\Delta - 2}{\Delta}$.
  For every $y \in [-\infty, +\infty]$, it holds that
  \begin{align*}
    \abs{h(y)} \leq \frac{1.5}{\Delta}.
  \end{align*}
\end{lemma}

\begin{remark}
\Cref{lem:boundedness-1} is the case S.1 in~\cite[Lemma 36]{chen2020rapid}. 
In \cite{chen2020rapid}, the result is stated for $y \in J$ for some interval $J$.
The proof works for all $y \in [-\infty, +\infty]$ (see proof of Lemma 36 in~\cite{chen2020rapid}).
%
\end{remark}


Note that $\sqrt{\beta \gamma} > \frac{\Delta - 2}{\Delta}$ is equivalent to $(\Delta - 1) < \overline{\Delta}$, which can be deduced from $(\Delta - 1) \leq (1 - \delta)\overline{\Delta}$.
Note that the boundedness condition is guaranteed by  \Cref{lem:boundedness-1}.
By \Cref{lem: two-spin -SI}, \Cref{lem:contraction}, and \Cref{lem:boundedness-1},
it suffices for us to show that for any $\theta \in (0, 1]$, every $\chi \in \{\0, \1\}$, and every $1 \leq d < \Delta$, $(\beta, \gamma, \lambda)$ is $d$-unique with gap $\delta$ implies that $(\beta, \gamma, \theta^\chi \lambda)$ is $d$-unique, which holds trivially by levering \Cref{lem:lambda-unique}.

\subsection{Verifying complete marginal stability}
In this section, we prove \Cref{lem:verify-cbmr}.
Recall that $\+I = (G,\beta,\gamma,\lambda)$ is an anti-ferromagnetic two-spin system instance satisfying \Cref{condition-proof-two-spin} with parameter $\delta \in (0,1)$.
Let $\Delta \geq 3$ denote the maximum degree of $G$.
Let $\mu$ denote the Gibbs distribution of $\+I$.
Let $\pi = \-{flip}(\mu, \chi)$ be the flipped distribution, where $\chi$ is defined in \eqref{eq:uniform-flip}.
We show that $\pi$ is completely $\exp(12^5)$-marginally stable.

Recall that  $\pi$ is the Gibbs distribution of  $\+I_{
\-{flip}}=(G, \bar{\beta},\bar{\gamma},\bar{\lambda})$ defined in~\eqref{eq-def-bar}.
By \Cref{observation-I-flip},
\begin{align}\label{eq-def-local-regular}
	\bar{\beta} \geq 0, \bar{\gamma} > 0, \bar{\beta}\bar{\gamma}<1, \text{ and } 0 < \bar{\lambda} \leq \tp{\frac{\bar{\gamma}}{\bar{\beta}}}^{\Delta/2},
\end{align}
To establish the complete marginal stability, we need to show that $(\*\phi * \pi)$ is marginally stable for all $\*\phi \in (0,1]^V$.
Equivalently, we consider the more general two-spin system instance  $\+J = (G=(V,E), \bar{\beta},\bar{\gamma},(\bar{\lambda}_v)_{v \in V})$ with local fields such that
\begin{align}\label{eq-regular-1}
	\bar{\beta} \geq 0, \bar{\gamma} > 0, \bar{\beta}\bar{\gamma}<1, \text{ and } 0 < \bar{\lambda}_v \leq \tp{\frac{\bar{\gamma}}{\bar{\beta}}}^{\Delta/2} \forall v \in V.	
\end{align}
Let $\nu$ be the Gibbs distribution of $\+J$, we will show that $\nu$ is $\exp(12^5)$-marginally stable.
Note that $\+I  = (G,\beta,\gamma,\lambda)$ satisfies \Cref{condition-proof-two-spin}, which implies
\begin{align}\label{eq-regular-2}
	G \text{ is regular or} \max\{\bar{\beta},\bar{\gamma}\} \leq 1.
\end{align}
%
%

%
%
To prove \Cref{lem:verify-cbmr}, we need the following technical lemmas.
\begin{lemma}\label{lem: two-spin -aux}
For any $0 < \bar{\lambda}_v \leq \tp{\frac{\bar{\gamma}}{\bar{\beta}}}^{\Delta/2}$, $(\bar{\beta},\bar{\gamma}, \bar{\lambda}_v)$ is $(\Delta-1)$-unique (with gap 0) and it holds that 
\begin{itemize}
	\item $\bar{\lambda}_v \bar{\gamma}^{-\Delta} \leq  12^4$;
	\item $\bar{\lambda}_v \bar{\gamma}^{-\Delta} (1 - \bar{\beta}\bar{\gamma}) \leq \frac{12^5}{\Delta}$.
\end{itemize}
\end{lemma}

We remark that
compared to the assumption in~\eqref{eq:2-spin-parameters-WLOG}, $(\bar{\beta},\bar{\gamma}, \bar{\lambda}_v)$ may not always satisfy $\bar{\beta} \leq \bar{\gamma}$, but the definition of the uniqueness condition literally follows \Cref{definition-d-uniqueness}.
The uniqueness condition is well-defined because $F_d(x) = x$ has a unique solution if $\bar{\beta}\bar{\gamma} < 1$.

\begin{lemma}\label{lem:cmbr-worst-case}
  Let $d = \Delta - 1$ and $\bar{\lambda}_{\max} = \max_{u \in V} \bar{\lambda}_u$.
  Let $v \in V$, $S \subseteq \Lambda \subseteq V \setminus \{v\}$, and $\sigma \in \Omega(\mu_\Lambda)$ be a partial pinning.
  It holds that
  \begin{align*}
    R^\sigma_v
    &\leq \frac{F_{\Delta}^{\bar{\lambda}_{\max}}(0)}{F_{\Delta}^{\bar{\lambda}_{\max}} \circ F_d^{\bar{\lambda}_{\max}} (0)}  R^{\sigma_S}_v,
  \end{align*}
  where $R^\sigma_v = \nu^\sigma_v(\1)/\nu^\sigma_v(\0)$ is the marginal ratio of $\nu^\sigma$, and 
  \begin{align} \label{eq:tree-recur}
   \forall \lambda > 0, d\in \mathds{Z}_{>0},\quad F_d^{\lambda}(x) = \lambda \tp{\frac{\beta x + 1}{x + \gamma}}^d,
  \end{align}
  is the uniform tree-recursion function. 
\end{lemma}
\begin{remark}
  Intuitively, \Cref{lem:cmbr-worst-case} says that the worst case of $R^\sigma_v/R^{\sigma_S}_v$ is achieved by a $\Delta$-regular tree rooted at $v$, where $\sigma_S$ fixes the values of all the vertices in $\{u \in V \mid \-{dist}_G(u, v) = 2\}$ to $\0$ and $\sigma$ further fixes the values of  all the vertices in $\{u \in V\mid \-{dist}_G(u, v) = 1\}$ to $\0$.
\end{remark}

The proofs of \Cref{lem: two-spin -aux} and \Cref{lem:cmbr-worst-case} are deferred to \Cref{sec: two-spin -aux-C} and \Cref{section-proof-aux-1} respectively.


We are ready to prove \Cref{lem:verify-cbmr}.
\begin{proof}[Proof of \Cref{lem:verify-cbmr}]
Let $\nu$ be the Gibbs distribution of $\+J$.
To prove that $\pi$ is complete $\exp(12^5)$-marginally stable, it suffices for us to show that $\nu$ is $\exp(12^5)$-marginally stable.

Let $S \subseteq \Lambda \subseteq V$, $v \in V \setminus \Lambda$ and $\sigma \in \Omega\tp{\nu_{V\setminus\Lambda}}$ be a partial configuration on $V \setminus \Lambda$.
We will show that
\begin{align*}
  R^{\sigma}_v \le C \text{ and } R^{\sigma}_v \le C R^{\sigma_S}_v, 
\end{align*}
where $R^{\sigma} \triangleq \frac{\nu^\sigma_v(\1)}{\nu^\sigma_v(\0)}$ denotes the marginal ratio of $\nu^\sigma$, and $C=\exp(12^5)$ be a universal constant.

For the first part, by considering the worst pinning of all neighbors of $v$, we have
\begin{align*}
  R^\sigma_v  \leq \bar{\lambda}_v\bar{\gamma}^{-\Delta}  \leq 12^4 \leq C,
\end{align*}
where inequalities follow from anti-ferromagnetism and \Cref{lem: two-spin -aux} respectively.
For the second part, we may assume that $R_v^{\sigma_S} > 0$, otherwise $R_v^\sigma = R_v^{\sigma_S} = 0$.
%
By \Cref{lem:cmbr-worst-case}, it holds that 
\begin{align*}
  \frac{R^\sigma_v}{R^{\sigma_S}_v} &\le \frac{F_{\Delta}^{\bar{\lambda}_{\max}}(0)}{F_{\Delta}^{\bar{\lambda}_{\max}} \circ F_d^{\bar{\lambda}_{\max}} (0)}
  = \frac{\bar{\lambda}_{\max} \bar{\gamma}^{-\Delta}}{\bar{\lambda}_{\max} \tp{\frac{\bar{\beta}\bar{\lambda}_{\max} \bar{\gamma}^{-d} + 1}{\bar{\lambda}_{\max} \bar{\gamma}^{-d} + \bar{\gamma}}}^\Delta}
    = \tp{\frac{\bar{\lambda}_{\max} + \bar{\gamma}^{d+1}}{\bar{\gamma}(\bar{\lambda}_{\max} \bar{\beta} + \bar{\gamma}^d)}}^\Delta \\
  &= \tp{1 + \frac{\bar{\lambda}_{\max} ( 1 - \bar{\beta}\bar{\gamma})}{ \bar{\lambda}_{\max} \bar{\beta}\bar{\gamma} + \bar{\gamma}^{d+1}}}^\Delta
   \leq \tp{1 + \frac{\bar{\lambda}_{\max} (1 - \bar{\beta}\bar{\gamma})}{\bar{\gamma}^{d+1}}}^\Delta
   \overset{(\star)}{\leq} \tp{1 + \frac{12^5}{\Delta}}^\Delta \leq \exp(12^5) = C,
\end{align*}
where $(\star)$ holds by \Cref{lem: two-spin -aux}. This concludes the proof.
\end{proof}

\subsubsection{Proof of \Cref{lem: two-spin -aux}} \label{sec: two-spin -aux-C}
In this section, we prove \Cref{lem: two-spin -aux}.
We first show that $(\bar{\beta},\bar{\gamma},\bar{\lambda}_v)$ is $d$-unique (with gap 0) for $d = \Delta - 1$.
Note that by \Cref{condition-proof-two-spin}, $(\beta, \gamma,\lambda)$ is $d$-unique with gap $\delta$.

Suppose $\lambda \leq (\frac{\gamma}{\beta})^{\Delta/ 2}$, then we have $\bar{\lambda} = \lambda, \bar{\beta} = \beta, \bar{\gamma} = \gamma$.
By \Cref{lem:lambda-unique}, we know that when $d \leq (1 - \delta) \overline{\Delta}$, it holds that $(\bar{\beta}, \bar{\gamma}, \bar{\lambda}_v)$ is $d$-unique with gap $\delta$; and when $d > (1 - \delta)\overline{\Delta}$, it holds that $\bar{\lambda}_v \leq \lambda \leq \lambda_{1,\delta}(d) < (\frac{\gamma}{\beta})^{\Delta/2}$, which implies $(\bar{\beta}, \bar{\gamma}, \bar{\lambda}_v)$ is $d$-unique with gap $\delta$.
The $\lambda > (\frac{\gamma}{\beta})^{\Delta/2}$ case is almost the same by noticing that when we fix $d$ and $\delta$, then $\lambda_{1,\delta}(d)$ and $\lambda_{2, \delta}(d)$ are actually functions of $\beta, \gamma$ that could be written as $\lambda_{1,\delta}(d;\beta, \gamma), \lambda_{2,\delta}(d; \beta, \gamma)$, and 
\begin{align*}
  \lambda \geq \lambda_{2,\delta}(d; \beta, \gamma) \quad \Longleftrightarrow \quad \bar{\lambda} \leq \lambda_{1, \delta}(d; \bar{\beta}, \bar{\gamma}),
\end{align*}
where $\bar{\lambda} = 1/\lambda, \bar{\beta} = \gamma$, and $\bar{\gamma} = \beta$.
Finally, note that since $(\bar{\beta}, \bar{\gamma}, \bar{\lambda}_v)$ is $d$-unique with gap $\delta$, it is also $d$-unique (with gap $0$).


Combining with \Cref{observation-I-flip}, it suffices to prove the following result: for any $\beta, \gamma, \lambda$ with $\beta \geq 0$, $\gamma > 0$, $\beta\gamma < 1$ and $0 < \lambda \leq (\gamma/\beta)^{\Delta/2}$ that is $(\Delta-1)$-unique, it holds that  $\lambda\gamma^{-\Delta} \leq 12^4$ and $\lambda \gamma^{-\Delta}(1-\beta\gamma ) \leq 12^5/\Delta$. We need the following lemma.

\begin{lemma}[\text{\cite[Lemma 35]{chen2020rapid}}]\label{lem: two-spin -aux-aux}
  Let $\Delta \geq 3$ be an integer and $d \triangleq \Delta - 1$.
  Let $\beta, \gamma, \lambda$ be real numbers such that $\beta \geq 0, \gamma > 0, \beta\gamma < 1, \lambda > 0$ and $(\beta, \gamma, \lambda)$ is $d$-unique (with gap 0). 
  \begin{enumerate}
    \item If $\beta = 0$, then we have $\lambda \leq \frac{4\gamma^{d+1}}{d-1}$.
  \item If $\beta > 0$ and $\sqrt{\beta\gamma} \leq \frac{\Delta-2}{\Delta}$, it holds that
    \begin{align*}
      \text{either} \quad &\lambda \leq \frac{18\gamma^{d+1}}{\theta(d)} \quad
      \text{or} \quad \lambda \geq \frac{\theta(d)}{18\beta^{d+1}},
    \end{align*}
    where $\theta(d) \triangleq d(1 - \beta\gamma) - (1 + \beta\gamma)$.
  \end{enumerate}
\end{lemma}

\begin{remark}
Lemma 35 in \cite{chen2020rapid} further assumes $\beta \leq \gamma$.
We remark that \Cref{lem: two-spin -aux-aux} can be verified by routinely going through the proof in~\cite{chen2020rapid}. 
\end{remark}

We first show that $\lambda\gamma^{-\Delta} \leq  12^4$.
Let $\overline{\Delta}\triangleq \frac{1 + \sqrt{\beta\gamma}}{1 - \sqrt{\beta\gamma}}$.
We consider $3$ cases:
(1) $\beta > 0$ and $\Delta < 2\overline{\Delta}$;
(2) $\beta > 0$ and $\Delta \geq 2\overline{\Delta}$;
(3) $\beta = 0$.

\paragraph{\textbf{Case (1): $\beta > 0$ and $\Delta < 2\overline{\Delta}$}}
Note that we have $\lambda \leq \tp{\frac{\gamma}{\beta}}^{\Delta/2}$, so it holds that
\begin{align*}
  \lambda\gamma^{-\Delta} \leq \tp{\frac{1}{\beta\gamma}}^{\Delta / 2} \leq  \tp{\beta\gamma}^{-\Delta} \leq \tp{\beta\gamma}^{-2\overline{\Delta}}.
\end{align*}

Note that from $\overline{\Delta} = \frac{1 + \sqrt{\beta\gamma}}{1 - \sqrt{\beta\gamma}}$, we have $(\beta\gamma)^{-1} = \tp{\frac{\overline{\Delta} + 1}{\overline{\Delta} - 1}}^2$.
Moreover, we have
\begin{align*}
  (\beta\gamma)^{- 2 \overline{\Delta}}
    = \tp{\frac{\overline{\Delta} + 1}{\overline{\Delta} - 1}}^{4\overline{\Delta}}
    \leq 12^{4},
\end{align*}
where in the last inequality, we use the fact that $3 \leq \Delta < 2 \overline{\Delta}$ which means $\overline{\Delta} > \frac{3}{2}$.
\paragraph{\textbf{Case (2): $\beta > 0$ and $\Delta \geq 2\overline{\Delta}$}}
In this case, $d \geq \frac{2}{3}\Delta \geq \frac{4}{3}\overline{\Delta} \geq  \overline{\Delta}$ is achieved, which means $\sqrt{\beta\gamma} \leq \frac{\Delta - 2}{\Delta}$.
In this case, by \Cref{lem: two-spin -aux-aux} with $\sqrt{\beta\gamma} \leq \frac{\Delta - 2}{\Delta}$ and the fact that $\lambda \leq \tp{\frac{\gamma}{\beta}}^{\Delta/2}$, it holds that
\begin{align*}
  \lambda\gamma^{-\Delta} &\leq \frac{18}{\theta(d)},
\end{align*}
where $d \triangleq \Delta - 1$ and $\theta(d) = d(1 - \beta\gamma) - (1 + \beta\gamma)$.
Note that we have
\begin{align*}
  \theta(d)
  &\overset{(\star)}{=} d\tp{1 - \tp{\frac{\overline{\Delta} - 1}{\overline{\Delta} + 1}}^2} - \tp{1 + \tp{\frac{\overline{\Delta} - 1}{\overline{\Delta} + 1}}^2}\\ 
  &= \frac{4\overline{\Delta}d - 2(\overline{\Delta}^2 + 1)}{(\overline{\Delta} + 1)^2}
   \overset{(*)}{\geq} \frac{\frac{4}{3} \cdot 4\overline{\Delta}^2 - 2(\overline{\Delta}^2 + 1)}{(\overline{\Delta} + 1)^2}
   \overset{(+)}{\geq} \frac{\frac{4}{3} \overline{\Delta}^2}{(\overline{\Delta} + 1)^2} \geq \frac{1}{3},
\end{align*}
where $(\star)$ holds by the fact that $\beta\gamma = \tp{\frac{\overline{\Delta} - 1}{\overline{\Delta} + 1}}^2$, $(*)$ holds by the fact that $d \geq \frac{2}{3}\Delta \geq \frac{4}{3}\overline{\Delta}$, $(+)$ holds by the fact that $\overline{\Delta} \geq 1$, and the last inequality holds by the fact that the function $f(x) \triangleq \frac{4x^2}{(x+1)^2}$ is monotone increasing when $x > 0$ and $f(1) = 1$.
Hence, in this case, we have $\lambda\gamma^{-\Delta} \leq 54$.
\paragraph{\textbf{Case (3): $\beta = 0$}}
In this case, by \Cref{lem: two-spin -aux-aux}, it holds that
\begin{align*}
  \lambda\gamma^{-\Delta} \leq \frac{4}{d-1} \leq 4,
\end{align*}
where $d \triangleq \Delta - 1 \geq 2$.

We next show that $\lambda\gamma^{-\Delta}\tp{1 - \beta\gamma} \leq \frac{12^5}{\Delta}$.
Let $\overline{\Delta} \triangleq \frac{1 + \sqrt{\beta\gamma}}{1 - \sqrt{\beta\gamma}}$, we consider $3$ cases: 
(1) $\beta > 0$ and $\Delta < 2 \overline{\Delta}$;
(2) $\beta > 0$ and $\Delta \geq 2 \overline{\Delta}$;
(3) $\beta = 0$.

\paragraph{\textbf{Case (1): $\beta > 0$ and $\Delta < 2 \overline{\Delta}$}}
First, by the previous result, it holds that $\lambda\gamma^{-\Delta} \leq 12^4$.
Note that from $\overline{\Delta} = \frac{1 + \sqrt{\beta\gamma}}{1 - \sqrt{\beta\gamma}}$, we have $(\beta\gamma)^{-1} = \tp{\frac{\overline{\Delta} + 1}{\overline{\Delta} - 1}}^2$ which implies
\begin{align*}
  1 - \beta\gamma
  &= 1 - \tp{\frac{\overline{\Delta} - 1}{\overline{\Delta} + 1}}^2
  = \frac{4\overline{\Delta}}{(\overline{\Delta} + 1)^2}
  \leq \frac{4}{\;\overline{\Delta}\;}
  \leq \frac{8}{\Delta},
\end{align*}
where in the last inequality, we use the fact that $\Delta < 2 \overline{\Delta}$.
Hence, it holds that
\begin{align*}
  \lambda\gamma^{-\Delta}(1 - \beta\gamma) \leq \frac{8 \cdot 12^4}{\Delta}.
\end{align*}

\paragraph{\textbf{Case (2): $\beta > 0$ and $\Delta \geq 2 \overline{\Delta}$}}
Note that $\frac{3}{2}d \geq \Delta \geq 2\overline{\Delta}$, it holds that $d \geq \frac{4}{3} \overline{\Delta} \geq \overline{\Delta}$, which means $\sqrt{\beta\gamma} \leq \frac{\Delta - 2}{\Delta}$.
By \Cref{lem: two-spin -aux} with $\sqrt{\beta\gamma} \leq \frac{\Delta - 2}{\Delta}$ and the fact that $\lambda \leq \tp{\frac{\gamma}{\beta}}^{\Delta/2}$, it holds that:
\begin{align*}
  \lambda  \leq \frac{18\gamma^{d+1}}{\theta(d)},
\end{align*}
where $d \triangleq \Delta - 1$ and $\theta(d) \triangleq d(1 - \beta\gamma) - (1 + \beta\gamma)$.
This lead us to
\begin{align*}
  \frac{\lambda (1 - \beta\gamma)}{\gamma^{d+1}}
  &\leq \frac{18 (1 - \beta\gamma)}{\theta(d)}
    = \frac{18(1 - \beta\gamma)}{d(1 - \beta\gamma) - (1 + \beta\gamma)} = \frac{18}{d - \frac{1 + \beta\gamma}{1 - \beta\gamma}}\\
  (\text{by } \Delta = d + 1 \geq 2 \overline{\Delta})\quad&\leq \frac{18}{d - \frac{(1 + \sqrt{\beta\gamma})^2}{(1 + \sqrt{\beta\gamma})(1 - \sqrt{\beta\gamma})}} =  \frac{18}{d - \frac{1 + \sqrt{\beta\gamma}}{1 - \sqrt{\beta\gamma}}}\\
  &= \frac{18}{d - \overline{\Delta}}
  \overset{(\star)}{\leq} \frac{72}{d}
  \overset{(*)}{\leq} \frac{108}{\Delta},
\end{align*}
where $(\star)$ is deduced by $d \geq \frac{4}{3}\overline{\Delta}$, and $(*)$ is because $\Delta \leq \frac{3}{2}d$.

\paragraph{\textbf{Case (3): $\beta = 0$}}
By \Cref{lem: two-spin -aux}, it holds that
\begin{align*}
  \lambda \leq \frac{4\gamma^{d+1}}{d - 1},
\end{align*}
which will lead us to
\begin{align*}
  \frac{\lambda (1 - \beta\gamma)}{\gamma^{d+1}}
  &\leq \frac{4}{d - 1} \leq \frac{12}{\Delta},
\end{align*}
where the last inequality comes from the fact that $3(d - 1) \geq \Delta$.
\subsubsection{Tree recursion analysis}\label{section-proof-aux-1}
{
  In order to prove \Cref{lem:cmbr-worst-case}, we first introduce the self-avoiding walk tree (SAW) in \cite{weitz2006counting}. 
  Given a graph $G=(V,E)$ with pinning $\sigma \in \{\0,\1\}^\Lambda$ on $\Lambda \subseteq V$, fields $\*\lambda \in \mathds{R}^V$ and vertex $v \in V$, the self-avoiding walk tree $T=T_{\mathrm{SAW}}(G,v)=(V_T,E_T)$ with fields $\*\lambda \in \mathds{R}^{V_T}$ is recursively constructed as follows.
  \begin{enumerate}
    \item If vertex $v$ is pinned, return the single vertex $v$ (with field $\lambda_v$).
    \item Otherwise, let $u_1,u_2,\ldots,u_d$ be the neighbors of $v$. 
    For each $1 \le i \le d$, denote $G_i$ be the graph obtained by deleting $v$, attaching new vertices $v_j$ with pinning $\0$ to vertices $u_j$ for all $1 \le j < i$, 
    and attaching new vertices $v_j$ with pinning $\1$ to vertices $u_j$ for all $i < j \le d$. 
    \item Let $T$ be a rooted tree at vertex $v$ (with field $\lambda_v$) with subtrees $T_1,T_2,\ldots, T_d$ (with fields $\*\lambda_{T_1},\*\lambda_{T_2},\ldots,\*\lambda_{T_d}$), where $T_i = T_{\mathrm{SAW}}(G_i,u_i)$. 
  \end{enumerate}
  Furthermore, given fields $\*\lambda \in \mathds{R}^V$ in $G=(V,E)$, 
  Observed in \cite{weitz2006counting}, the self-avoiding walk tree preserves marginal ratio.
  \begin{proposition}[\cite{weitz2006counting, LLY13}]\label{prop:saw}
    Let  $G=(V,E)$ be a graph, $\beta \in \mathds{R}_{\ge 0},\gamma \in \mathds{R}_{>0},\*\lambda \in \mathds{R}^V_{> 0}$ be parameters, $\sigma \in \{\0,\1\}^\Lambda$ be a valid pinning on $\Lambda \subseteq V$, and $v \in V$ be a vertice. 
    Denote the Gibbs distribution of two-spin model $(G,\beta,\gamma,\*\lambda)$ and $(T=T_{\mathrm{SAW}}(G,v),\beta,\gamma,\*\lambda_T)$ by $\mu_G$ and $\mu_T$ respectively. Then
    \begin{align*}
      \frac{\mu_{G,v}(\1)}{\mu_{G,v}(\0)} = \frac{\mu_{T,v}(\1)}{\mu_{T,v}(\0)}.
    \end{align*}  
    Furthermore, denote the marginal ratio $\frac{\mu_{T_u,u}(\1)}{\mu_{T_u,u}(\0)}$ by $R_u$, where $T_u$ is the subtree rooted at $u$ and $\mu_{T_u,u}$ be the Gibbs distribution of two-spin model $(T_u,\beta,\gamma,\*\lambda_{T_u})$. 
    For all $u \in T$, the marginal ratio $R_u$ satisfies
    \begin{align*}
      R_u = \lambda \prod_{i=1}^d \tp{\frac{\beta R_{u_i}+1}{R_{u_i}+\gamma}},
    \end{align*}
    where $u_1,u_2,\ldots,u_d$ denotes the children of $u$ in $T_u$.
  \end{proposition}
  \begin{proof}[Proof of \Cref{lem:cmbr-worst-case}]
    Without loss of generality, we may assume $R^{\sigma_S}_v > 0$.
    Denote the neighbors of $v$ in $G$ by $N_G(v)$. 
    Let 
    \begin{align*}
      S_0 = N_G(v) \setminus S, S_{\0} = N_G(v) \cap \sigma_S^{-1}(\0),\text{ and } S_{\1}=N_G(v) \cap \sigma_S^{-1}(\1).
    \end{align*}
    By monotonicity of anti-ferromagnetic two-spin system,
    \begin{align}\label{eq:worst-bound}
      R^\sigma_v = \frac{\nu^{\sigma}_v(\1)}{\nu^{\sigma}_v(\0)} \le \bar{\lambda}_v \bar{\gamma}^{-\abs{S_0}-\abs{S_{\0}}} \bar{\beta}^{\abs{S_{\1}}}
    \end{align}
    Let $T$ be the self-avoiding walk tree of $G$ with pinning $\sigma_S$ and $\mu$ be the Gibbs distribution of two-spin model $(T,\bar{\beta},
    \bar{\gamma},\bar{\*\lambda}_T)$.
    Let $N_T(v)$ denote all children of vertex $v$ in $T$.
     By \Cref{prop:saw},
    \begin{align}\label{eq:recursion}
      R^{\sigma_S}_v = \frac{\nu^{\sigma_S}_v(\1)}{\nu^{\sigma_S}(\0)} = \frac{\mu_{T,v}(\1)}{\mu_{T,v}(\0)} = \bar{\lambda}_v \prod_{u \in N_T(v)} \tp{\frac{\bar{\beta} \frac{\mu_{T_u,u}(\1)}{\mu_{T_u,u}(\0)}+1}{\frac{\mu_{T_u,u}(\1)}{\mu_{T_u,u}(\0)}+\bar{\gamma}}},
    \end{align}
    where $T_u$ is the subtree rooted at $u$ and $\mu_{T_u,u}$ be the Gibbs distribution of two-spin model $(T_u,\bar{\beta},\bar{\gamma},\bar{\*\lambda}_{T_u})$.
    From the construction of $T$, we have the following properties.
    \begin{enumerate}
      \item There are $\abs{S_{\0}}$ children of vertex $v$ with pinning $\0$, $\abs{S_{\1}}$ children with pinning $\1$, and $\abs{S_{0}}$ children without pinning;
      \item For each $u \in N_T(v)$, $\bar{\lambda}_{T,u} \le \bar{\lambda}_{\max}$.
    \end{enumerate}
    By monotonicity of anti-ferromagnetic two-spin system, for each $u \in N_T(v)$ without pinning,
    \begin{align*}
      \frac{\mu_{T_u,u}(\1)}{\mu_{T_u,u}(\0)} \le \bar{\lambda}_{\max} \bar{\gamma}^{-d_u},
    \end{align*}
    where $d_u \le d = \Delta - 1$ is the number of children of $u$ in $T_u$. 
    Note that if $G$ is regular, then $d_u = d$.
    By~\eqref{eq-regular-2},
    \begin{align*}
      \frac{\mu_{T_u,u}(\1)}{\mu_{T_u,u}(\0)} \le \bar{\lambda}_{\max} \bar{\gamma}^{-d} = F^{\bar{\lambda}_{\max}}_d(0).
    \end{align*}
    Together with~\eqref{eq:recursion} and the monotonicity of anti-ferromagnetic two-spin system,
    \begin{align*}
      R^{\sigma_S}_v \ge \bar{\lambda}_v \bar{\beta}^{\abs{S_{\1}}} \bar{\gamma}^{-\abs{S_{\0}}} \tp{\frac{\bar{\beta} F^{\bar{\lambda}_{\max}}_d(0)+1}{F^{\bar{\lambda}_{\max}}_d(0) + \bar{\gamma}} }^{\abs{S_0}}
    \end{align*}
    Combining with~\eqref{eq:worst-bound}, 
    \begin{align*}
      \frac{R^\sigma_v}{R^{\sigma_S}_v}
      \le \tp{\frac{\bar{\gamma}(\bar{\beta} F^{\bar{\lambda}_{\max}}_d(0) + 1)}{F^{\bar{\lambda}_{\max}}_d(0) + \bar{\gamma}} }^{-\abs{S_0}} 
      \le \tp{\frac{\bar{\gamma}(\bar{\beta} F^{\bar{\lambda}_{\max}}_d(0) + 1)}{F^{\bar{\lambda}_{\max}}_d(0) + \bar{\gamma}} }^{-\Delta}
      = \frac{F_{\Delta}^{\bar{{\lambda}}_{\max}}(0)}{F_{\Delta}^{\bar{\lambda}_{\max}} \circ F_d^{\bar{\lambda}_{\max}} (0)}. 
    \end{align*}
    This concludes the proof.
  \end{proof}
}

\subsection{Modified log-Sobolev constant in subcritical regime} 
In this section, we prove \Cref{lem:good-regime-mLSI}.
In this proof, we consider the continuous-time Markov chain. Let $\Omega$ be a discrete and finite state space. Let matrix $Q: \Omega \times \Omega \to \mathds{R}_{\geq 0}$ denote the \emph{transition rate}. We remark that the row sum of $Q$ may \emph{not} be 1. 
The continuous-time Markov chain is a stochastic process $(Y_t)_{t \in \mathds{R}_{\geq 0}}$, For any $t > 0$, $Y_t$ follows the distribution $P_t(Y_0,\cdot)$ and $P_t = \exp(Lt) = \sum_{k = 0}^\infty \frac{t^kL^k}{k!}$, where the \emph{generator} $L$ of the continuous time Markov chain is an operator defined by for any $\psi: \Omega \to \mathds{R}$,
\begin{align*}
	L\psi(x) = \sum_{y \in \Omega}Q(x,y)(\psi(y) - \psi(x)). 
\end{align*}
Suppose $Q$ satisfies the detailed balance equation with respect to the distribution $b:\Omega \to \mathds{R}_{>0}$, i.e.
\begin{align*}
	\forall x, y \in \Omega, \quad b(x)Q(x,y) = b(y)Q(y,x).
\end{align*}
The modified log-Sobolev constant for continuous-time Markov chain is defined by
\begin{align*}
\rho(Q) \triangleq \min\left\{ \frac{\+E_{Q}(f, \log f)}{\Ent[b]{f}} \mid \forall f:\Omega \to \mathds{R}_{\geq 0}:\Ent[b]{f} \neq 0  \right\},
\end{align*}
where the Dirichlet form $\+E_{Q}(f, \log f)$ is defined by
\begin{align*}
	\+E_{Q}(f, \log f) \triangleq \frac{1}{2}\sum_{x, y \in \Omega}b(x)Q(x,y)(f(x) - f(y))(\log f(x) - \log f(y)).
\end{align*}

Back to our proof, let $\theta \leq 12^{-6}$ be a constant.
Note that $\theta * \pi$ could be seen as a  two-spin  system with parameters $\bar{\beta},\bar{\gamma}$ and $\theta\bar{\lambda}$. 
Fix $\Lambda \subseteq [n]$ and $\sigma \in \Omega(\pi_\Lambda)$, let $\nu \triangleq (\theta * \pi)^\sigma$, we will show that
\begin{align*}
  \rho^{\-{GD}}(\nu) \geq \frac{1}{4n}.
\end{align*}
Let $\Omega \triangleq \Omega(\nu)$.
For $i \in [n]$, $\eta_i:\Omega \to \{\0,\1\}^V$ is defined as
\begin{align*}
  \forall x \in \Omega, \forall j \in [n], \quad (\eta_ix)_j \triangleq \begin{cases}
    x_j & j \neq i \\
    - x_j & j = i
  \end{cases},
\end{align*}
where for convenience, we denote $\eta_i(x)$ as $\eta_i x$.

The continuous-time Glauber dynamics over $\nu$ has the transition rate $Q: \Omega \times \Omega \to \mathds{R}_{\geq 0}$ as
\begin{align*}
  \forall x \in \Omega, i \in [n] \setminus \Lambda, \quad Q(x, \eta_i x) &= \frac{1}{n} \frac{\nu(\eta_i x)}{\nu(x) + \nu(\eta_i x)} = \begin{cases}
    \frac{1}{n} \frac{\theta \bar{\lambda} \bar{\beta}^{s_i}}{\bar{\lambda} \bar{\beta}^{s_i} + \bar{\gamma}^{\Delta_i - s_i}}, & (\eta_ix)_i = \1 \\
    \frac{1}{n} \frac{\bar{\gamma}^{\Delta_i - s_i}}{\theta \bar{\lambda} \bar{\beta}^{s_i} + \bar{\gamma}^{\Delta_i - s_i}}, & (\eta_ix)_i = \0,
  \end{cases}
\end{align*}
\begin{align*}
	  \text{for any other } x,y \in \Omega \times \Omega \text{ not covered by the above case},\quad Q(x,y) = 0. 
\end{align*}
where $\Delta_i$ is the degree of $i$ in the graph and $s_i$ denotes the number of $\1$-neighbors of $i$ with respect to configuration $x$.
Now, we consider a tuned version of Glauber dynamics over $\nu$ whose transition rate $\widehat{Q}: \Omega \times \Omega \to \mathds{R}_{\geq 0}$ is defined as
\begin{align*}
  \forall x \in \Omega, i \in [n] \setminus \Lambda, \quad \widehat{Q}(x, \eta_i x) &= \begin{cases}
    \frac{\nu(\eta_i x)}{\nu(x)} = \theta \bar{\lambda} \bar{\gamma}^{-\Delta_i} (\bar{\beta}\bar{\gamma})^{s_i}, & (\eta_i x)_i = \1 \\
    1, & (\eta_i x)_i = \0.
  \end{cases}
\end{align*}
\begin{align*}
	  \text{for any other } x,y \in \Omega \times \Omega \text{ not covered by the above case},\quad \widehat{Q}(x,y) = 0. 
\end{align*}
It is straightforward to check that both $Q$ and $\widehat{Q}$ are reversible with respect to $\nu$.

Note that for $x \in \Omega$ and $i \in [n] \setminus \Lambda$, it holds that
\begin{align*}
  Q(x, \eta_i x) &= \frac{1}{n} \cdot \frac{\nu(x_-)}{\nu(x_-) + \nu(x_+)} \cdot \widehat{Q}(x, \eta_i x),
\end{align*}
where $x_-$ and $x_+$ are obtained from configuration $x$ with the $i$-th position being modified to $\0$ and $\1$, respectively.
If $\nu(x_-) = 0$, then it holds that $x_i = +$ and
\begin{align*}
	Q(x,\eta_i x) =\widehat{ Q} (x,\eta_i x) = 0.
\end{align*} 
If $\nu(x_-) > 0$, we have
\begin{align*}
  \frac{\nu(x_-)}{\nu(x_-) + \pi(x_+)}
  &= \frac{1}{1 + \frac{\nu(x_+)}{\nu(x_-)}}
    \overset{(\star)}{\geq} \frac{1}{1 + \theta \bar{\lambda} \bar{\gamma}^{-\Delta}}
    \geq \frac{1}{1 + \theta \cdot 12^4},
\end{align*}
where $(\star)$ could be deduced from $G$ is $\Delta$-regular or $\bar{\gamma} \leq 1$ (see \eqref{eq-regular-2}) and the last inequality holds by \Cref{lem: two-spin -aux}.
Since $\theta \leq 12^{-4}$, it holds that for any $x, y \in \Omega \times \Omega$,
\begin{align*}
{Q(x, \eta_i x)} \geq \frac{1}{2n}\widehat{Q}(x, \eta_i x),	
\end{align*}
which implies
\begin{align*}
  \forall f: \Omega \to \mathds{R}_{\geq 0}, \quad 2n \cdot \+E_Q(f, \log f) \geq \+E_{\widehat{Q}}(f, \log f).
\end{align*}
Hence, it holds that 
\begin{align} \label{eq:good-regime-reduce}
  \rho^{\-{GD}}(\nu) &\geq \frac{1}{2n} \cdot \rho^{\widehat{\-{GD}}}(\nu),
\end{align}
where we use $\rho^{\-{GD}}(\nu) $ and $\rho^{\widehat{\-{GD}}}(\nu)$  to denote the modified log-Sobolev constant of the continuous-time Glauber dynamics and continuous-time tuned Glauber dynamics respectively.
Remark that by our definition, the discrete-time Glaubder dynamics and continuous-time Glauber dynamics have the same modified log-Sobolev constant.
Hence, to prove \Cref{lem:good-regime-mLSI}, it suffices for us to bound $\rho^{\widehat{\-{GD}}}(\nu)$.

We will use the following general result.
Let $T$ be the transition rate of a continuous-time Markov chain $M$ on $\+X$.
Let $\+G$ be a set of bijective maps from $\+X$ to $\+X$.
We say $\+G$ is a mapping  representation of $T$ if 
\begin{itemize}
	\item for any $x, y \in \+X$ such that $T(x,y) > 0$, there exists a unique $\delta \in \+G$ such that $y = \delta x$;
	\item for any $\delta \in \+G$, there exists a unique $\delta^{-1} \in \+G$ such that for any $x \in \+X$, $\delta^{-1}(\delta(x)) = x$.
\end{itemize}
 
\begin{theorem} [\text{\cite[Theorem 1.1, Theorem 3.9]{EHMT17}}] \label{thm:EHMT17}
  Let $\mu$ be a distribution over a finite set $\+X$.
  Let $T$ be the transition rate of a continuous time Markov chain $M$ satisfying the detailed-balance equation with respect to $\mu$.
  Let $\+G$ a mapping  representation of $T$ satisfying $\alpha\beta x = \beta\alpha x$ for all $x \in \+X, \alpha, \beta \in G$.
  If there exist $H_1 \subseteq \+X \times \+G$ and $H_2 \triangleq \{(\alpha x, \alpha^{-1}) \mid (x, \alpha) \in H_1\}$ such that $H_1 \cap H_2 = \emptyset$, $H_1 \cup H_2 = \+X \times \+G$ and
  \begin{align*}
   \forall i \in \{1, 2\},\quad \kappa_i &\triangleq \min_{\substack{(x , \alpha) \in H_i \\ T(x, \alpha x) > 0}} \left[T(x, \alpha x) - \mathds{1}_{\alpha \neq \alpha^{-1}}T(\alpha x, \alpha (\alpha x)) - \sum_{\eta: \eta \neq \alpha, \alpha^{-1}} \frac{(q - q_*)(\alpha x, \alpha^{-1}, \eta)}{T(x, \alpha x) \mu(x)} \right] \geq 0,
  \end{align*}
  where we set $q(x, \alpha, \eta) \triangleq T(x, \alpha x)T(x, \eta x)\mu(x)$ (we assume $q(x,\alpha,\eta) = 0$ if $\mu(x) = 0$) and
  \begin{align*}
  q_*(x, \alpha, \eta) &\triangleq \min\{q(x, \alpha, \eta), q(\alpha x, \alpha^{-1}, \eta), q(\eta x, \alpha, \eta^{-1}), q(\alpha\eta x, \alpha^{-1}, \eta^{-1})\}.
  \end{align*}
  Then, we have $\rho^{M}(\mu) \geq \kappa_1 + \kappa_2$, where $\rho^{M}(\mu)$ denotes the modified log-Sobolev constant of $M$.
\end{theorem}
The above theorem is slightly different from the original theorem in~\cite{EHMT17}, but it can be proved by going through the proof in~\cite{EHMT17}. We give the proof in \Cref{section-app-Mls} for completeness.

In our proof, we define $\+G = \{\eta_i \mid i \in [n] \setminus \Lambda\}$.
Note that for any $x, y \in \Omega$ such that $\widehat{Q}(x, y) > 0$, it must hold that $x$ and $y$ disagree only at one vertex, say $i \in V$.
We have $y = \eta_ix$ and $y \neq \eta_j x$ for all $j \neq i$.
For any $\eta_ i \in \+G$, it holds that $\eta_i^{-1} = \eta_i$.
Hence, $\+G$ is a mapping representation of $\widehat{Q}$.
For any $\eta_i,\eta_j \in \+G$, any $x \in \Omega$, it is straightforward to verify that $\eta_i \eta_j x = \eta_j \eta_i x$.
We define $H_1$ and $H_2$ as 
\begin{align*}
  H_1 &= \{(x, \eta_i) \mid x \in \Omega, i \in [n] \setminus\Lambda, x_i = \0\} \\
  H_2 &= \{(x, \eta_i) \mid x \in \Omega, i \in [n] \setminus\Lambda, x_i = \1\}.
\end{align*}
It is straightforward to verify $H_1 \cup H_2 = \Omega \times \+G$, $H_1 \cap H_2 = \emptyset$, and $H_2 = \{ (\eta_i x, \eta_i^{-1}) \mid (x,\eta_i) \in H_1\}$.
In our application, $\eta_i = \eta_i^{-1}$ for all $\eta_i \in \+G$.
The $\kappa_i$ for $i \in \{1, 2\}$ could be rewritten as
\begin{align*}
 \kappa_i &= \min_{\substack{(x , \alpha) \in H_i \\ \widehat{Q}(x, \alpha x) > 0}} \left[\widehat{Q}(x, \alpha x)  - \sum_{\eta: \eta \neq \alpha, \alpha^{-1}} \frac{(q - q_*)(\alpha x, \alpha^{-1}, \eta)}{\widehat{Q}(x, \alpha x) \nu(x)} \right].	
\end{align*}
Besides, definitions of $H_1$ and $H_2$ and the reversibility, it is straightforward to verify that
\begin{align}\label{eq-h1=h2}
	(A) = \left\{(x,\alpha) \mid (x,\alpha) \in H_1 \,\land\, \widehat{Q}(x,\alpha x) > 0\right\} &= \left\{ (x, \alpha) \mid (\alpha x,\alpha) \in H_2 \,\land\, \widehat{Q}(\alpha x, x) > 0   \right\} = (B)\\
	\text{(replace $x$ with $\alpha x$, $\alpha$ is a bijection)}\quad	&= \left\{ (\alpha x, \alpha) \mid (x,\alpha) \in H_2 \,\land\, \widehat{Q}(x,\alpha x) > 0   \right\}\notag
\end{align}
To verify~\eqref{eq-h1=h2}, by reversibility, $\nu(x)\widehat{Q}(x,\alpha x) = \nu(\alpha x)\widehat{Q}(\alpha x, x)$. For $(x,\alpha) \in H_1$, it holds that $\nu(x) > 0$, if  $\widehat{Q}(x,\alpha x) > 0$, then $\nu(\alpha x) > 0$ and $\widehat{Q}(\alpha x, x) > 0$, which implies $(x,\alpha) \in B$, thus $A \subseteq B$. Similarly, for any $(x,\alpha) \in B$, we can verify that $(x,\alpha) \in A$, thus $B \subseteq A$.
Note that $\alpha = \alpha^{-1}$ for all $\alpha \in \+G$.
Hence,
  \begin{align}\label{eq-H1-H2}
  \forall i \in \{1,2\},\quad
   \kappa_i = \min_{\substack{(x, \alpha) \in H_{3-i} \\ \widehat{Q}(x, \alpha x) > 0}} \left[\widehat{Q}(\alpha x, x) - \sum_{\eta: \eta \neq \alpha, \alpha^{-1}} \frac{(q - q_*)(x, \alpha, \eta)}{\widehat{Q}(\alpha x, x) \nu(\alpha x)} \right].
  \end{align}

To levering \Cref{thm:EHMT17} for the tuned Glauber dynamics, we have the following result.
For any two vertices $i,j$, we use $i \sim j$ to denote that $i$ and $j$ are adjacent in $G$.
\begin{lemma} \label{lem:good-regime-aux}
  Let $T, \nu$ in \Cref{thm:EHMT17} be $\widehat{Q}, \nu=(\theta * \pi)^\sigma$, respectively.
  We have the following results.

Let $i,j \in [n] \setminus \Lambda$ such that $i \sim j$ in the graph $G$ and $x \in \Omega$ where $x_i = x_j = \0$, then we have
\begin{align*}
  q(\eta_i x, \eta_i^{-1}, \eta_j) = q(\eta_j x, \eta_i, \eta_j^{-1}) = q(\eta_i \eta_j x, \eta_i^{-1}, \eta_j^{-1}) = \bar{\beta}\bar{\gamma} \cdot q(x, \eta_i, \eta_j).
\end{align*}
Moreover when $i \not\sim j$ in the graph $G$, then for any $x \in \Omega$, it holds that 
\begin{align*}
  q(\eta_i x, \eta_i^{-1}, \eta_j) = q(\eta_j x, \eta_i, \eta_j^{-1}) = q(\eta_i \eta_j x, \eta_i^{-1}, \eta_j^{-1}) = q(x, \eta_i, \eta_j).
\end{align*}
\end{lemma}

The proof of \Cref{lem:good-regime-aux} is deferred to the end of this section.

We claim that for any $i \sim j$, $y \in \Omega$, if $y_i = \1$ or $y_j = \1$, then it holds that 
\begin{align}\label{eq-q*-q=0}
(q - q_*)(y, \eta_i, \eta_j) = 0.	
\end{align}
To verify the claim, we need to consider three cases (1) $y_i = \1$ and $y_j = \0$; (2) $y_i = \1$ and $y_j = \1$; (3) $y_i = \0$ and $y_j = \1$. We verify the first case, the other two cases can be verified by a similar argument. 
Consider the configuration $x = \eta_i y$.
It holds that $x \in \Omega$ because $\theta\bar{\lambda} > 0$ and $\bar{\gamma} > 0$. 
By \Cref{lem:good-regime-aux}, 
\begin{align*}
q_*(y,\eta_i,\eta_j) &= \min\{ q(y,\eta_i,\eta_j), q(\eta_i y, \eta_i^{-1},\eta_j), q(\eta_j y,\eta_i,\eta_j^{-1}), q(\eta_i\eta_j y,\eta_i^{-1}, \eta_j^{-1})\}\\
&=\min\{ q(y,\eta_i,\eta_j), q(\eta_i y, \eta_i,\eta_j), q(\eta_j y,\eta_i,\eta_j), q(\eta_i\eta_j y,\eta_i, \eta_j)\}\\
&=\min\{ q(\eta_i x ,\eta_i,\eta_j), q(x, \eta_i,\eta_j), q(\eta_j \eta_i x,\eta_i,\eta_j), q(\eta_j x,\eta_i, \eta_j)\}\\
(\ast)\quad&= q(\eta_i x,\eta_i,\eta_j) = q(y,\eta_i,\eta_j),
\end{align*}
where $(\ast)$ holds because \Cref{lem:good-regime-aux}, $\eta_i = \eta_i^{-1}$ and $\eta_j = \eta_j^{-1}$.
Besides, for any $i \not \sim j$, any $x \in \Omega$, 
\begin{align}\label{eq-i-notim-j}
(q - q_*)(x, \eta_i, \eta_j) = 0.	
\end{align}


By \eqref{eq-H1-H2},~\eqref{eq-q*-q=0},~\eqref{eq-i-notim-j} and the definitions of $H_1$ and $H_2$, it holds that
\begin{align} \label{eq:kappa-1}
  \kappa_1 = \min_{\substack{(x, \eta_i) \in H_{2} \\ \widehat{Q}(x, \eta_i x) > 0}} \left[\widehat{Q}(\eta_i x, x) - \sum_{j \neq i} \frac{(q - q_*)(x, \eta_i, \eta_j)}{\widehat{Q}(\eta_i x, x) \nu(\eta_i x)} \right] \overset{(\star)}{=} \min_{\substack{(x, \eta_i) \in H_{2} \\ \widehat{Q}(x, \eta_i x) > 0}} \widehat{Q}(\eta_i x, x) \geq 0,
\end{align}
where $(\star)$ holds because $x_i = \1$ and we can use~\eqref{eq-q*-q=0},~\eqref{eq-i-notim-j}. 
On the other hand, we have
\begin{align*}
  \kappa_2
  = \min_{\substack{(x, \eta_i) \in H_{1} \\ \widehat{Q}(x, \eta_i x) > 0}} \left[\widehat{Q}(\eta_i x, x) - \sum_{\substack{j \in [n]\setminus \Lambda \\ j \sim i \\ x_j = \0}} \frac{(1 - \bar{\beta}\bar{\gamma}) q(x, \eta_i, \eta_j)}{\widehat{Q}(\eta_i x, x) \nu(\eta_i x)}\right]
 \overset{(\ast)}{=} \min_{\substack{(x, \eta_i) \in H_{1} \\ \widehat{Q}(x, \eta_i x) > 0}}  \left[\widehat{Q}(\eta_i x, x) - \sum_{\substack{j \in [n]\setminus \Lambda \\ j \sim i \\ x_j = \0}} \frac{(1 - \bar{\beta}\bar{\gamma}) q(x, \eta_i, \eta_j)}{\widehat{Q}(x, \eta_i x) \nu(x)}\right],
\end{align*}
where ($\ast$) holds  by reversibility.
Note that by definition, it holds that $q(x, \eta_i, \eta_j) = \nu(x) \widehat{Q}(x, \eta_i x) \widehat{Q}(x, \eta_j x)$, which implies
\begin{align} \label{eq:kappa-2} 
  \kappa_2 = \min_{(x, \eta_i) \in H_1} \left[\widehat{Q}(\eta_i x, x) - \sum_{\substack{j \in [n]\setminus \Lambda \\ j \sim i \\ x_j = \0}} (1 - \bar{\beta}\bar{\gamma}) \widehat{Q}(x, \eta_j x)\right] 
\overset{(\star)}{\geq} 1 - \Delta \cdot (1 - \bar{\beta}\bar{\gamma}) \cdot \theta\bar{\lambda} \bar{\gamma}^{-\Delta}
   \geq 1/2, 
\end{align}
where $(\star)$ is deduced from the fact that $G$ is $\Delta$-regular or $\bar{\gamma} \leq 1$ and the last inequality holds by \Cref{lem: two-spin -aux} and the fact that $\theta \leq \frac{1}{2} \cdot 12^{-5}$.

Combining \eqref{eq:kappa-1} and \eqref{eq:kappa-2} with \Cref{thm:EHMT17}, it holds that
\begin{align*}
  \rho^{\widehat{\-{GD}}}(\nu) \geq 1/2,
\end{align*}
which, together with \eqref{eq:good-regime-reduce}, implies that
\begin{align*}
  \rho^{\-{GD}}(\nu) &\geq \frac{1}{2n}\rho^{\widehat{\-{GD}}}(\nu) \geq \frac{1}{4n}.
\end{align*}

Finally, we finish the proof by proving \Cref{lem:good-regime-aux}.

\begin{proof}[Proof of \Cref{lem:good-regime-aux}]
  We first consider the case where $i \sim j$ and $x_i = x_j = \0$.
  By definition, we have
\begin{align*}
  q(\eta_i x, \eta_i^{-1}, \eta_j)
  = \nu(\eta_i x) \widehat{Q}(\eta_i x, x) \widehat{Q}(\eta_i x, \eta_i \eta_j x).
\end{align*}
Note that if $\nu(\eta_i x) = 0$, then by the definition of two-spin system, it must hold that $\bar{\beta}= 0$, we have
\begin{align*}
 q(\eta_i x, \eta_i^{-1}, \eta_j) = 	\bar{\beta}\bar{\gamma} \cdot q(x, \eta_i, \eta_j) = 0.
\end{align*}
If $\nu(\eta_i x) \neq 0$, we have
\begin{align*}
  q(\eta_i x, \eta_i^{-1}, \eta_j)  &= \nu(\eta_i x) \cdot 1 \cdot \widehat{Q}(\eta_i x, \eta_i \eta_j x) = \nu(x) \frac{\nu(\eta_ix)}{\nu(x)} \cdot \widehat{Q}(x, \eta_j x) \cdot \bar{\beta}\bar{\gamma} = \nu(x) \widehat{Q}(x, \eta_i x) \widehat{Q}(x, \eta_j x) \cdot \bar{\beta}\bar{\gamma}.
\end{align*}
Similarly, it holds that
\begin{align*}
  q(\eta_j x, \eta_i, \eta_j^{-1}) &= \nu(x) \widehat{Q}(x, \eta_i x) \widehat{Q}(x, \eta_j x) \cdot \bar{\beta}\bar{\gamma}.
\end{align*}
Lastly, we analyze $ q(\eta_i\eta_j x, \eta_i^{-1}, \eta_j^{-1})$. Similarly, we assume $\nu(\eta_i\eta_j x) > 0$, otherwise the result holds trivially. Note that $\bar{\lambda} >0$ and $\bar{\gamma} > 0$, thus we have $\nu(\eta_i x) > 0$ and $\nu(\eta_j x) > 0$.
We have
\begin{align*}
  q(\eta_i\eta_j x, \eta_i^{-1}, \eta_j^{-1})
  &= \nu(\eta_i \eta_j x) \widehat{Q}(\eta_i \eta_j x, \eta_i x) \widehat{Q}(\eta_i \eta_j x, \eta_j x) \\
  &= \nu(\eta_i \eta_j x) 
  = \nu(x) \frac{\nu(\eta_i x)}{\nu(x)} \frac{\nu(\eta_i \eta_j x)}{\nu(\eta_i x)} \\
  &= \nu(x) \widehat{Q}(x, \eta_i x) \widehat{Q}(\eta_i x, \eta_i \eta_j x) \\
  &= \nu(x) \widehat{Q}(x, \eta_i x) \widehat{Q}(x, \eta_j x) \cdot \bar{\beta}\bar{\gamma}. \qedhere
\end{align*}

 We then consider the case where $i \not\sim j$. We prove that $q(\eta_i x, \eta_i^{-1}, \eta_j) = q(x, \eta_i, \eta_j)$. If $\nu(\eta_i x) = 0$, then it holds that $q(\eta_i x, \eta_i^{-1}, \eta_j) = q(x, \eta_i, \eta_j) = 0$. Suppose $\eta_i x$ is a feasible configuration. We have
 \begin{align*}
  q(\eta_i x, \eta_i^{-1}, \eta_j) &= \nu(\eta_i x) \widehat{Q}(\eta_i x, x)\widehat{Q}(\eta_i x, \eta_j\eta_i x)\\
  \text{(by reversibility)}\quad &=	\nu(x) \widehat{Q}(x,\eta_i x) \widehat{Q}(\eta_i x, \eta_j\eta_i x)\\
  \text{($\ast$)}\quad &=	\nu(x) \widehat{Q}(x,\eta_i x) \widehat{Q}(x, \eta_j x)\\
  &= q(x, \eta_i, \eta_j).
 \end{align*}
 where $(\ast)$ holds since $\widehat{Q}(\eta_i x, \eta_j\eta_i x) = \widehat{Q}(x, \eta_j x)$. This is because $i \not \sim j$, both transitions $\eta_i x \to \eta_j\eta_i x$ and $x \to \eta_j x$ are to flip the value of $j$, and such transition probabilities depend only on the configuration of $j$ and $j$'s neighbors. 
 
 The equation $q(\eta_j x, \eta_i, \eta_j^{-1}) =q(x, \eta_i, \eta_j) $ can be proved in a similar way.
 
 Finally, we prove $q(\eta_i \eta_j x, \eta_i^{-1}, \eta_j^{-1}) = q(x, \eta_i, \eta_j)$. 
 Suppose $\nu(\eta_i\eta_j x) = 0$. We have $q(\eta_i \eta_j x, \eta_i^{-1}, \eta_j^{-1}) = 0$. There are three cases for $\eta_i\eta_j x$: (1) if $i$ violates the local constraints, then $\nu(\eta_i x) = 0$; (2) if $j$ violates the local  constraints, then $\nu(\eta_j x) = 0$ (3) if some $k \notin \{i,j\}$ violates the local constraints, then $\nu(x) = 0$. Hence, we have $q(x, \eta_i, \eta_j) = 0$. 
 Similarly, if $\nu(\eta_i x) = 0$ or $\nu(\eta_j x) = 0$, it holds that $q(\eta_i \eta_j x, \eta_i^{-1}, \eta_j^{-1}) = q(x, \eta_i, \eta_j) = 0$.
 Suppose $\nu(\eta_i\eta_j x)\neq 0$,   $\nu(\eta_i x) \neq 0$ and $\nu(\eta_j x) \neq 0$. 
 We have
  \begin{align*}
  q(\eta_i \eta_j x, \eta_i^{-1}, \eta_j^{-1})  &= \nu(\eta_i \eta_j x) \widehat{Q}(\eta_i \eta_j x, \eta_j x)\widehat{Q}(\eta_i \eta_j x, \eta_i x)\\
  \text{(by reversibility)}\quad &=	\nu(\eta_j x) \widehat{Q}(\eta_j x, \eta_i\eta_j x) \widehat{Q}(\eta_i \eta_j x, \eta_i x)\\
  \text{(by $i\not\sim j$)}\quad &=	\nu(\eta_j x) \widehat{Q}(x,\eta_i x) \widehat{Q}(\eta_j x, x)\\
   \text{(by reversibility)}\quad & = \nu(x)\widehat{Q}(x,\eta_i x) \widehat{Q}(x,\eta_j x)\\
   &= q(x,\eta_i,\eta_j).
 \end{align*}
\end{proof}


\bibliographystyle{alpha}
\bibliography{refs}

\clearpage
\appendix
\section{Mixing time from modified log-Sobolev constant}\label{sec:append}
In this section, we prove the mixing time results in \Cref{thm:2spin-theorem} and \Cref{cor:hardcore}. We remark that \Cref{cor:Ising} directly follows from \Cref{thm:2spin-theorem}.

\begin{proof}[Proof of the mixing time result \Cref{thm:2spin-theorem}]
  To prove these corollaries, it only remains to give a lower bound for $\mu_{\min}$. Similar to \cite{chen2021rapid}, the marginal bound $b \triangleq\min_{v \in V} \min_{\sigma \in \Omega(\mu_{V \setminus \{v\} })} \min_{c \in \Omega(\mu^\sigma_v)} \mu^\sigma_v(c)$ can be bounded by
  \begin{align*}
    \frac{1}{b} \le 
    \begin{cases}
      \tp{\lambda+\frac{1}{\lambda}}\tp{\frac{1}{\gamma}+\gamma+2}^\Delta,&\beta = 0;\\
      \tp{\lambda+\frac{1}{\lambda}} \tp{\frac{1}{\beta}+2}^\Delta, &\beta > 0.
    \end{cases}
  \end{align*}
  Therefore,
  \begin{align*}
    \log \log \frac{1}{\mu_{\min}} &\le \log n + \log \log \frac{1}{b} \le \log n + \log \tp{n \log \alpha + \log \tp{\lambda+ \frac{1}{\lambda}}} \\
    &\le 2 \log n + \log \log \alpha + \log \log \tp{\lambda+\frac{1}{\lambda}},  
  \end{align*}
  where $\alpha =     \begin{cases}
    \frac{1}{\gamma}+\gamma+2,&\beta = 0;\\
    \frac{1}{\beta}+2, &\beta > 0.
  \end{cases}$. Together with the first part of \Cref{thm:2spin-theorem}, we prove the mixing time.
\end{proof}

\begin{proof}[Proof of \Cref{cor:hardcore}]
  Recall that the $O(n \log n)$ mixing time can be achieved via standard path coupling technique when $\lambda \le \frac{1}{2\Delta}$. 
  Therefore, we may assume that $\lambda \ge \frac{1}{2\Delta}$.
  To prove this corollary, it only remains to give a lower bound for $\mu_{\min}$. Similar to \cite{chen2021rapid}, the marginal bound 
  \[b \triangleq\min_{v \in V} \min_{\sigma \in \Omega(\mu_{V \setminus \{v\} })} \min_{c \in \Omega(\mu^\sigma_v)} \mu^\sigma_v(c)\]
  can be bounded by
  \begin{align*}
    b \ge \min\tp{\frac{1}{1+\lambda},\frac{\lambda}{1+\lambda}} \ge \frac{1}{2\Delta + 1}.
  \end{align*}
  Therefore,
  \begin{align*}
    \log \log \frac{1}{\mu_{\min}} \le \log n + \log \log \frac{1}{b} \le \log n + \log \log (2\Delta + 1).  
  \end{align*}
  Together with \Cref{thm:2spin-theorem}, we prove this corollary.
\end{proof}

\section{Modified log-Sobolev inequality in sub-critical regime}\label{section-app-Mls}
In this section, we prove \Cref{thm:EHMT17}.
We need several notations and definitions in \cite{EHMT17}.

At first, we will use the following fact about the mapping representation $\+G$.
\begin{fact}\label{fact-mls}
For any function $F: \+X \times \+G \to \mathds{R}$, it holds that 
\begin{align*}
\sum_{x \in \+X, \delta \in \+G}F(x,\delta)T(x,\delta x)\mu(x) = \sum_{x \in \+X, \delta \in \+G}F(\delta x,\delta^{-1})T(x,\delta x)\mu(x)
\end{align*}
\begin{proof}
We have 
\begin{align*}
\sum_{x \in \+X, \delta \in \+G}F(x,\delta)T(x,\delta x)\mu(x) = \sum_{x \in \+X}\mu(x)\sum_{\delta \in \+G}F(x,\delta)T(x,\delta x) =  \sum_{x \in \+X}\mu(x)\sum_{\delta \in \+G}F(x,\delta^{-1})T(x,\delta^{-1}x),
\end{align*}
where the last equation holds because $\{\delta \mid \delta \in \+G\} = \{\delta^{-1}\mid \delta \in \+G\}$. We then have
\begin{align*}
\sum_{x \in \+X}\mu(x)\sum_{\delta \in \+G}F(x,\delta^{-1})T(x,\delta^{-1}x) = 	\sum_{\delta \in \+G, x \in \+X} \mu(x)F(x,\delta^{-1})T(x,\delta^{-1}x) = \sum_{\delta \in \+G, x \in \+X} \mu(\delta x)F(\delta x,\delta^{-1})T(\delta x,x),
\end{align*}
where the equation holds because every $\delta$ is a bijection, thus $\{x \mid x \in \+X\} = \{\delta x \mid x \in \+X\}$. Finally, by reversibility, we have
\begin{align*}
\sum_{\delta \in \+G, x \in \+X} \mu(\delta x)F(\delta x,\delta^{-1})T(\delta x,x) = \sum_{x \in \+X, \delta \in \+G}F(\delta x,\delta^{-1})T(x,\delta x)\mu(x). &\qedhere
\end{align*}
\end{proof}
\end{fact}

Let $\mu$ be a distribution over a finite set $\+X$, and a continuous Markov chain with transition rate $T$ satisfying the detailed-balance equation with respect to $\mu$.
Given $f \in \mathds{R}^{\+X}$, denote $\nabla f(x,y) = f(y) - f(x)$.
For each $\psi \in \mathds{R}^{\+X}$ and $\rho \in \mathds{R}_{>0}^{\+X}$ satisfying $\E[\mu]{\rho} = 1$, 
we define $\+A(\rho,\psi)$ and $\+B(\rho,\psi)$ as follows.
\begin{align}\label{eq:ricci-A}
  \+A(\rho,\psi) = \frac{1}{2} \sum_{x,y \in \+X} \mu(x) T(x,y) \tp{\nabla \psi(x,y)}^2 \hat{\rho}(x,y)
\end{align}
\begin{align}\label{eq:ricci-B}
  \+B(\rho,\psi) = \frac{1}{2} \sum_{x,y \in \+X} \mu(x) T(x,y) \tp{\frac{1}{2} \hat{L}\rho(x,y) \tp{\nabla \psi(x,y)}^2 - \hat{\rho}(x,y) \nabla \psi (x,y) \nabla L\psi(x,y) } 
\end{align}
where 
\begin{align*}
  \hat{\rho}(x,y) &= \theta(\rho(x),\rho(y)),\\
  \hat{L}\rho(x,y) &= \partial_1 \theta(\rho(x),\rho(y)) L\rho(x) +\partial_2 \theta(\rho(x),\rho(y)) L\rho(y),\\
  \text{and } \theta(x,y) &= \begin{cases}
    \frac{x-y}{\log x-\log y},&  x \neq y\\
    x ,& otherwise.
  \end{cases}
\end{align*}

The relation between $\+A(\rho,\psi)$, $\+B(\rho,\psi)$ and modified log-Sobolev inequality was established.
\begin{proposition}[\text{\cite[Lemma 2.3]{EHMT17}}]\label{prop:ricci-mLSI}
 If for any $\rho \in \mathds{R}_{>0}^{\+X}$ satisfying $\E[\mu]{\rho}=1$ and $\psi \in \mathds{R}^{\+X}$ , 
  \begin{align*}
    \+B(\rho,\psi) \ge \kappa \+A(\rho,\psi), 
  \end{align*}
  for some $\kappa \in (0,1)$,
  then the  modified log-Sobolev constant is at least $2\kappa$. 
\end{proposition}

Let $\+G$ be a group acting on $\+X$ such that for each $x,y \in \+X$ with transition rate $T(x,y) > 0$, there exists a unique $\delta \in \+G$ satisfying $y = \delta x$.
For each $\rho \in \mathds{R}^{\+X}_{>0}$ satisfying $\E[\mu]{\rho} = 1$ and $\psi \in \mathds{R}^{\+X}$, we may rephrase $\+A(\rho,\psi)$ and $\+B(\rho,\psi)$ as
\begin{align}
  \+A(\rho,\psi) &= \frac{1}{2} \sum_{x \in \+X,\delta \in \+G} \mu(x) T(x,\delta x) \hat{\rho}(x,\delta x) \tp{\nabla_\delta \psi(x)}^2,\label{eq-A-rewrt}\\
  \+B(\rho,\psi) &= \sum_{x \in \+X,\delta, \eta \in \+G} \mu(x) T(x,\delta x) T(x,\eta x) B(\rho,\psi)(x,\delta,\eta),\label{eq-B-rewrt}\\
  B(\rho,\psi)(x,\delta,\eta) 
  &=\frac{1}{2} \tp{\nabla_\delta \psi(x)}^2 \hat{\rho}_1(x,\delta x) \nabla_\eta \rho(x) + \nabla_\delta \psi(x) \nabla_\eta \psi(x) \hat{\rho}(x,\delta x),\notag
\end{align} 
where $\nabla_\delta \psi(x) = \psi(\delta x) - \psi(x)$, $\nabla_\rho \psi(x) = \rho(\delta x) - \rho(x)$ and $\hat{\rho}_i(x,y) =\frac{\partial\theta}{\partial z_i}\bigg|_{(z_1,z_2) = (\rho(x),\rho(y))}$ for $i=1,2$. 

To verify~\eqref{eq-B-rewrt}, by the definition of $\+B$, we have
\begin{align}\label{eq-appendix-cal-B}
\+B(\rho,\psi) =  \sum_{x,y \in \+X} \mu(x) T(x,y) \tp{\frac{1}{4} \hat{L}\rho(x,y) \tp{\nabla \psi(x,y)}^2 - \frac{1}{2}\hat{\rho}(x,y) \nabla \psi (x,y) \nabla L\psi(x,y) }.	
\end{align}
We have
\begin{align*}
&\sum_{x,y \in \+X} \mu(x) T(x,y)\hat{L}\rho(x,y) \tp{\nabla \psi(x,y)}^2 = \sum_{x \in \+X,\delta \in \+G}	\mu(x)T(x,\delta x) \tp{\nabla_\delta \psi(x)}^2  (\hat{\rho}_1(x,\delta x)L\rho(x) + \hat{\rho}_2(x,\delta x)L\rho(\delta x) )\\
=\,& \sum_{x \in \+X,\delta \in \+G}	\mu(x)T(x,\delta x) \tp{\nabla_\delta \psi(x)}^2 \hat{\rho}_1(x,\delta x)L\rho(x) + \sum_{x \in \+X,\delta \in \+G}	\mu(x)T(x,\delta x) \tp{\nabla_\delta \psi(x)}^2\hat{\rho}_2(x,\delta x)L\rho(\delta x)\\
\overset{(\ast)}{=}\,& \sum_{x \in \+X,\delta \in \+G}	\mu(x)T(x,\delta x) \tp{\nabla_\delta \psi(x)}^2 \hat{\rho}_1(x,\delta x)L\rho(x) + \sum_{x \in \+X,\delta \in \+G}	\mu(x)T(x,\delta x) \tp{\nabla_{\delta^{-1}} \psi(\delta x)}^2\hat{\rho}_2(\delta x,x)L\rho(x)\\
\overset{(\star)}{=}\,&2\sum_{x \in \+X,\delta \in \+G}	\mu(x)T(x,\delta x) \tp{\nabla_\delta \psi(x)}^2 \hat{\rho}_1(x,\delta x)L\rho(x)\\
=\,&2\sum_{x \in \+X,\delta,\eta \in \+G}	\mu(x)T(x,\delta x)T(x,\eta x) \tp{\nabla_\delta \psi(x)}^2 \hat{\rho}_1(x,\delta x)\nabla_{\eta}\rho(x). \quad\text{(by definition of $L\rho$)},
\end{align*}
where $(\ast)$ holds because of \Cref{fact-mls}, $(\star)$ holds because $\nabla_{\delta^{-1}} \psi(\delta x) = -\nabla_\delta \psi(x)$ and $\hat{\rho}_2(\delta x,x) =  \hat{\rho}_1(x,\delta x)$. Besides, we have
\begin{align*}
 &\sum_{x,y \in \+X} \mu(x) T(x,y) 	\hat{\rho}(x,y) \nabla \psi (x,y) \nabla L\psi(x,y)\\ 
 =\,& \sum_{x \in \+X ,\delta\in \+G} \mu(x) T(x,\delta x) \hat{\rho}(x,\delta x)\nabla \psi_\delta (x) \tp{\sum_{\eta \in \+G}T(\delta x, \eta \delta x)\nabla_{\eta}\psi(\delta x) - \sum_{\eta \in \+G}T(x,\eta x)\nabla_{\eta}\psi(x)},
\end{align*}
Again, by \Cref{fact-mls}, we have
\begin{align*}
&\sum_{x \in \+X ,\delta\in \+G} \mu(x) T(x,\delta x) \hat{\rho}(x,\delta x)\nabla \psi_\delta (x) \sum_{\eta \in \+G}T(\delta x, \eta \delta x)\nabla_{\eta}\psi(\delta x)\\
=\, & \sum_{x \in \+X ,\delta\in \+G} \mu(x) T(x,\delta x) 	\hat{\rho}(\delta x,x)\nabla \psi_{\delta^{-1}} (\delta x)\sum_{\eta \in \+G}T(x, \eta x)\nabla_{\eta}\psi(x)\\
=\,& - \sum_{x \in \+X ,\delta\in \+G} \mu(x) T(x,\delta x) 	\hat{\rho}(\delta x,x)\nabla \psi_{\delta} (x)\sum_{\eta \in \+G}T(x, \eta x)\nabla_{\eta}\psi(x) \quad \text{(by $\nabla_{\delta^{-1}} \psi(\delta x) = -\nabla_\delta \psi(x)$)}\\
=\,& - \sum_{x \in \+X ,\delta\in \+G} \mu(x) T(x,\delta x) 	\hat{\rho}(x,\delta x)\nabla \psi_{\delta} (x)\sum_{\eta \in \+G}T(x, \eta x)\nabla_{\eta}\psi(x), \quad \text{(by $	\hat{\rho}(\delta x, x) = 	\hat{\rho}(x,\delta x)$)}
\end{align*}
which implies 
\begin{align*}
\sum_{x,y \in \+X} \mu(x) T(x,y) 	\hat{\rho}(x,y) \nabla \psi (x,y) \nabla L\psi(x,y) = -2	\sum_{x \in \+X ,\delta\in \+G} \mu(x) T(x,\delta x) 	\hat{\rho}(x,\delta x)\nabla \psi_{\delta} (x)\sum_{\eta \in \+G}T(x, \eta x)\nabla_{\eta}\psi(x).
\end{align*}
By~\eqref{eq-appendix-cal-B}, we have
\begin{align*}
\+B(\rho,\psi) = \sum_{x \in \+X,\delta, \eta \in \+G} \mu(x) T(x,\delta x) T(x,\eta x) B(\rho,\psi)(x,\delta,\eta).
\end{align*}

\begin{lemma}[Lemma 3.6, \cite{EHMT17}]\label{lemma:ricci-i}
 Let $\mu$ be a distribution over a finite set $\+X$.
  Let $T$ be the transition rate of a continuous-time Markov chain $M$ satisfying the detailed-balance equation with respect to $\mu$.
  Let $\+G$ a mapping  representation of $T$. 
  If $H$ be a subset of $\+X \times \+G$ that satisfies $H \cup H^{-1} = \+X \times \+G$, where $H^{-1} = \{(\delta x,\delta^{-1}) \mid (x,\delta) \in H\} $, then for any $\rho \in \mathds{R}_{>0}^{\+X}$ satisfying $\E[\mu]{\rho} = 1$, and $\psi \in \mathds{R}^{\+X}$,
  \begin{align*}
    \sum_{(x,\delta) \in H} \mu(x) T(x,\delta x) B(\rho,\psi)(x,\delta,\delta) \ge \frac{1}{2}\+A(\rho,\psi). 
  \end{align*}
\end{lemma}

\begin{proof}[Proof of \Cref{lemma:ricci-i}]  
  Note that for any $x,y \in \mathds{R}_{>0}$,
  \begin{align}\label{eq-rho-rho-hat}
    x \partial_1 \theta(x,y) + y \partial_2 \theta(x,y) = \theta(x,y).
  \end{align}
  The above equation is in~\cite[Lemma 3.5]{EHMT17}.
  Therefore, for any $x \in \+X$ and $\delta \in \+G$, we have
  \begin{align*}
    B(\rho,\psi)(x,\delta, \delta) &= \frac{1}{2}\tp{\nabla_\delta \psi (x)}^2 \tp{ \hat{\rho}_1(x,\delta x) \tp{\rho(\delta x)-\rho(x)} + 2\hat{\rho}(x,\delta x)}\\
   (\text{by~\eqref{eq-rho-rho-hat}})\quad &= \frac{1}{2} \tp{\nabla_\delta \psi(x)}^2 \tp{ \hat{\rho}_1(x,\delta x) \rho(\delta x)+\hat{\rho}_2(x,\delta x) \rho(\delta x) + \hat{\rho}(x,\delta x)}\\
    &\ge \frac{1}{2} \tp{\nabla_\delta \psi(x)}^2 \hat{\rho}(x,\delta x),
  \end{align*}
  where the last inequality follows from the fact that $\partial_1 \theta(x,y) + \partial_2 \theta(x,y) = 
  \begin{cases}
    \frac{(x-y)^2}{xy(\log x-\log y)^2},& x \neq y;\\
    1,& x = y.
  \end{cases}$, which is non-negative.
  Furthermore, by reversibility, $\tp{\nabla_\delta \psi(x)}^2 = \tp{\nabla_{\delta^{-1}}\psi(\delta x)}^2$ and $\hat{\rho}(x,\delta x) = \hat{\rho}(\delta x,x)$, 
  \begin{align*}
    \sum_{(x,\delta) \in H} \mu(x) T(x,\delta x) \tp{\nabla_\delta \psi(x)}^2 \hat{\rho}(x,\delta x)
    &= \sum_{(x,\delta) \in H} \mu(\delta x) T(\delta x,\delta^{-1}(\delta x)) \tp{\nabla_{\delta^{-1}} \psi(\delta x)}^2 \hat{\rho}(\delta x,\delta^{-1}(\delta x))\\
    \text{(by the definition of $H^{-1}$)}\quad&= \sum_{(x,\delta) \in H^{-1}} \mu(x) T(x,\delta x) \tp{\nabla_\delta \psi(x)}^2 \hat{\rho}(x,\delta x).
  \end{align*}
  Hence,
  \begin{align*}
    \sum_{(x,\delta) \in H} \mu(x) T(x,\delta x) B(\rho,\psi)(x,\delta,\delta) &\ge \frac{1}{2}\sum_{(x,\delta) \in H} \mu(x) T(x,\delta x) \tp{\nabla_\delta \psi(x)}^2 \hat{\rho}(x,\delta x)\\
    &\ge \frac{1}{4} \sum_{(x,\delta) \in H \cup H^{-1}} \mu(x) T(x,\delta x) \tp{\nabla_\delta \psi(x)}^2 \hat{\rho}(x,\delta x)\\
    &= \frac{1}{4} \sum_{x \in \+X, \delta \in \+G} \mu(x) T(x,\delta x) \tp{\nabla_\delta \psi(x)}^2 \hat{\rho}(x,\delta x)\\
   \text{(by~\eqref{eq-A-rewrt})}\quad &= \frac{1}{2} \+A(\rho,\psi). &\qedhere
  \end{align*}
\end{proof}

We now prove \Cref{thm:EHMT17}
\begin{proof}[Proof of \Cref{thm:EHMT17}]
For the mapping $G$ satisfying the condition in the theorem, the following inequality was proved in \cite{EHMT17} (see proof of Theorem~3.9 in~\cite{EHMT17}):
  \begin{align}\label{eq:ricci-bq}
    \+B(\rho,\psi) \ge \sum_{(x,\delta) \in \+X \times \+G} B(x,\delta,\delta) \tp{q(x,\delta,\delta) - \*1_{\delta \neq \delta^{-1}} q(\delta x,\delta^{-1},\delta) -\sum_{\eta:\eta \neq \delta, \delta^{-1}} (q-q_*)(\delta x,\delta^{-1},\eta)}.
  \end{align}  
  By the definition of $q$, reversibility and non-negativity, we have
  \begin{align*}
    \+B(\rho,\psi) &\ge \sum_{(x,\delta) \in \+X \times \+G} \mu(x) T(x,\delta x) B(x,\delta,\delta) \tp{T(x,\delta x) - \*1_{\delta \neq \delta^{-1}} T(\delta x,\delta(\delta x)) -\sum_{\eta:\eta \neq \delta, \delta^{-1}} \frac{(q-q_*)(\delta x,\delta^{-1},\eta)}{\mu(x) T(x,\delta x)} }\\
    &\overset{(\ast)}{=} \sum_{(x,\delta) \in H_1\cup H_2} \mu(x) T(x,\delta x) B(x,\delta,\delta) \tp{T(x,\delta x) - \*1_{\delta \neq \delta^{-1}} T(\delta x,\delta(\delta x)) -\sum_{\eta:\eta \neq \delta, \delta^{-1}} \frac{(q-q_*)(\delta x,\delta^{-1},\eta)}{\mu(x) T(x,\delta x)} }\\
    &\overset{(\star)}{\ge} \kappa_1 \sum_{(x,\delta) \in H_1} \mu(x) T(x,\delta x) B(x,\delta,\delta) + \kappa_2 \sum_{(x,\delta) \in H_2} \mu(x) T(x,\delta x) B(x,\delta,\delta),
  \end{align*}
  where $(\ast)$ holds because $H_1 \cup H_2 = \+X \times \+G$, and $(\star)$ holds because $H_1 \cap H_2 = \emptyset$.
  We now use \Cref{lemma:ricci-i}. Note that $H_2 = H_1^{-1}$ and $H_1 = H_2^{-1}$ and $H_1 \uplus H_2 = \+X \times \+G$. We have
  \begin{align*}
  	 \+B(\rho,\psi) \geq \frac{\kappa_1 + \kappa_2}{2} \+A(\rho,\psi). 
  \end{align*}
  By \Cref{prop:ricci-mLSI}, we conclude the proof.
\end{proof}

\section{Monotonicity of uniqueness condition}\label{section-app-monotone}
In this section, we prove a stronger version of \Cref{prop:LLY}. 
\begin{proposition}[\cite{LLY13}]\label{prop:equiv-LLY}
  Let $\beta,\gamma,\lambda$ be parameters of an anti-ferromagnetic system, and $\hat{x}_d \in \mathds{R}_{>0}$ be the unique fixed point of recursion $F_d(x) = \lambda \tp{\frac{\beta x+1}{x+\gamma}}^d$ for any $d \in \mathds{Z}_{>0}$. 
  The following statements are equivalent.
  \begin{enumerate}
    \item $\gamma \le 1$;
    \item $\abs{F'_d(\hat{x}_d)}$ is monotone increasing in $d$.
  \end{enumerate}
\end{proposition}
\begin{proof}
When $\gamma > 1$, the fixed point $\hat{x}_d$ satisfies
\begin{align*}
  \hat{x}_d = \lambda \tp{\frac{\beta \hat{x}_d + 1}{\hat{x}_d + \gamma}}^d \le \frac{\lambda}{\gamma^d}.
\end{align*}
Hence,
\begin{align*}
  0 \le \lim_{d \to +\infty}\abs{F'_d(\hat{x}_d)} = \lim_{d \to +\infty} \frac{d(1-\beta \gamma)\hat{x}_d}{(\beta \hat{x}_d + 1)(\hat{x}_d+\gamma)} \le \lim_{d \to +\infty} \frac{d \lambda}{\gamma^d} = 0.
\end{align*}
Note that $F'_1(\hat{x}_1)>0$. Therefore, $F'_d(\hat{x}_d)$ must not monotone increase for all $d$.

When $\gamma \le 1$,
define $c(d): [1,\infty) \to \mathds{R}$ as
\begin{align*}
	c(d) = \frac{d(1-\beta\gamma)\hat{x}_d}{(\beta \hat{x}_d + 1)(\hat{x}_d+\gamma)} =\frac{d(1-\beta\gamma)\hat{x}_d}{p(\hat{x}_d)}, \quad \text{where } p(x) = (\beta x + 1)(x+\gamma).
\end{align*} 
Note that $c(d) = \abs{F'_d(\hat{x}_d)}$. Hence, it suffices to show that 
\begin{align*}
\gamma \geq 1 \implies \forall d > 0, c'(d) > 0.
\end{align*}
Let $q(x) = \frac{\beta x + 1}{x + \gamma}$. Take the derivative of $c(d)$ with respect to $d$, we have
\begin{align*}
  c'(d) 
  = \frac{(1-\beta\gamma)\hat{x}_d}{p(\hat{x}_d)}
  \left(1-d\ln q(\hat{x}_d)\cdot
  \frac{\beta \hat{x}_d^2-\gamma}{p(\hat{x}_d)+d(1-\beta\gamma)\hat{x}_d}\right).
\end{align*}

Since $ \frac{(1-\beta\gamma)\hat{x}_d}{p(\hat{x}_d)} > 0$, we only need to verfiy that 
\begin{align*}
d\ln q(\hat{x}_d)\cdot
  \frac{\beta \hat{x}_d^2-\gamma}{p(\hat{x}_d)+d(1-\beta\gamma)\hat{x}_d} \leq 1.	
\end{align*}
Let $x = \hat{x}_d$. Note that $d \geq 1$. It suffices to show that 
\begin{align*}
	\forall x > 0, \quad (\beta x^2 - \gamma)\ln\tp{ \frac{\beta x + 1}{x + \gamma}} \leq (1-\beta \gamma)x.
\end{align*}
Note that $\beta x^2 - \gamma > 0$ if and only if $x \geq \sqrt{\gamma/\beta}$; $\ln\tp{ \frac{\beta x + 1}{x + \gamma}} \geq 0$ if and only if $x \leq \frac{1-\gamma}{1-\beta}$. Since $0\leq \beta\leq \gamma\leq 1$ and $\beta\gamma < 1$, we have
\begin{align*}
	\sqrt{\frac{\gamma}{\beta}} \geq \frac{1-\gamma}{1-\beta}.
\end{align*}
Note that $(1-\beta \gamma)x > 0$, we only need to prove that 
\begin{align*}
	\forall \frac{1-\gamma}{1-\beta} < x < \sqrt{\frac{\gamma}{\beta}}, \quad \ln \tp{ \frac{x+\gamma}{\beta x + 1}} \leq \frac{(1-\beta \gamma)x}{\gamma - \beta x^2}
\end{align*}
Note that $\frac{x+\gamma}{\beta x + 1} = 1 + \frac{(1-\beta)x +(\gamma -1)}{\beta x + 1} \leq \exp\tp{\frac{(1-\beta)x +(\gamma -1)}{\beta x + 1}}$, it suffices to show that 
\begin{align*}
\forall \frac{1-\gamma}{1-\beta} < x < \sqrt{\frac{\gamma}{\beta}}, \quad 	\frac{(1-\beta)x +(\gamma -1)}{\beta x + 1} \leq \frac{(1-\beta \gamma)x}{\gamma - \beta x^2}.
\end{align*}
Note that $\gamma - \beta x^2 > 0$ if $x < \sqrt{\frac{\gamma}{\beta}}$. The above inequality is equivalent to for all $\frac{1-\gamma}{1-\beta} < x < \sqrt{\frac{\gamma}{\beta}}$,
\begin{align*}
\beta(\beta-1)x^3 + \beta\gamma(\beta - 1)x^2 + (\gamma - 1)x + \gamma(\gamma - 1) \leq 0.
\end{align*}
The above inequality holds because if $\gamma \leq 1$, then $\beta(\beta-1) < 0$, $\beta\gamma(\beta - 1) < 0$, $\gamma - 1<0$ and $\gamma(\gamma - 1)<0$. 
\end{proof}

\section{Boundedness of anti-ferromagnetic two-spin system}\label{sec:append-D}
\begin{proof}[Proof of \Cref{lem:boundedness}]
Note that the first part directly follows from \Cref{prop:LLY} and Lemma 36 in \cite{chen2020rapid}. 
Therefore, we will only focus on the case where $G$ is $(\Delta-1)$-regular.

First, we prove this for the case where $\sqrt{\beta\gamma} \geq \frac{\Delta - 2}{\Delta}$, for all $y \in [-\infty, +\infty]$, it holds that
\begin{align*}
  \abs{h(y)}
  = \frac{(1-\beta\gamma) \e^y}{(\beta \e^y+1)(\e^y + \gamma)}
  = \frac{1-\beta\gamma}{\beta \e^y + \gamma \e^{-y} + 1 + \beta\gamma}
  \leq \frac{1 - \beta\gamma}{1 + \beta\gamma + 2\sqrt{\beta\gamma}}
  = \frac{1 - \sqrt{\beta\gamma}}{1 + \sqrt{\beta\gamma}}
  \leq \frac{1}{\Delta - 1}
  \leq \frac{1.5}{\Delta},
\end{align*}
where in the last inequality, we use the fact that $\frac{2}{3}\Delta \leq (\Delta - 1)$.

Now, we only left consider the case where $\sqrt{\beta\gamma} < \frac{\Delta - 2}{\Delta}$.
When $\lambda < \lambda_{1}(\Delta-1)$, for any $y \in J_{\lambda,\Delta - 1}$, $\e^y \le \frac{\lambda}{\gamma^{\Delta - 1}} \leq \frac{18 \gamma}{\theta(\Delta-1)}$, where the last inequality follows from \Cref{lem: two-spin -aux-aux} and $\theta(d) \triangleq d(1 - \beta\gamma) - (1 + \beta\gamma)$. Therefore, $\abs{h(y)}$ can be bounded as follows
\begin{align*}
  \abs{h(y)} = \frac{(1-\beta\gamma) \e^y}{(\beta \e^y+1)(\e^y + \gamma)} \le \frac{1-\beta\gamma}{\gamma \e^{-y}+1+\beta\gamma} \leq \frac{1-\beta\gamma}{\frac{\theta(\Delta-1)}{18} + 1+ \beta\gamma} = \frac{18(1-\beta\gamma)}{\Delta(1-\beta\gamma)+18 \beta\gamma+16} \le \frac{18}{\Delta}.
\end{align*}

Similarly, when $\lambda > \lambda_{2}(\Delta-1)$, for any $y \in J_{\lambda,\Delta - 1}$, $\e^y \ge \lambda \beta^{\Delta - 1} \geq \frac{\theta(\Delta-1)}{18 \beta}$, where the last inequality follows from \Cref{lem: two-spin -aux-aux}. Therefore,
\begin{align*}
  \abs{h(y)} = \frac{(1-\beta\gamma) \e^y}{(\beta \e^y+1)(\e^y + \gamma)} \le \frac{1-\beta\gamma}{\beta \e^y + 1 +\beta\gamma} \leq  \frac{1-\beta\gamma}{\frac{\theta(\Delta-1)}{18} + 1+ \beta\gamma} = \frac{18(1-\beta\gamma)}{\Delta(1-\beta\gamma)+18 \beta\gamma+16} \le \frac{18}{\Delta}.
\end{align*}
This concludes the proof of \Cref{lem:boundedness}.
\end{proof}

\end{document}